\documentclass[sigconf,authorversion]{acmart}

\AtBeginDocument{%
  }

\copyrightyear{2023}
\acmYear{2023}
\setcopyright{acmlicensed}\acmConference[STOC '23]{Proceedings of the 55th
Annual ACM Symposium on Theory of Computing}{June 20--23, 2023}{Orlando,
FL, USA}
\acmBooktitle{Proceedings of the 55th Annual ACM Symposium on Theory of
Computing (STOC '23), June 20--23, 2023, Orlando, FL, USA}
\acmPrice{15.00}
\acmDOI{10.1145/3564246.3585113}
\acmISBN{978-1-4503-9913-5/23/06}

\acmSubmissionID{stoc23main-p56-p}
\received{2022-11-07}
\received[accepted]{2023-02-06}

\usepackage{hyperref}
\usepackage{color}
\usepackage{microtype}

\makeatletter
\newcommand{\otherlabel}[2]{\protected@edef\@currentlabel{#2}\label{#1}}
\makeatother

\definecolor{darkgreen}{RGB}{0,100,0}
\definecolor{firebrick}{RGB}{178,34,34}

\DeclareMathOperator{\ax}{{ax}}

\DeclareMathOperator{\radius}{radius}
\DeclareMathOperator{\cent}{center}
\DeclareMathOperator{\F}{\mathcal{F}}
\DeclareMathOperator{\wfs}{wfs}
\DeclareMathOperator{\length}{length}
\DeclareMathOperator{\B}{\mathbb{B}}
\DeclareMathOperator{\vol}{Vol}

\newcommand{\ud}{\mathrm{d}}
\newcommand\R{\mathbb{R}}
\newcommand\Q{\mathbb{Q}}

\newcommand{\Nat}{\mathbb{N}}

\newcommand{\M}{\mathcal{M}}

\newcommand{\Image}{\mathrm{im}}

\newcommand{\Gdiam}{{\operatorname{GeoDiameter}}}

\newcommand{\defunder}[1]{\underset{\text{def.}}{#1} \:}
\newcommand{\PushedPath}[2]{{#1}^{\left[#2\right]_K}}

\AtEndPreamble{%
\theoremstyle{acmdefinition}
\newtheorem{assumption}[theorem]{Assumption}
\newtheorem{remark}[theorem]{Remark}}

\begin{document}

\title{ Hausdorff and Gromov-Hausdorff stable subsets of the medial axis}

\author{Andr{\'e} Lieutier}
\email{andre.lieutier@gmail.com}
\affiliation{%
  \institution{No affiliation}
	\city{Aix-en-Provence}
	\country{France}
}

\author{Mathijs Wintraecken}
\email{mathijs.wintraecken@inria.fr}
\affiliation{%
  \institution{Inria Sophia-Antipolis, Universit{\'e} C{\^o}te d'Azur}
  \streetaddress{2004 Route des Lucioles}
  \city{Valbonne}
  \country{France}
	\\
	\institution{Institute of Science and Technology Austria}
  \streetaddress{Am Campus 1}
	\city{Klosterneuburg}
	\country{Austria} 
		}


\begin{abstract}
In this paper we introduce a pruning of the medial axis called the $(\lambda,\alpha)$-medial axis ($\ax_\lambda^\alpha $). We prove that the $(\lambda,\alpha)$-medial axis of a set $K$ is stable in a Gromov-Hausdorff sense under weak assumptions. More formally we prove that if $K$ and $K'$ are close in the Hausdorff ($d_H$) sense then the  $(\lambda,\alpha)$-medial axes of $K$ and $K'$ are close as metric spaces, that is the Gromov-Hausdorff distance ($d_{GH}$) between the two is $\frac{1}{4}$-H{\"o}lder in the sense that $d_{GH} (\ax_\lambda^\alpha (K),\ax_\lambda^\alpha (K')) \lesssim d_H(K,K')^{1/4}$. 
The Hausdorff distance between the two medial axes is also bounded, by
$d_{H} (\ax_\lambda^\alpha (K),\ax_\lambda^\alpha (K')) \lesssim d_H(K,K')^{1/2}$. 
These quantified stability results provide guarantees for practical computations of medial axes from approximations.
Moreover, they provide key ingredients for  studying the computability of the medial axis in the context of computable analysis.
\end{abstract}

\begin{CCSXML}
<ccs2012>
<concept>
<concept_id>10003752.10010061.10010063</concept_id>
<concept_desc>Theory of computation~Computational geometry</concept_desc>
<concept_significance>500</concept_significance>
</concept>
</ccs2012>
\end{CCSXML}

\ccsdesc[500]{Theory of computation~Computational geometry}

\keywords{Medial axis, Gromov-Hausdorff distance, metric stability, computable analysis}
\maketitle

\textit{This is the full version of the paper accepted at STOC'23. Part I is almost identical to the STOC paper, while Part II is not contained in the conference paper.}

\section*{\textit{Part I: Introduction, motivation and non-technical overview of the results}}

\section{Introduction}\label{section:introduction}
Given a closed subset $K$ of Euclidean space $\R^n$, its {\it medial axis}, denoted $\ax (K)$, is the set of points in the complement $K^c$ of $K$ for which there are at least two closest points in $K$, or, equivalently, on its boundary $\partial K$. Note that the definitions of the medial axis used in preceding papers on the same topic %of the same line of work 
\cite{LIEUTIERhomotopytype,cl2005lambda}
 considered an open subset $\mathcal{O} \subset \R^n$ instead. Because the medial axis was then defined as the set of points in $\mathcal{O}$
  with at least two closest points on the complement $\mathcal{O}^c$ we see that the difference is only cosmetic by setting $K = \mathcal{O}^c$.
The properties of the medial axis and its computation have been intensively studied, both in theory and in particular applications contexts, 
see \cite{Attali2009} for an overview, or \cite{tagliasacchi20163d}\footnote{%which,
 Unfortunately, \cite{tagliasacchi20163d} mixes up the $\theta$-medial and $\lambda$-medial axis in Figure $11$ of that paper.} %a wrong statement on figure  
for an application oriented review of  general notions of shapes skeletons and computation methods.

One obvious motivation for studying the stability of the medial axis is to be able to guarantee the (approximate) correctness of the information that can be extracted from the medial $\ax(K')$ axis of an approximated shape $K'$ of an exact, or ideal, shape $K$. Here the approximation error could be the unavoidable finite accuracy of physical measurements or some small perturbations induced by rounding or geometric data conversions. % where the approximation error $d(K',K)$, for some distance $d$, could be just the unavoidable finite accuracy of physical measurements or some small perturbations induced by rounding or geometric data conversions. Once $\ax(K')$ is computed, what reliable information can be expect to be preserved  between the unknown $\ax(K)$ and  $\ax(K')$.

Another, more formal,  motivation for studying the stability is related to its formal computation. Indeed, among the significant amount of practical proposed algorithms for the map $K \mapsto \ax(K)$, the model of computation is usually implicit, which we find problematic in the case of this particularly unstable object. {We refer to Section \ref{section:introModelOfComputation} for a more extensive discussion of this issue.} 

The idea of pruning, or filtering, the medial axis, in order to improve its stability, has been, sometime implicitly, 
a key ingredient  in  realistic algorithms.
For example, in \cite{foskey2003efficient}, the {\it $\theta$-simplified medial axis} of $K$ is defined as the set of points $x$ on the medial axis of $K$ {for which $x$}
has at least $2$ closest points  $p,q \in K$ such that the angle $\angle pxq$ is greater than $\theta$.
Since the medial axis of a finite discrete set $S \subset \R^d$ is the $(d-1)$-skeleton of the Voronoi diagram of $S$, following some pioneering works such as \cite{attali1997computing,amenta1998crust},
in \cite{dey2004approximating}, the Voronoi cells  of a point sample are pruned along some parametrized criterion, namely a  {\it angle condition}
 or a {\it ratio condition} %{\color{red} [MW one of the reviewers complained here] 
on the circumradius of the set of closest points and the distance between the point on the medial axis and its closest projection. %}

This paper pursues the quest for provably stable filtrations of the medial axis for general closed subset of Euclidean space, in the spirit of \cite{cl2005lambda}.
Other prunings of the medial axis have been suggested in e.g. \cite{attali1996modeling, dey2006defining, blanc2018salience, Erin2016, LIU20111496, Burning2015, SHAKED1998156, giesen2009scale}. 
%{\color{red} 
%[MW not sure about this, the reviewer complained about this and this is way to do something about it without wasting too much space] 
Each pruning method comes with some drawbacks (as well as strong points). We refer to \cite{Burning} for a discussion of the particular deficiencies of a number of these methods in more detail.
%In a related paper and video \cite{ , Burning} [TODO add journal version of the Bing's house paper] we discuss the particular deficiencies of a number of these methods in more detail, moreover we exhibit the instability of the burning method explicitly.
%}  

\section{The critical function and the \texorpdfstring{$\lambda$}{lambda}-medial axis}\label{section:introCriticalLambdaMedialAxis}

In this section we review a number of results from \cite{chazal2009sampling,cl2005lambda} on the critical function of a compact set and some related notions. These results are both key ingredients in our proofs and a source of inspiration for some of the statements.
The {\it reach} of a set $K$ is the minimal distance between a set and its medial axis. It was introduced by Federer \cite{Federer}
in order to extend curvature measures to more general sets.
%define the class of sets with positive reach, for which generalized curvature measures can be defined.
The reach is also the lowest upper bound over the set of the {\it local feature size} { \cite{Federer, amenta2001power}}, that is the distance of a point to the medial axis.\footnote{ The nomenclature was introduced by Amenta et al. \cite{amenta2001power}
in order to state conditions under which the topology of a set can be determined from a sampling of it, however the concept was known to Federer \cite{Federer}.}
The {\it critical function} $\chi_K : (0, \infty)\rightarrow [0,1]$ (\cite{chazal2009sampling}  and \cite[Section 9]{boissonnat2018geometric})
 of a compact set $K$ has been introduced in order to quantify how the topology of a set can be determined from a Hausdorff approximation of it,
 in particular when the reach is $0$, which is common for non-smooth sets. 
 
For a point $x\in \R^n$, we denote by  $R_K (x)$ its distance to $K$ and by  $\F_K (x)$ the radius of the smallest ball enclosing the points in $K$ closest to $x$, see Section \ref{section:MedialAxisAndFlow} for details.
The critical function $\chi_K$ of $K$  is then defined as:
\begin{equation} \label{equation:CriticalFunctionFirstDefinition}
 \chi_K (t) \defunder{=}  \inf_{R_K(x) = t} \sqrt{ 1 - \left( \frac{\F_K (x) }{ R_K (x)}\right)^2}.
\end{equation}
The medial axis $\ax(K)$  can %{\color{red} \sout{equivalently} Erin complained here, I personally feel that it is clear but the word does not add much. Alternatively: One of the (many) definitions of the medial axis $\ax(K)$ is the set of points $x$ in $\R^n$ such that $\F_K (x)>0$.} 
be defined as the set of points $x$ in $\R^n$ such that $\F_K (x)>0$. It follows that, 
when $K$ has positive reach, $\chi_K (t)=1$ for $t$ smaller than the reach. 

We write \[
K^{\oplus t} \defunder{=} \{x \in \R^n, d(x, K)  \leq t\}
\] 
for the $t$-offset of $K$.
For $t>0$, the topology of this offset can only change at critical values of the distance function, that is values for which $\chi_K$ vanishes. 
{ For a given $\mu \in (0,1]$,} the $\mu$-Reach ($r_\mu$) is defined as 
\begin{align}
r_\mu (K) &\defunder{=} \inf \{t  \mid \chi_K (t) < \mu \}. 
\nonumber
\end{align} 
If $K$ has positive $\mu$-reach for some $\mu>0$, then $K^{\oplus r_\mu}$  deforms retract on $K$, see \cite[Theorem 12]{kim2019homotopy}.
{Notices that $r_1(K)$ is the reach of $K$.}

%Moreover, if,  for some $r_\mu>0$ and $\mu>0$, one has $t\in(0, r_\mu) \Rightarrow \chi_K (t) \geq \mu$, then $K^{\oplus r_\mu}$  deforms retract on $K$, see \cite[Theorem 12]{kim2019homotopy}.

In \cite{cl2005lambda} the {\it $\lambda$-medial axis} of $K$, denoted here $\ax_\lambda (K)$, was introduced.
Where the medial axis is  the set of points in $\R^n$ such that $\F_K (x)>0$,  the $\lambda$-medial axis of $K$ is a filtered version of 
it,  defined as  the set of points in $\R^n$ such that $\F_K (x)\geq \lambda$.
{ Since $\F_K$ is upper semi-continuous \cite[Corollary 4.7]{LIEUTIERhomotopytype}, $\ax_\lambda (K)$ is a closed set.}
%{\color{red} Is this just a less precise version of the statement below (with three points) or is there something else going on?} 
%{\color{red} 
For a given value of the filtering (pruning) parameter $\lambda$, %the authors in \cite{cl2005lambda}  proved that 
$\ax_\lambda (K)$
enjoys some geometrical and topological stability, see \cite{cl2005lambda} and the overview in Section \ref{sec:RecapStabilityLambdaMedialAxis} for details. %}

The medial axis is the limit of $\lambda$-medial axes in the sense that:
$\lambda' \leq \lambda \Rightarrow \ax_{\lambda'} (K) \supset \ax_\lambda (K)$ and
\begin{equation}\label{equation:LambdaMedialAxisFiltration}
\bigcup_{\lambda >0} \ax_\lambda (K) = \ax  (K).
\end{equation}
%It is a filtered stable version of the medial axis that, u

\section{Overview of results}\label{section:introContributions}
%{\color{red} [If we want to include contributions in the title of Section 6, I would call this overview of results.] } 

%{\color{red} [I now completely removed the modulus of continuity from this section. I'll leave it in the next section of course. The main reason was that I didn't see how to fit in the definitions without disturbing the flow. ]}

In this paper,  we show that a simple variant of the previous filtering $\lambda \mapsto \ax_\lambda (K)$,
% considering  instead merely the $\lambda$-medial axis of the $\alpha$-offset of $K$, for some $\alpha,\lambda >0$, 
enables significantly stronger stability statements.

The {\it $(\lambda,\alpha)$-medial axis}  of a closed set $K \subset \R^n$, 
denoted here $\ax_\lambda^\alpha (K)$, is  the $\lambda$-medial axis of the $\alpha$-offset\footnote{The $\alpha$-offset is denoted by $K^{\oplus \alpha}$, see \eqref{eqdef:Offset_Minkowski} { and the text following that equation for an explanation of the notation}.} of $K$:
\[
\ax_\lambda^\alpha (K) \defunder{=} \ax_\lambda (K^{\oplus \alpha}).
\]
%Notices that, since $\F_K$ is upper-semi continuous \cite[Corollary 4.7]{LIEUTIERhomotopytype}, $\ax_\lambda^\alpha (K)$ is a closed set.

It is just another similar way of filtering the medial axis,  where \eqref{equation:LambdaMedialAxisFiltration}
 is replaced by 
\begin{align}
\bigcup _{\lambda >0} \: \bigcup _{0< \alpha < \lambda}   \ax_{\lambda}^\alpha (K) %= \bigcup _{\lambda >0}   \ax_{\lambda}^{c \lambda} (K) 
= \ax(K).
\tag{\ref{equation:UnionOfLambdaAlphaMedialAxesGivesMedialAxis}}
\end{align}
% \eqref{equation:UnionOfLambdaAlphaMedialAxesGivesMedialAxis}.

The stability properties are then improved in two different ways: \newline
First, for $\lambda,\alpha>0$, if  $\chi_K$ does not vanish on some interval $[a,b]$ such that 
$a< \alpha$ and $\alpha+ \lambda < b$, then the map $(\lambda, \alpha, K) \mapsto \ax_\lambda^\alpha (K)$ is continuous 
for the (two-sided) Hausdorff distance on both the input $K$ and the output $\ax_\lambda^\alpha (K)$. 
Moreover, we give an explicit H{\"o}lder exponent in terms of 
%,  with a quantified modulus of continuity with respect to 
$\lambda$, $\alpha$: %and 
For $K: (\lambda, \alpha, K) \mapsto \ax_\lambda^\alpha (K)$ the H{\"o}lder exponent is $1$ with respect to $\lambda$ and $\alpha$, i.e. it is locally Lipschitz with respect to $\lambda$ and $\alpha$ (Lemma \ref{corollary:LambdaMapstoAXLambdaAlphaIsHausdorffLipchitz} and Lemma \ref{lemma:corollary:AlphaMapstoAXLambdaAlphaIsHausdorffLipchitz}). The map is $\frac{1}{2}$-H{\"o}lder
with respect to $K$ (Lemma \ref{lemma:AXLambdaHausdorffStableSymmetric}).

Secondly, we extend the stability results to the Gromov-Hausdorff distance, see Section \ref{section:GromovHausdorffDistance} for a formal definition. 
%The proof of Hausdorff stability in \cite{cl2005lambda} already relies on an upper bound of  the path length along the flow induced by the generalized gradient of the distance function (to the set), which may suggest to consider geodesic distances inside the medial axes. {\color{red} In hindsight perhaps, but I don't think the Gromov-Hausdorff distance was on your radar back then}
%Pursuing this idea, 
We show here that connected  $(\lambda,\alpha)$-medial axes are compact subsets of Euclidean space and have finite geodesic diameter (Theorem \ref{theorem:AXLambdaAlphaIsGeodesicAndFiniteDiameter}). 
Therefore $(\lambda,\alpha)$-medial axes equipped with intrinsic geodesic distances on $\ax_\lambda^\alpha (K)$) give meaningful metric spaces. 
We show that $\ax_\lambda^\alpha (K)$ seen as metric spaces is Gromov-Hausdorff stable under Hausdorff distance perturbation  of $K$, 
which can be expressed as the continuity of the map $(\lambda, \alpha, K) \mapsto \ax_\lambda^\alpha (K)$ under the associated metrics.
Moreover we  again establish bounds on the H{\"o}lder exponent 
in this new metric context: this map is locally Lipschitz with respect to $\lambda$ and $\alpha$  (Lemmas \ref{lemma:AxLambdaAlphaGHStableWRTLambda} and \ref{lemma:AxLambdaAlphaGHStableWRTAlpha}) and 
  $\frac{1}{4}$-H{\"o}lder with respect to $K$ (Theorem \ref{theorem:GromovHausdorffStability}).

This Gromov-Hausdorff stability gives metric %induces a form of topological {\color{red} I would not say topological here we can just remove it I believe, but wanted to check} 
stability which complements %beside 
the homotopy type preservation and Hausdorff distance stability. %,in a sense beyond the mere homotopy type preservation but weaker than an unrealistic homeomorphism. 
 It is the strongest form of stability we can hope for because the stronger property of bounded Fr{\'e}chet distance\footnote{ Recall that the Fr{\'e}chet distance between two subsets $S_1,S_2$ of a same metric space is the infimum of $\sup_{x\in S_1} d(x,h(x))$ among all possible homeomorphisms $h: S_1 \to S_2$. 
  It is therefore infinite when shapes are not homeomorphic. Note that we do not consider the orientation of the sets $S_1$ and $S_2$. } is impossible to achieve because of topological instability. In particular small smooth changes in a set can create changes in the topology of the medial axis.
 %, e.g. a circle being deformed into an ellipse. [OR] see Figure \ref{figure:FrechetVersusGH}
 %Recall that the Fr{\'e}chet distance between two subsets $S_1,S_2$ of a same metric space is the infimum of $\sup_{x\in S_1} d(x,h(x))$ among all possible homeomorphisms $h: S_1 \to S_2$. 
%%Recall that the Frechet distance between two subsets of a same metric space is the infimum of deviations among all possible homeomorphic maps between them.
  %It is therefore infinite when shapes are not homeomorphic.

Figure \ref{figure:FrechetVersusGH} illustrates three situations where the two shapes, the red and the blue, share the same homotopy type, as they all deform  retract to a circle, and are close to each other with respect to the Hausdorff distance: any point in the red shape is near the blue shape and the reverse holds as well.
On the first example, both distances, Fr{\'e}chet ($d_F$) and Gromov-Hausdorff ($d_{GH}$) are large, { because the distances in the `tail' differ significantly thanks to the zigzag. }
Because of our bound on the Gromov-Hausdorff distance (Theorem \ref{theorem:GromovHausdorffStability}), this situation cannot occur if the red and blue sets are the medial axis of two sets with small Hausdorff distance between them.
%It therefore cannot
 %be the situation of  medial axes of a red and a blue set   within a small Hausdorff distance of each other, 
%since, according to  Theorem \ref{theorem:GromovHausdorffStability}  
 %the Gromov-Hausdorff distance $d_{GH}$ between the two medial axes would  be small then.
 
 On the two next examples of Figure \ref{figure:FrechetVersusGH}  the red and the blue shapes do correspond
 to medial  axes of two sets close to each other in Hausdorff distance (in dotted lines). On the middle, the medial axes are similar but not homeomorphic, so that the Frechet distance %does not apply (or 
is infinite.  
In the last case they are homeomorphic but the Fr{\'e}chet distance would still be large  (you would need to rotate one of them by $90^\circ$ for the homeomorphism). 
In contrast, as asserted by Theorem \ref{theorem:GromovHausdorffStability}, the Gromov-Hausdorff distance between them is small. % on the two examples. 

 \begin{figure}[ht!]
\centering
\includegraphics[width=0.49\textwidth]{%pictures/
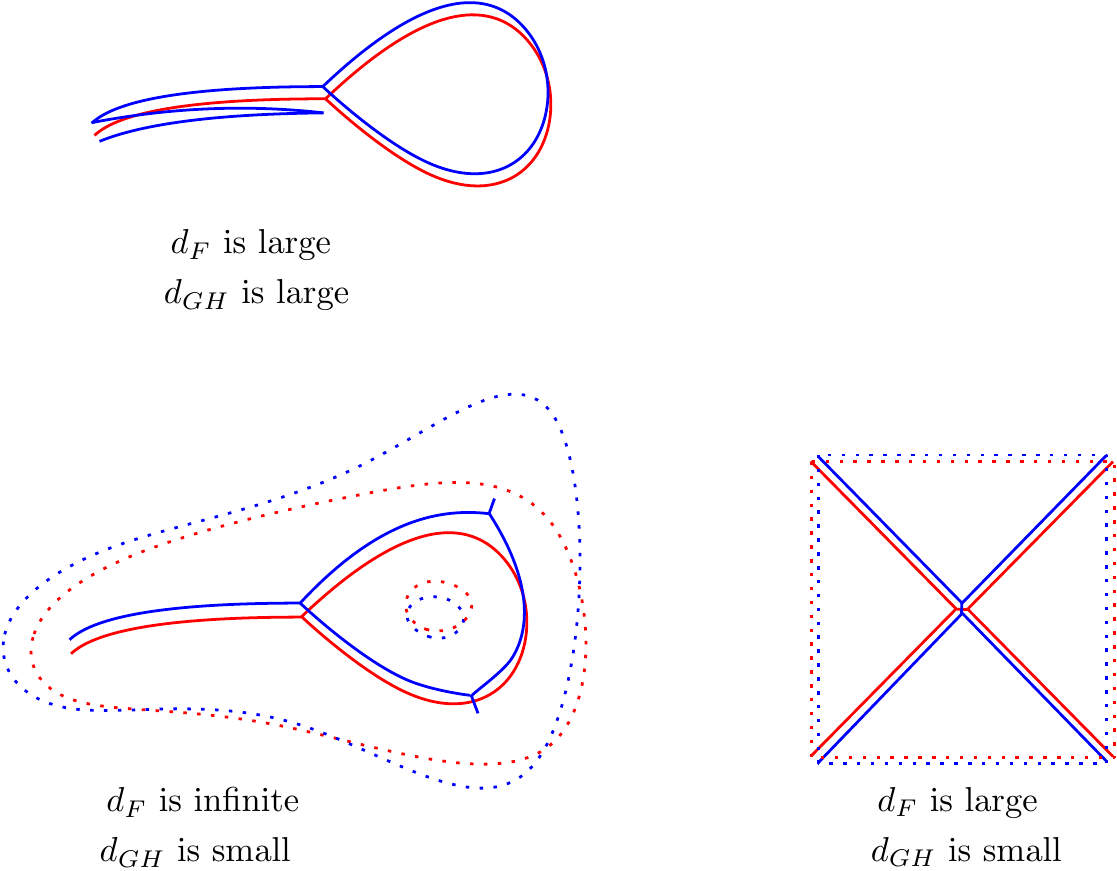}
\caption{Comparison between Fr{\'e}chet ($d_F$) and Gromov-Hausdorff ($d_{GH}$) stability. 
On the two examples below, the shapes are ($\lambda,\alpha$)-filtered medial axes of nearby sets (in dotted lines), and
as asserted by  Theorem \ref{theorem:GromovHausdorffStability}, the Gromov-Hausdorff distance between them  is small. 
\label{figure:FrechetVersusGH}
}
\Description[Figures with pairs of spaces with the Fr{\'e}chet and Gromov-Hausdorff distance]{On the top one sees a figure with a zigzag in blue and one without in red. This makes both the Fr{\'e}chet and Gromov-Hausdorff distance large. On the bottom left one sees two medial axis that are close, but not homeomorphic (due to some small Y-branching). In this case the Fr{\'e}chet is infinite, while the Gromov-Hausdorff distance in very small. In the final panel one sees two medial axis that are perturbations of a cross (X) where in one case the crossing point is perturbed into a small horizontal line segment and in the other case in a small vertical line segment. The homeomorphism that maps the two perturbed crosses onto each other has to rotate the shape 90 degrees, so that the Fr{\'e}chet distance is large. The Gromov-Hausdorff in the other hand is small.  }
\end{figure}

Gromov-Hausdorff stability can be seen informally as a weakening of  Frechet distance that ignores small scale features.
%  and, roughly speaking,  may be in a sense the best we can hope in the direction of topological stability.

%{\color{red} [Do we want to say something about the intuition suggested by the reviewer: `Taking the offset smooths a set' I do not agree with the assessment and we do say a lot of other things about the intuition, but it is something to think about.]} 

\section{Motivation}
\subsection{Medial axis computation algorithms and models of computation}\label{section:introModelOfComputation}

The medial axis is known to be unstable in theory \cite{Attali2009}, and, as a consequence, its computation is often problematic in practice. 
A typical  illustration of this instability is when $K^c$ is an open disk in the plane: its medial axis is a point, but a  $C^\infty$ perturbation, arbitrary small, in the $C^0$ sense of differential topology \cite{Hirsch1976},  of its boundary,
 may produce an arbitrary large perturbation (measured in the Hausdorff distance) of the resulting medial axis. 
 
 Computing the medial axis consists in, given as input some representation of the closed set $K$, to compute as output some representation of $\ax(K)$. 
  Let us recall  two possible  computation models under which  what it means to ``compute'' $K \mapsto  \ax(K)$.
 
 In computational geometry, the implicit computation model (sometimes called exact computation paradigm in order to distinguish it from the unrealistic ``Real RAM'' computation model) assumes that both input and output can be exactly represented by finite data in the computer. This implies that input and output have  to belong to countable sets,\footnote{As only countable sets can have each of its elements  representable by a finite word.} such as, for example, integer,  rational or algebraic numbers, or polynomials built on top of them. Given a set of rational or algebraic points, or given a polyhedron with rational or algebraic vertices coordinates, for example, we now that the medial axis is a finite algebraic complex and, as such, belongs to a countable set, therefore exactly presentable on a computer.  These are situations where 
it makes sense to compute the medial axis in this exact computation model,  even if it may be difficult. 
 
Computable analysis, pioneered with the notion of computable real numbers  introduced by   
Turing in his 1936 undecidability paper %{\color{red} the two have the same title. Is this the way it should be? } 
\cite{turing1937entscheidungs,turing2004computable}, is   
  studied in the logic and theoretical computer science literature
   \cite{grzegorczyk1955computable, Lacombe1, Lacombe2, Lacombe3, Lacombe4, malajovich1992ker, weihrauch2000computable, brattka2008tutorial, battenfeld2008topological}, 
   but its formalism is most often ignored  in  applications.

However, it is actually implicit in many practical computations involving real numbers and real functions, 
for example in numerical analysis, where a typical example would be the finite element method.
In this context, one considers that input and output can belong to topological spaces with countable bases of neighbourhoods, typically metric spaces with dense countable subsets, called separable metric spaces, who, as a consequence, have at most the cardinality of real numbers.  Examples of such metric spaces are: %{\color{red} [I tinkered with the distance between the items to safe space, but feel free to change if you think it is ugly]}
\begin{itemize} 
    \setlength{\itemsep}{0pt}%
    \setlength{\parskip}{0pt}
\item Real numbers with their natural topology (rational numbers are dense).
\item Continuous functions on a compact set with the sup norm (polynomials  with rational coefficients are dense, by the Stone-Weierstrass Theorem). 
\item $L^p$ (classes of) functions with their associated  $L^p$ norms (rational step functions are dense).
\item Compacts subsets of Euclidean spaces endowed with the Hausdorff distance (finite points sets in $\Q^n$  are dense).
\end{itemize} 
 
In the context of these separable metric spaces, an algorithm, in this model of computation,
 takes as input a sequence belonging to the dense subset, so that each element of the sequence, belonging to a countable space, admits a finite representation.\footnote{The dense set has, formally, to be recursively enumerable.} It then computes,  
  for each element of the input sequence, an element of the output sequence in such a way that the output sequence converges to the image of the limit of the input sequence.
 This mere definition assumes that the (theoretical) output of the limit of the input sequence, is the limit of the sequence 
 of (actual) outputs of items of the input sequence. This is the reason why, in the context of computable analysis, 
 only continuous functions, that commute with limits, can be computable\footnote{In fact computability of the function requires moreover the modulus of continuity of the map to be computable, in particular should not tend to $0$ slower than  any recursive function.}. 
 For example, integer part function is computable, in this model,
 only at non-integer numbers. In decimal representation, if, after the dot, an infinite sequence of $9${s} appears, the algorithm would read the input forever.  
%In terms of {\color{red} the?} condition number, a non-computable situation can be interpreted as one of infinite (absolute) condition number. {\color{red} Should we define the condition number?}

Recall that a continuous function $\omega:{\R_{\geq 0} \rightarrow \R_{\geq 0} }$,  with $\omega(0)=0$, is a {\it modulus of continuity} of a map $f:X\rightarrow Y$ between metric spaces if for all $x_1,x_2 \in X$, 
%{\color{red} [I tinkered with the text because I did not like the spacing in the formula, it used to be `quad', but the improvement with align was not that great]}
\begin{align} 
%\forall x_1,x_2 &\in X,&  
d_Y(f(x_1),f(x_2)) &\leq \omega(d_X(x_1, x_2)).  
\nonumber
\end{align} 
If one wishes to control some form of theoretical algorithmic efficiency in the context of computable analysis,
 a modulus of continuity of the operator, that associates to some uncertainty on the input an upper bound on the induced uncertainty on the output,
 needs to be estimated. 
 
 We do not  need to enter here in the technicalities of computable analysis. Our contribution consists in stating some explicit modulus of continuity,
 which, on the theoretical side,  would be a crucial ingredient in the proofs of computability and complexity in the context of computable analysis, but is 
 also, on the application side, a way to guarantee some accuracy in practical computations.
Indeed, practical implementations of the computation of the medial axis apply some kind of  approximation during the
 computation process. In a practical situation, this approximation process is already inherent to the data collection process, as any physical numerical measure
 is meant at some, finite, accuracy. Second, the actual input of an algorithm is often the output of a preceding algorithm which cannot, reasonably,
  be assumed recursively to compute exact output from exact inputs: recursion on  algebraic numbers representations are possible along a finite depth of computation only.
When, along the  process,  
some form of rounding, pixelization,  small features collapses or filtering, is performed, 
being able to upper bound the impact on the output seems sensible,
%.
%{\color{red} [The following sentence needs to be looked at. `reveals' in this context is strange, but I am not entirely sure what we want to say.]} 
%{\color{blue} ( }For the actual computation of objects as unstable as the medial axis and some of its variants \cite{tagliasacchi20163d}, an estimation of the modulus of continuity reveals also profitable, this time  in a more practical perspective.
%{\color{blue} ) replace by ?? $\rightarrow$ (Therefore, for the actual computation of objects as unstable as the medial axis and some of its variants \cite{tagliasacchi20163d}, an estimation of the modulus of continuity is profitable also  in a more practical perspective.) $\rightarrow$ or remove the sentence which anyway does not bring much more information ?? } {\color{red} or add: ... ``sensible,
 and in fact necessary for provably correct algorithms.

Since  $K \mapsto \ax (K)$ is not continuous in general when the topology of both inputs and outputs are defined by the Hausdorff distance, 
we see two ways of stating a continuity, or stability, property, for the operator  $K \mapsto \ax (K)$. One possibility is to consider a stronger topology on the input,
a form of Fr{\'e}chet, or ambient diffeomorphism based, $C^k$ distance, which would apply to smooth objects and representations.

Another possibility is to consider a weaker topology on the output, by considering filtered medial axes. 
In this model, the input  sequence encodes $K$ in the form of approximations $\big(\tilde{K}_i\big)_{i\in \Nat}$ that converge to $K$ in Hausdorff distance.
For the $\tilde{K}_i$ one would typically choose finite point sets or (geometric) simplicial complexes (meshes/triangulations). As $i$ would increase in one would not only add more points or simplices to $K_i$, but also make the coordinates of the points/vertices  more precise by adding digits to their coordinates.
%Finite point sets, or polyhedral descriptions, are typical examples of such approximations,
%where the infinite sequence, in theory,  progressively encodes not only new points, or vertices and mesh updates,
%but also adds new digits to the  vertices coordinates.
%Finite point sets, or polyhedral descriptions, are typical examples of such approximations,
%where the infinite sequence, in theory,  progressively encodes not only new points, or vertices and mesh updates,
%but also adds new digits to the  vertices coordinates.

The output sequence encodes $\ax(K)$, 
in the form of progressive  approximations of the map $(\lambda, \mu)  \mapsto \ax_\lambda^\alpha(K)$, 
for decreasing values of  $\lambda,\alpha$.
These approximations (effectively) 
%{\color{red} [what is the precise meaning of effectively in this context?]} 
%{\color{purple} ... Well "effectivity"  is a standard notion for computer science theorists but  should we  elaborate on this there ?... Yet, for your education, informally, "effectively" means that  each time you write "$\forall \epsilon$ such  that ..., $\exists \alpha$ such that ...", you make it effective by replacing it by :"there exists a Turing machine ( or an algorithm that processes finite data only (no real numbers for example) that, given $\epsilon$ such that ..., compute $\alpha$ such that ...". A $\beta$-Holder map for example  is trivially effectively continuous because a Turing machine can compute $\epsilon \mapsto \alpha = C. \epsilon^{1/\beta}$.} MW: Many thanks, this sense and I agree not to explain more.
converge, 
where a  basis of neighbourhoods  (in the space of functions) 
of $(\lambda, \alpha)  \mapsto \ax_\lambda^\alpha(K)$ 
%{\color{blue} {\color{red} 
is given by the sets of maps 
$(\lambda, \alpha) \mapsto f(\lambda, \alpha)$ satisfying 
%{\color{red} [This looks a bit strange, but I cannot lay my finger on the issue]} 
$\lambda, \alpha > t  \Rightarrow d_\star(f(\lambda, \alpha)	, \ax_\lambda^\alpha(K)) < \epsilon$ 
 for some  $\epsilon, t >0$.

 This approach does not require any smoothness assumption on $K$. The present paper focuses on this filtered approach,
 where the considered distance $d_\star$ between sets  is  either the Hausdorff distance, either the Gromov-Hausdorff distance on geodesic metric spaces.
 
Describing formally  effective types and algorithms for the computation of the medial axis is beyond the scope of this paper. However, let us make some suggestions for further work in this direction.  Probably the simplest  model would consider the space of finite set of points with rational coordinates as inputs.  
These inputs together form a countable, and recursively enumerable set which is naturally equipped with the Hausdorff distance. The topological completion of the set of inputs gives all compact subsets of Euclidean space. 
The corresponding output space would consist of the filtered Voronoi Diagrams for which the coordinates of the Voronoi vertices are rational numbers. The H{\"o}lder  modulii of continuity proven in this paper
would allow to formally state the effectivity of the model.

%{\color{red} Tinkered with the text, but I am not sure I am completely done}
%{\color{blue} 
The model could also be formalized in the context of Scott domains \cite{abramsky1994domain,edalat1998computational}, \cite[Chapter 1]{amadio1998domains}
 and their associated  information orders.\footnote{ It is possible, following \cite{edalat1998computational}, to topologically embed our input and output metric spaces as maximal elements of some Scott domains. Our bounded modulus of continuity would then allow to provide effective structures for them.} In this context, our results answer the following question: If the only information we have about some compact set $K$ is its Hausdorff approximation $K'$, what information can we infer about its medial axis $\ax(K)$?
%}

%giving effective ..

\subsection{Motivation from mathematics: the stability of the cut locus.}

The medial axis is closely related to the cut locus. We recall
\begin{definition}
Let $\M$ be a smooth (closed)  Riemannian manifold and let $p \in \M$. For every $v \in T_p \M$, with $|v|=1$, we can consider the geodesic $\gamma_v (t) = \exp _p (t v )$ emanating from $p$ in the direction $v$. Let $\gamma_v(\tau)$ be the first point along $\gamma_v$ such that the geodesic $\{ \gamma_v(t) \mid t \in [0, \pi]\}$ is no longer the unique minimizing geodesic to $p$. The cut locus of $p$ is the union of these points for all unit length $v$ in $T_p\M$. 
\end{definition}
The cut locus is therefore more general in the sense that it is defined for general Riemannian manifolds, while more restrictive in the sense that it only considers a single point.\footnote{The reach and medial axis can be defined for closed subsets of Riemannian manifolds \cite{kleinjohann1980convexity, kleinjohann1981nachste,bangert1982sets,JDM2022}. } 

The stability and structure of the singularities of the cut locus has been a studied intensely. 
Buchner \cite{buchner1977stability} derived the following result:
\begin{theorem} 
Let $G$ be the space of metrics on a smooth manifold, endowed with the Whitney topology. Each metric $g\in G$ and $p \in \M$ yield a cut locus $C_{p,g}$.  The cut locus $C_{p,g}$ is called stable if there is a neighbourhood $W \subset G$ of $g$ such that for any $g'\in W$ there exists a diffeomorphism $\Phi : \M \to \M$ such that $\Phi(C_{p, g'} ) =C_{p,g}$. If the dimension of $\M$ is low ($\leq 6$) then $C_{p,g}$ is stable for an open and dense subset of $G$.   
\end{theorem} 
Wall \cite{Wall1977} extended this result to arbitrary dimensions at the cost of weakening the diffeomorphism to a homeomorphism. 
The structure of the singularities of the cut locus were also described by Buchner in \cite{Buchner1977Structure}. A similar description for the singularities of medial axis of a smooth manifold can be found in \cite{yomdin1981local}, see also \cite{mather1983distance}, as well as \cite{van2007maxwell}.    

This paper follows the tradition of these investigations of the stability of cut locus and the medial axis. However, there are also some significant differences. First and foremost we take a metric viewpoint instead of analytical. This viewpoint does not require us to make a distinction between low dimensional and high dimensional spaces. We made the constants explicit in view of the applications in computer science, in particular computational geometry and topology, shape recognition, shape segmentation, and manifold learning. 

{The authors are currently working on the stability of the cut locus and medial axis of smooth sets, using the tools which we develop in this paper. }

\section{Overview of the stability of the \texorpdfstring{$\lambda$}{lambda}-medial axis} \label{sec:RecapStabilityLambdaMedialAxis}
Under mild conditions, the $\lambda$-medial axis enjoys some nice stability properties, assuming that $K^c$ is bounded. Informally:
\begin{enumerate} %[label={P.\arabic*}]
    \setlength{\itemsep}{0pt}%
    \setlength{\parskip}{0pt}
\item \label{P1} %[(1)]  
When $\chi_K$ does not vanish on $(0,\lambda]$, the $\lambda$-medial axis preserves the homotopy type of the complement $K^c$ of $K$ \cite[Theorem 2]{cl2005lambda},
\item \label{P2} %[(2)]  
Taking the Hausdorff distance on the input $K$ and the one  sided Hausdorff distance on the output $\ax_\lambda(K)$
 we get a kind of modulus of continuity. If $d_h(K,K')$ is small, the points in $\ax_\lambda (K)$ are ``near'', in a quantified way,  $\ax_{\lambda'} (K')$, for some 
 $\lambda' < \lambda$ close to $\lambda$  \cite[Theorem 3]{cl2005lambda}.% These values are proved to be finite for analytic sets but the practical impact of this restriction is unclear for  high-dimensional sets. 
 \item  \label{P3} %[(3)] 
For ``regular values of $\lambda$''  the map $K \mapsto \ax_\lambda(K)$ is continuous for the Hausfdorff distance  \cite[Theorem 5]{cl2005lambda}.
 However, the modulus of continuity can be arbitrarily large in general.
\end{enumerate}%{itemize} 

Property \ref{P1} gives some stability on the homotopy type with respect to Hausdorff perturbation of $K$, since, 
 under similar  conditions on the critical function of $K$, when $d_h(K, K')$ is small,  the offsets of $K'$ may share the homotopy  type of $K$ \cite{chazal2005weak,chazal2009sampling}.
Properties \ref{P2} and  \ref{P3} give  precise, quantified, stability results,  much stronger than the mere half-continuity of the medial axis itself, see e.g. \cite{Attali2009}.

\section{{Contributions:} the improved stability of the \texorpdfstring{$(\lambda,\alpha)$}{(lambda,alpha)}-medial axis}\label{section:introIntuitionLambdaAlphaMedialAxis}
%{\color{red} [I did not add contributions in the title as that is already the name of Section \ref{section:introContributions}]}
Before entering into the formal proofs, let us give some intuition about the $(\lambda,\alpha)$-medial axis  stability.

This improved stability 
%of the $(\lambda,\alpha)$-medial axis with respect to the $\lambda $-medial axis 
can be illustrated in the case of a finite set $K$.
Figure \ref{figure:LambdaAlphaMedialAxisTwoPloints} illustrates the $(\lambda,\alpha)$-medial axis in the simplest non-trivial case, where $K$ consists of two points in the plane.
In this case the $\lambda $-medial axis would be empty as long  as $\lambda$ is strictly  greater than the half distance  between the two points and it would become
 the whole bisector line as soon as $\lambda$ is smaller  or equal to this value.

\begin{figure}[ht!]
\centering
\includegraphics[width=0.5\textwidth]{%pictures/
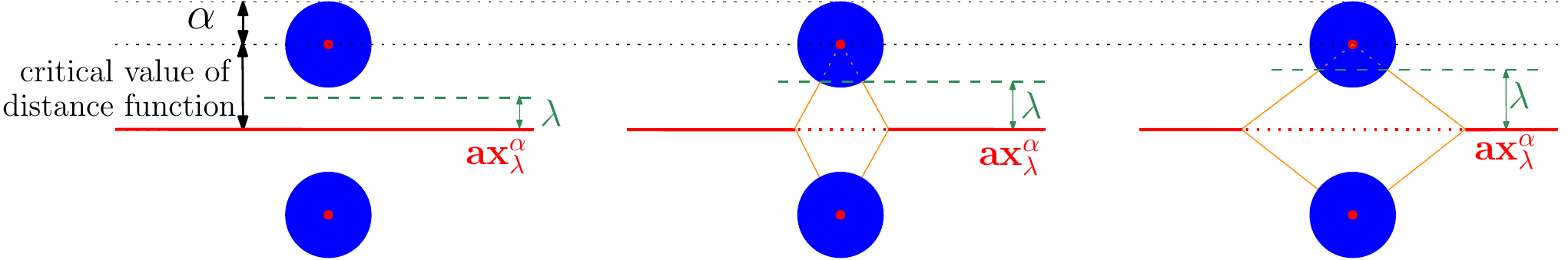}
\caption{
Comparison between $\lambda$-medial axis and $(\lambda$,$\alpha)$-medial axis evolutions for increasing $\lambda$, in the particular case where $K$ consists of two points in the plane.
The $\lambda$-medial axis would be either the whole bisector line of the two points, for $\lambda$ smaller or equal to half the distance between the points, either the empty set for larger value of $\lambda$. By contrast, the evolution, for increasing $\lambda$, of the $(\lambda,\alpha)$-medial axis, which is also the  $\lambda$-medial axis of the union of two disks of radius $\alpha$, evolves continuously, in Hausdorff distance, as $\lambda$ increases.
}
\Description{The $\lambda$-medial axis and $(\lambda$,$\alpha)$-medial axis of two points in the plane.}
\label{figure:LambdaAlphaMedialAxisTwoPloints}
\end{figure}

 By contrast, the $(\lambda,\alpha)$-medial axis, for a fixed value of $\alpha>0$, here the radius of the two disks of the $\alpha$-offset of $K$, is Hausdorff continuous with respect to $\lambda$. Indeed, as $\lambda$ increases, when $\alpha + \lambda$ equals the half distance between the two points, $\ax_\lambda^\alpha(K)$, which until then is the whole bisector line, starts to be disconnected, creating a hole. 
 But, since the  hole  grows continuously, its birth is not  a discontinuity for the Hausdorff distance. However, in the neighborhood of this event, the hole size grows quadratically 
 with $\lambda$: This does not contradict the claim that the map $\lambda \mapsto \ax_\lambda^\alpha(K)$ is Lipschitz, as the precise conditions of the claim require us to avoid situations where $\alpha+\lambda$ is a zero  of $\chi_K$.

\begin{figure}[htb!]
\centering
\includegraphics[width=0.5\textwidth]{%pictures/
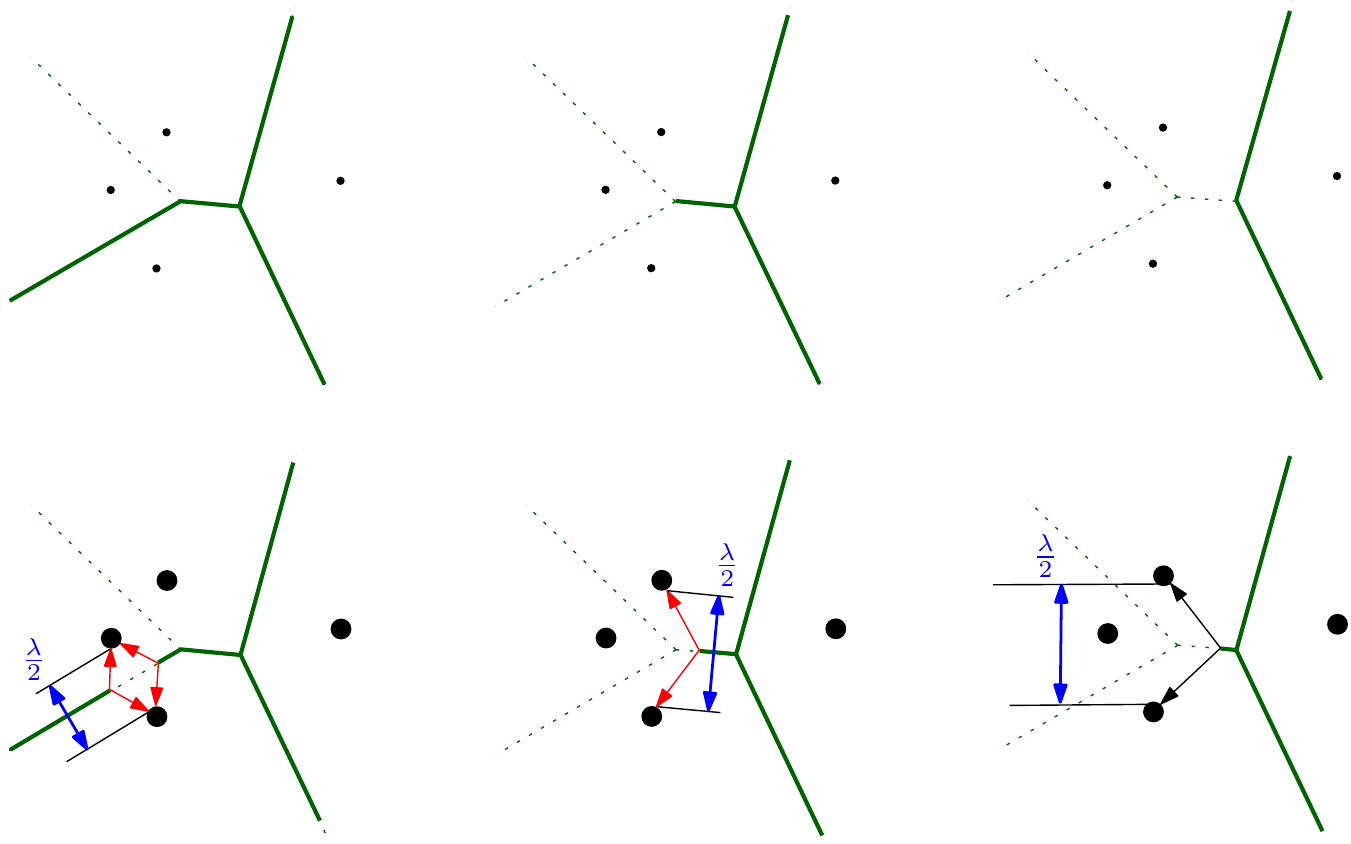}
\caption{
Comparison between $\lambda$-medial axis and $(\lambda$,$\alpha)$-medial axis evolutions for increasing $\lambda$, in the particular case where $K$ is a finite subset of the plane.
 In this case both filtered medial axes are subsets of the union of the edges of the  Voronoi diagram. 
 On the second row, the points have been replaced by disks of radius $\alpha$, offset of the points. 
The evolution of the $(\lambda$,$\alpha)$-medial axis is Hausdorff continuous whenever $\alpha+\lambda$ is not a critical value of the distance function. 
On the other hand, as seen on first row,  
the $\lambda$-medial axis   contains precisely the whole Voronoi edges or vertices whose dual simplex lies in a ball
 of radius $\lambda$. {The $\lambda$-medial axis is} therefore Hausdorff discontinuous for each value of  $\lambda$ which is the radius of the smallest ball enclosing some Delaunay simplex. 
  }
\Description[The $\lambda$-medial axis and $(\lambda$,$\alpha)$-medial axis of four points in the plane.]{The $\lambda$-medial axis and $(\lambda$,$\alpha)$-medial axis of four points in general position in the plane (for various values of $\lambda$). The points are arranged as follows: Three on the left and one on the right. The three points on the left form an obtuse triangle arranged in such a way that its circumcentre, one point on the left and the point on the right lie very close to a horizontal line. The two other points on the left lie on either side of the horizontal line and nearly above/below each other.}
\label{figure:LambdaAlphaMedialAxisSubsetVoronoi}
\end{figure}
 
Figure \ref{figure:LambdaAlphaMedialAxisSubsetVoronoi} shows a situation where $K$ is made of four points in the plane.
The $\lambda $-medial axis is made of these edges and vertices whose dual Delaunay simplex has smallest enclosing radius greater or equal to $\lambda$.
As a function of $\lambda$,  it is therefore Hausdorff distance discontinuous for each value of $\lambda$ that is equal to a such radius.

In contrast, the $(\lambda,\alpha)$-medial axis, as a function of $\lambda$ for fixed $\alpha>0$, can be Hausdorff discontinuous only when $\alpha+\lambda$ is a zero of the 
{critical function $\chi_K$}.
%{\color{red} Please read the following carefully} 
We have depicted such a transition in Figure \ref{figure:LambdaAlphaMedialAxisSubsetVoronoi_2}:  % on the bottom right:
%A such situation would appear if, beyond the example bottom right, as depicted on Figure \ref{figure:LambdaAlphaMedialAxisSubsetVoronoi_2},
Here we increase $\lambda$ further until  $\alpha+\lambda =\rho$, where $\rho$ is the circumradius of the unique acute triangle in the Delaunay diagram, 
and { therefore the unique value of the distance to $K$ corresponding to a local maximum}. 
Until $\alpha+\lambda =\rho$ the $(\lambda,\alpha)$-medial axis would contain the Voronoi vertex dual to this acute triangle (for $\alpha+\lambda =\rho$ the Voronoi vertex would be an isolated point).
 Since this points would disappear from the $(\lambda,\alpha)$-medial axis for $\alpha+\lambda>\rho$, 
it results a Hausdorff distance discontinuity of $\lambda \mapsto \ax_\lambda^\alpha(K)$. 
%A such situation would appear if, beyond the example bottom right, as depicted on Figure \ref{figure:LambdaAlphaMedialAxisSubsetVoronoi_2},
%if we increase $\lambda$ further until  $\lambda =\rho$, where $\rho$ is the circumradius of the unique acute triangle in the Delaunay diagram, 
%the $(\lambda,\alpha)$-medial axis would contains, as an isolated point, the Voronoi vertex dual to this acute triangle. Since this points would disappear from the $(\lambda,\alpha)$-medial axis for $\lambda>\rho$, 
%it results a Hausdorff distance discontinuity of $\lambda \mapsto \ax_\lambda^\alpha(K)$. 

\begin{figure}[ht!]
\centering
\includegraphics[width=0.5\textwidth]{%pictures/
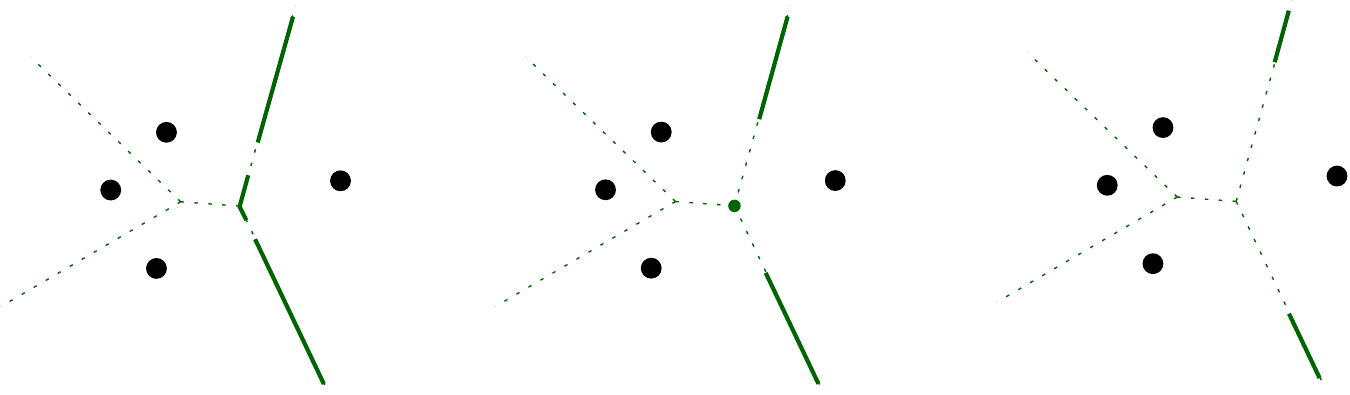}
\caption{
Further evolution  of the ($\lambda$,$\alpha$)-medial axis, after the steps of the bottom row of figure \ref{figure:LambdaAlphaMedialAxisSubsetVoronoi}.
For some small interval of values of $\lambda$, the Voronoi vertex is an isolated point on $\ax_\lambda^\alpha(K)$, as illustrated on the middle,
When $\alpha+\lambda $ equals the Delaunay triangle circumradius, this points disappears from $\ax_\lambda^\alpha(K)$, which 
 corresponds to a discontinuity of  $\lambda \rightarrow \ax_\lambda^\alpha(K)$  for the Hausdorff distance.
 }
\label{figure:LambdaAlphaMedialAxisSubsetVoronoi_2}
\Description{The same configuration of points in the plane as in the previous figure (Figure \ref{figure:LambdaAlphaMedialAxisSubsetVoronoi}).}
\end{figure}

In general,  Hausdorff distance discontinuities of $\lambda \mapsto \ax_\lambda(K)$ may appear anywhere, at some ``non regular values'', 
as mentioned in item \ref{P3} %{\color{red} TODO fixme } 
of Section \ref{sec:RecapStabilityLambdaMedialAxis} and illustrated on top row of Figure \ref{figure:LambdaAlphaMedialAxisSubsetVoronoi} 
{ (where $\lambda$ is a zero  of $\chi_K$)}
and in Figure \ref{figure:MLambdaDiscontinuous} {(where $\lambda$ is not a zero  of $\chi_K$)}.
By contrast, the map $\lambda \mapsto \ax_\lambda^\alpha(K)$, for $\alpha>0$,
 is continuous (locally Lipschitz) when $\chi_K(\alpha + \lambda)$ does not vanish, in other words the interval on which the homotopy type of 
 $ \ax_\lambda^\alpha(K)$ remains stable.

%{\color{red} This sentence is a bit too long and I don't know where to break precisely} 
%By contrast, the map $\lambda \mapsto \ax_\lambda^\alpha(K)$, for $\alpha>0$,
 %is continuous (locally Lipschitz) in the ``useful interval'' $\lambda \in [\lambda_{\min}, \lambda_{\max}]$, see for example bottom row of Figure \ref{figure:LambdaAlphaMedialAxisSubsetVoronoi} and in Figure \ref{figure:MLambdaContinuous}.
%Namely when the critical function $\chi_K$ does not vanish in the interval 
 %$[\alpha +\lambda_{\min}, \alpha +\lambda_{\max}]$, in other words the interval on which the homotopy type of offsets and, equivalently,  offsets medial axes, is guaranteed to be stable.

 \subsection{The case of \texorpdfstring{a set $K\subset \R^n$}{ a subset of Euclidean space } with positive \texorpdfstring{$\mu$}{mu}-reach  and its Hausdorff approximation \texorpdfstring{$K'$}{K'}}
 \label{section:CaseKPositiveReach}

%{\color{red} MW I'll try to add comments before the paragraphs in the hope that I don't mess things up. Typo: part \ref{section:MedialAxisAndFlow} $\to$ Section \ref{section:MedialAxisAndFlow} (perhaps in Part \ref{part:TechnicalStatementsAndProofs}). More generally: 
%In Section \ref{section:MedialAxisAndFlow} we consider the general situation of sets whose $\alpha$-offsets have positive $\mu$-reach. 
%Also I would write: the modulus of continuity\footnote{We recall that a  modulus of continuity of a map $f:X \to Y$ between metric space $(X,d_X)$ and $(Y,d_Y)$ is a continuous function $\omega$ (that vanishes at 0) such that $d_Y(f(x), f(x')) \leq \omega(d_X(x,x'))$ for all $x,x'\in X$. } for the map $K  \mapsto \ax_\lambda^\alpha (K)$, where the metric on $K$ is the Hausdorff distance and the metric on the medial axis can be either the Hausdorff distance or the Gromov-Hausdorff distance.
%} 

In \ref{part:TechnicalStatementsAndProofs} we consider the general situation of sets whose $\alpha$-offsets have positive $\mu$-reach. 
In particular, Lemma \ref{lemma:AXLambdaHausdorffStableSymmetric} and 
Theorem \ref{theorem:GromovHausdorffStability} use a symmetrical formulations on the  pair of sets $K$ and $K'$ in order to state 
a modulus of continuity
%\footnote{We recall that a  modulus of continuity of a map $f:X \to Y$ between metric space $(X,d_X)$ and $(Y,d_Y)$ is a continuous function $\omega$ (that vanishes at 0) such that $d_Y(f(x), f(x')) \leq \omega(d_X(x,x'))$ for all $x,x'\in X$. } {\color{blue} Andre: we could remove the footnote as the modulus of continuity is now defined in Section \ref{section:introModelOfComputation} ?} MW: agreed
for the map $K  \mapsto \ax_\lambda^\alpha (K)$, where the metric on $K$ is the Hausdorff distance and the metric on the medial axis can be either the Hausdorff distance or the Gromov-Hausdorff distance.

In this section we consider the simpler setting, where we don't need to offset for $K$ to achieve positive $\mu$ reach, that is $r_\mu (K)>0$ and we are given a set $K'$ that is close to $K$ in terms of the Hausdorff distance.
This allows a concise formal expression of our main results in a simpler setting, while illustrating a typical application.

%We consider in this section  instead the particular case of a set $K\subset  \R^n$, for which 
% we assume  $r_\mu (K)>0$, and for which an approximation  is known, as a closet set  $K' \subset  \R^n$, 
% for which  
% \[
%d_H(K', K) < \epsilon
%\]
% for some $\epsilon >0$.
%  This particular situation allows a concise formal expression of our main results in a simpler setting,
%while illustrating a typical application case.

%{\color{red} I would be tempted to make this into a theorem environment so that we can refer to it.}
Overall this section, we make the following assumption:
\begin{assumption}[Assumption for Section \ref{section:CaseKPositiveReach}]
We assume  $K,K'$ to be closed sets such that, for  some $\epsilon >0$, 
 $d_H(K', K) < \epsilon$,  $r_\mu (K)>0$ and, 
 the complements $K^c$ and $K'^c$ to be bounded:
denoting $\B(0, R) \subset  \R^n$  the ball of radius $R>0$, one has 
$K\cup \B(0, R)  = K'\cup \B(0, R) = \R^n$.
We assume moreover $0< \alpha < \alpha_{\max}$ and $0< \lambda <  \lambda_{\max}$, for $\alpha_{\max} +  \lambda_{\max} <  r_\mu (K)/2$
and we denote  $\tilde{\mu} = \min( \mu, \sqrt{3}/2)$.
\end{assumption}
In particular, assuming $\alpha_{\max} +  \lambda_{\max} <  r_\mu (K)/2$ allows a simple expression for  $\tilde{\mu}$.

%{\color{red} This may be a problem to change, but I am not sure if $\beta$ is good notation, I would be tempted to go for $\tilde{\mu}$ or $\mu_{B}$, to indicate the relation with $\mu$.  } 

\subsubsection{Hausdorff stability}
As a consequence of Lemma \ref{corollary:LambdaMapstoAXLambdaAlphaIsHausdorffLipchitz} and Lemma 
\ref{lemma:corollary:AlphaMapstoAXLambdaAlphaIsHausdorffLipchitz} we have that:
\begin{proposition}
For any $\lambda_{\min}>0$,
  the map $\lambda \mapsto \ax_\lambda^\alpha (K)$
is $\left( \frac{R^2}{\alpha \lambda_{\min} \tilde{\mu}^2} \right)$-Lipschitz in the interval $[\lambda_{\min}, \lambda_{\max}]$
 for  Hausdorff distance. 
 
 Similarly,  for  $\alpha_{\min}>0$,
 the map
$ %\[
\alpha \mapsto \ax_{\lambda}^\alpha (K)
$ %\]
is $\left( \frac{R^2}{\alpha_{\min}  \lambda \tilde{\mu}^2} \right)$-Lipschitz in the interval $[\alpha_{\min}, \alpha_{\max}]$
 for  Hausdorff distance.
 \end{proposition}

% On the topological side, we know by , ..., that $\ax_{\lambda}^\alpha (K)$ has the homotopy type of $K^c$.

%{\color{red} I would suggest (please check): 
We will now combine this with a result from \cite[Theorem 3.4]{chazal2009sampling}. Let $\mu'< \mu$ and $\alpha >0$. By definition of $r_{\mu} (K)$, the critical function of $K$ is above $\mu$ on the interval $(0, r_{\mu} (K) )$. Theorem 3.4 of \cite{chazal2009sampling} now says that if $K'$ is sufficiently close to $K$ in Hausdorff distance, then the critical function of $K'$ will also be above $\mu'$ on the interval $(\alpha, r_{\mu} (K) - \alpha)$, see Figure \ref{figure:CriticalFunction}.
%}
% 
% Now, given $\mu'< \mu$ and $\alpha >0$ , we know from   \cite[Theorem 3.4]{chazal2009sampling}
%  that, since, by definition of $r_{\mu} (K)$, the critical function of $K$ is above $\mu$ on the interval $(0, r_{\mu} (K) )$, if $K'$ is close enough to $K$,
%  in Hausdorff distance, then the critical function of $K'$ will be above $\mu'$ on the interval $(\alpha, r_{\mu} (K) - \alpha)$, see Figure \ref{figure:CriticalFunction}.
  In other words, there is $\epsilon >0$ such that:
 \begin{equation}\label{equation:StabilityRmu}
d_H(K',K) < \epsilon  \Rightarrow    r_{\mu'}^\alpha (K') > r_{\mu} (K) - \alpha.
 \end{equation}
 Then, Lemma \ref{lemma:AXLambdaHausdorffStableSymmetric} gives us that:
 
\begin{proposition}\label{proposition:PositiveMuReachHausdorffHolder}
Denoting  $\tilde{\mu}' = \min( \mu', \sqrt{3}/2)$,
there is $\epsilon_{\max} >0$  depending only on $K$, such that, for, $\epsilon < \epsilon_{\max}$, one has:
\begin{equation}
d_H \left(  \ax_\lambda^\alpha(K') ,  \ax_\lambda^\alpha(K) \right) <  \frac{22}{3} \frac{R^2 }{\alpha^{\frac{1}{2}}  \tilde{\mu}'^{\frac{3}{2}} \lambda } \:  \epsilon^{\frac{1}{2}}. 
\end{equation}
 \end{proposition}
Note that, thanks to \cite{chazal2009sampling}, under the conditions of the proposition, that is for sufficiently small $\epsilon$,
  $\ax_\lambda^\alpha(K') $, $\ax(K)$ and $K^c$ have  same homotopy type (Theorem \ref{theorem:AxLambdaAlphaHasRightHomotopyType} below).

%{\color{red} either in should be and or I am not getting something} {\color{blue}  I do not understand your comment}

\subsubsection{Gromov-Hausdorff stability}

Lemma \ref{lemma:OpenSetWithGeodesicFiniteDiameter} and Theorem \ref{theorem:AXLambdaAlphaIsGeodesicAndFiniteDiameter}
give an explicit upper bound in the geodesic diameter of $\ax_\lambda^\alpha(K)$, assuming $K^c$ to be connected, as:
\begin{equation}\label{equation:PositiveMuReachGeodesicDimeter}
 \Gdiam( \ax_\lambda^\alpha(K)) 
\leq 2 \frac{R}{\tilde{\mu}^2} +  2 \alpha  \left( \left( \frac{4R}{ \alpha} \right)^n+1 + \frac{2}{\mu} \right)  e^{\frac{1}{\mu} + \frac{R}{\alpha \tilde{\mu}^2}}
\end{equation}

Thanks to \eqref{equation:StabilityRmu}, a similar bound holds for $\Gdiam( \ax_\lambda^\alpha(K')) $, for sufficiently small $\epsilon$. %when $\epsilon$ is small enough.

This  bound is exponential in $\frac{2 R}{\alpha \tilde{\mu}^2}$
and therefore increases quickly as $\alpha \rightarrow 0$. We do not know if this bound is close to be tight.\footnote{But it is seems likely to be pessimistic in practical situations.}

Lemmas \ref{lemma:AxLambdaAlphaGHStableWRTLambda}
and \ref{lemma:AxLambdaAlphaGHStableWRTAlpha} give:

\begin{proposition}\label{proposition:GHStabilityWRTLambdaMu}
For any $\lambda_{\min}>0$,
  the map $\lambda \mapsto \ax_\lambda^\alpha (K)$
is $\left(   \frac{R^2 (2\alpha_{\min}\lambda_{\min} + D}{\alpha_{\min}^2  \tilde{\mu}^2} \right)$-Lipschitz in the interval $[\lambda_{\min}, \lambda_{\max}]$
 for  Gromov-Hausdorff  distance, where \[D= \max_{\lambda \in [\lambda_{\min}, \lambda_{\max}]} \Gdiam( \ax_{\lambda}^\alpha(K)). \]
Similarly,  for  $\alpha_{\min}>0$,
 the map
$ %\[
\alpha \mapsto \ax_{\lambda}^\alpha (K)
$ %\]
is $\left(   \frac{R (2\alpha_{\min} + D}{\alpha_{\min}^2  \tilde{\mu}^2} \right)$-Lipschitz in the interval $[\alpha_{\min}, \alpha_{\max}]$
 for  Gromov-Hausdorff distance,  where $D=  \max_{\lambda \in [\alpha_{\min}, \alpha_{\max}]}  \Gdiam( \ax_\lambda^{\alpha}(K))$. 
 \end{proposition}
 
Note that the geodesic diameter enters as a factor in the Lipschitz constant. This is due to the fact that the Gromov-Hausdorff distance 
is defined as a global upper bound on differences of lengths,
while, here,  the metrics differ mainly by a multiplicative factor.
 In a sense, the metric discrepancy would be more tightly bounded by a mix of additive and multiplicative bounds, where Gromov-Hausdorff distances consider additive discrepancy only.
Replacing $D$ by its universal upper bound \eqref{equation:PositiveMuReachGeodesicDimeter}
is likely, in general, to give an overestimated Lipschitz constant with respect to the  one using the actual diameter 
  $D=  \max_{\lambda \in [\alpha_{\min}, \alpha_{\max}]}  \Gdiam( \ax_\lambda^{\alpha}(K))$.

Also, Theorem \ref{theorem:GromovHausdorffStability} gives us:
\begin{proposition}\label{proposition:PositiveMuReachGromovHausdorffHolder}
Denoting  $\tilde{\mu}' = \min( \mu', \sqrt{3}/2)$,
there is $\epsilon_{\max} >0$ depending only on $K$, such that, for, $\epsilon < \epsilon_{\max}$, one has:
\begin{equation}
d_{GH} \left(  \ax_\lambda^\alpha(K') ,  \ax_\lambda^\alpha(K) \right) < 2 \left( \frac{22}{3} \right)^{\frac{3}{2}}    \frac{R^3 (2\alpha_{\min} + D)}{ \alpha_{\min}^{\frac{7}{4}}   \tilde{\mu}'^{\frac{9}{4}}  \lambda^{\frac{3}{2}}  } \:  \epsilon^{\frac{1}{4}},
\end{equation}
where  $D= \max( \Gdiam( \ax_\lambda^\alpha (K)), \Gdiam( \ax_\lambda^\alpha (K')) )$.
 \end{proposition}
Again, taking for $D$ the upper bound \eqref{equation:PositiveMuReachGeodesicDimeter} allows a uniform bound which is enough in theory.

%{\color{red} Alternative below:}
%However, if one is willing to compute a practical 
%Gromov-Hausdorff approximation of $ \ax_\lambda^\alpha(K)$, since the diameter of $K$ is likely to be much smaller than the  bound  \eqref{equation:PositiveMuReachGeodesicDimeter}, one should try to control the geodesic diameter of $\ax_\lambda^\alpha (K')$. 
%
%For example, if $K'$ is finite (in fact the union of the complement of $\B^\circ(0,R)$ with a finite set),
% one would look for a bound on the geodesic diameter of the subset of the $(n-1)$-skeleton % {\color{blue}  why $d$ ?? we are in $\R^n$ !!}Sorry
% of the Voronoi diagram corresponding to $\ax_\lambda^\alpha (K')$.
% 
%{\color{red} Alternative:}
For a more practical bound on $d_{GH} \left(  \ax_\lambda^\alpha(K') ,  \ax_\lambda^\alpha(K) \right)$ it would be easier to calculate a bound on the  geodesic diameter of $\ax_\lambda^\alpha (K')$. For example, if $K'$ is finite (in fact the union of the complement of $\B^\circ(0,R)$ with a finite set) one could determine a bound on the geodesic diameter of the subset of the $(n-1)$-skeleton of part of the Voronoi diagram corresponding to $\ax_\lambda^\alpha (K')$.
 
%
%{\color{red} [Remove or exchange the following, it is already said below]} Also, the  approximation given by the proof of Theorem \ref{theorem:GromovHausdorffStability} 
%is in fact stronger than what is expressed by Propositions  
%\ref{proposition:GHStabilityWRTLambdaMu} and \ref{proposition:PositiveMuReachGromovHausdorffHolder},
%since the relation associated to the Gromov-Hausdorff distance is related to the ambient metric, in particular if two points are in relation their 
%ambient Euclidean distance is in $O \big( \epsilon^{\frac{1}{2}} \big)$.
%

%{\color{red} [Remove or exchange the following, it is already said below]} Also, the  approximation given by the proof of Theorem \ref{theorem:GromovHausdorffStability} 
%is in fact stronger than what is expressed by Propositions  
%\ref{proposition:GHStabilityWRTLambdaMu} and \ref{proposition:PositiveMuReachGromovHausdorffHolder},
%since the relation associated to the Gromov-Hausdorff distance is related to the ambient metric, in particular if two points are in relation their 
%ambient Euclidean distance is in $O \big( \epsilon^{\frac{1}{2}} \big)$.
%{\color{red}
\subsection{Method} 
%[I don't see a better place, but this does not mean that there is no better place]
All proofs in the paper are based on the flow of the (generalized) gradient of the distance function from a point $x$ to $K$, see Section \ref{section:MedialAxisAndFlow} for a formal definition. { The flow has been used before, among others to establish the following results: 
\begin{itemize} 
\item The medial axis has the same homotopy type as the set \cite{LIEUTIERhomotopytype}. 
\item The topologically guaranteed reconstruction for non-smooth sets \cite{chazal2009sampling}. 
\end{itemize} The flow also plays a central role in the work on the $\lambda$-medial axis \cite{cl2005lambda}.
These tools were developed for %are convenient for addressing
 non-smooth objects, %by using the much 
and rely on the weak %er 
regularity properties based on the $\mu$-reach and the critical function (Section \ref{section:MedialAxisAndFlow}).
Our stability results rely on the stability of the flow and its gradient under Hausdorff perturbation of $K$, and by quantifying how quickly we enter the $(\lambda,\alpha)$-medial axis following the flow of the gradient, assuming that we start not too far from the $(\lambda,\alpha)$-medial axis. 
}
 %Roughly speaking the gradient is stable far away from the $(\lambda,\alpha)$-medial axis, while we flow quickly into the $(\lambda,\alpha)$-medial axis if we are close to it. 
%}

\section{Future work}
%{\color{red}[MW all reviewers wanted this in the main part]}
Beyond the stability properties presented in this paper, several  questions remain open.
We do not know if our moduli of continuity are optimal, or if other filtrations could offer better H{\"o}lder exponents for the stability.
More precisely, because the dependence of the ($\lambda$,$\alpha$)-medial axis on $\lambda$ and $\alpha$ is Lipschitz, it is only the Lipschitz constant that can be improved. 
%the corresponding modulus of continuity can only be improved up to some constant.
This contrasts with the H{\"o}lder exponents for the map $K \mapsto \ax_\lambda^\alpha$,  namely $\frac{1}{2}$ for the Hausdorff distance and $\frac{1}{4}$
 the Gromov-Hausdorff distance, which may not be optimal.

Our stability property expressed in term of Gromov-Hausdorff distance hides a stronger statement. 
Indeed the Gromov-Hausdorff distance applies to two independent metric spaces, while our two metric spaces are also subset of a same Euclidean space.
While this has not been made explicit in the statement of Theorem \ref{theorem:GromovHausdorffStability}, when  $d_H( K,  K') < \epsilon$,
 \eqref{equation:RelationForGromovHausdorff} gives a $\mathcal{O}(\epsilon^\frac{1}{2})$ bound  on the ambient Euclidean distance between points pairs in relation
that upper bounds the $\mathcal{O}(\epsilon^\frac{1}{4})$ Gromov-Hausdorff distance.
For example, in Figure \ref{figure:FrechetVersusGH} on the right, a simple rotation could define a relation giving a zero
Gromov-Hausdorff distance (an isometry), while in fact our construction defines another relation for which points in relation are much closer in ambient space.
In order to fully express our stability properties induced by the flow, we should  introduce in a future work a sharpening of the Gromov-Hausdorff distance, where the relation 
realizes not only a small geodesic metric distortion, but also a small ambient displacement.

\section*{\textit{Part II: The technical statements and proofs}}\otherlabel{part:TechnicalStatementsAndProofs}{Part II}

\section{Definitions and previous work}\label{section:DefinitionsPreviousWork}

In this section, we recall some definitions and results, mainly introduced in \cite{LIEUTIERhomotopytype, cl2005lambda, chazal2009sampling},
in order to make the paper self-contained. %We refer to the appendix for a recap of some essential results from analysis, the definition of the Gromov-Hausdorff distance, and the topology of the $\lambda$-medial axis. 

Throughout the paper, $K$ will be a \textit{closed subset} of $\R^n$, whose \textit{complement $K^c$ is bounded and connected}.{The fact that $K^c$ is connected is essential for the medial axis to be patch connected and thus is needed for a bound on the geodesic diameter. The geodesic diameter in turn is needed for the bound on the Gromov-Hausdorff distance (our bound is proportional to the geodesic diameter). } { We stress that the connectedness assumption is not required for our results on the Hausdorff distance.  } 

%{\color{blue} Should we say instead that the connectedness of $K^c$ is required each time one consider the Hausdorff distance..?} 
%{\color{red} I did something here because of complaints by one of the reviewers, it may just be ok as it was.}

%We consider all over the paper a \textit {closed subset} $K$ of the Euclidean space $\R^n$
 %for which  we assume that its \textit {complement $K^c$ is bounded and connected}.
%We note that the boundary $\partial K$ of $K$ is compact. 

%%%%%%%%%%%%%%%%%%%%%%%%%%%%%%%%%%%%
%%%%%%%%%%%%%%%%%%%%%%%%%%%%%%%%%%%%

%%%%%%%%%%%%%%%%%%%%%%%%%%%%%%%%%%%%
{%\color{blue}
\subsection{Homotopy equivalence and weak deformation retraction}\label{section:HomotopyEquivalenceDefinition}
%The general definition of 
{A } homotopy equivalence between topological spaces $X$ and $Y$ 
%in the particular case when 
{ where} $Y \subset X$ and the map from $Y$ to $X$ is the inclusion map,
is %{\color{red} \sout{sometimes}} 
called  {\it weak deformation retract}, more formally:
\begin{definition}[weak deformation retract]\label{definition:HomotopyEquivalenceDefinition}
If $Y\subset X$ and there exists a continuous map $H:[0,1]\times X \rightarrow X$ such that:
\begin{itemize}
\item $\forall x \in X,\: H(0,x) = x$,
\item $\forall x \in X, \:  H(1,x) \in Y$,
\item $\forall y \in Y, \forall t \in [0,1], \: H(t,y) \in Y$,
\end{itemize}
then we say that $H$ is a weak deformation retract of $X$ on $Y$.
 In particular $X$ and $Y$ have same homotopy type.
\end{definition}
}

%%%%%%%%%%%%%%%%%%%%%%%%%%%%%%%%%%%%
\subsection{The medial axis and associated flow}\label{section:MedialAxisAndFlow}

We define the following functions on $\R^n$ associated to $K$: %{\color{blue} I was wondering if you used $R_K$ before, because to me $D_K$ would be more natural.  } 
\begin{align}
x \mapsto R_K (x) &\defunder{=} d(x, K) = \min_{y\in K} d(x, y) 
\label{equation:DefinitionR}
\\
x \mapsto \Theta_K (x) & \defunder{=} \left\{ y \in K \mid d(x, y) = R_K (x) \right\}.\label{equation:DefinitionTheta}
\end{align}
For a bounded set $S$ we denote the center of the smallest ball enclosing $S$ by $\cent(S)$, we write $\radius(S)$ for the radius of this ball. 
%, respectively, denote the center and the radius of the smallest ball enclosing $S$.

We further define,
\begin{equation}\label{equation:DefinitionF}
\F_K (x) \defunder{=} \radius(\Theta_K(x))
\end{equation}
and for $x \notin K$:
\begin{equation}\label{equation:DefinitionNabla}
\nabla_K (x) \defunder{=} \frac{ x - \cent(\Theta_K(x))}{R_K(x)}
\end{equation}
The medial axis $\ax (K)$ of $K$ is defined as:
\begin{equation}\label{equation:DefinitionMedialAxis}
\ax (K) \defunder{=} \{ x \in \R^n \mid \F_K (x) >0 \} = \{x \in \R^n\mid \sharp \left( \Theta_K (x) \right) \geq 2 \}
\end{equation}
where $\sharp X$ denotes the cardinality of the set $X$.

 \begin{figure}[h!]
\centering
\includegraphics[width=.5\textwidth]{%pictures/
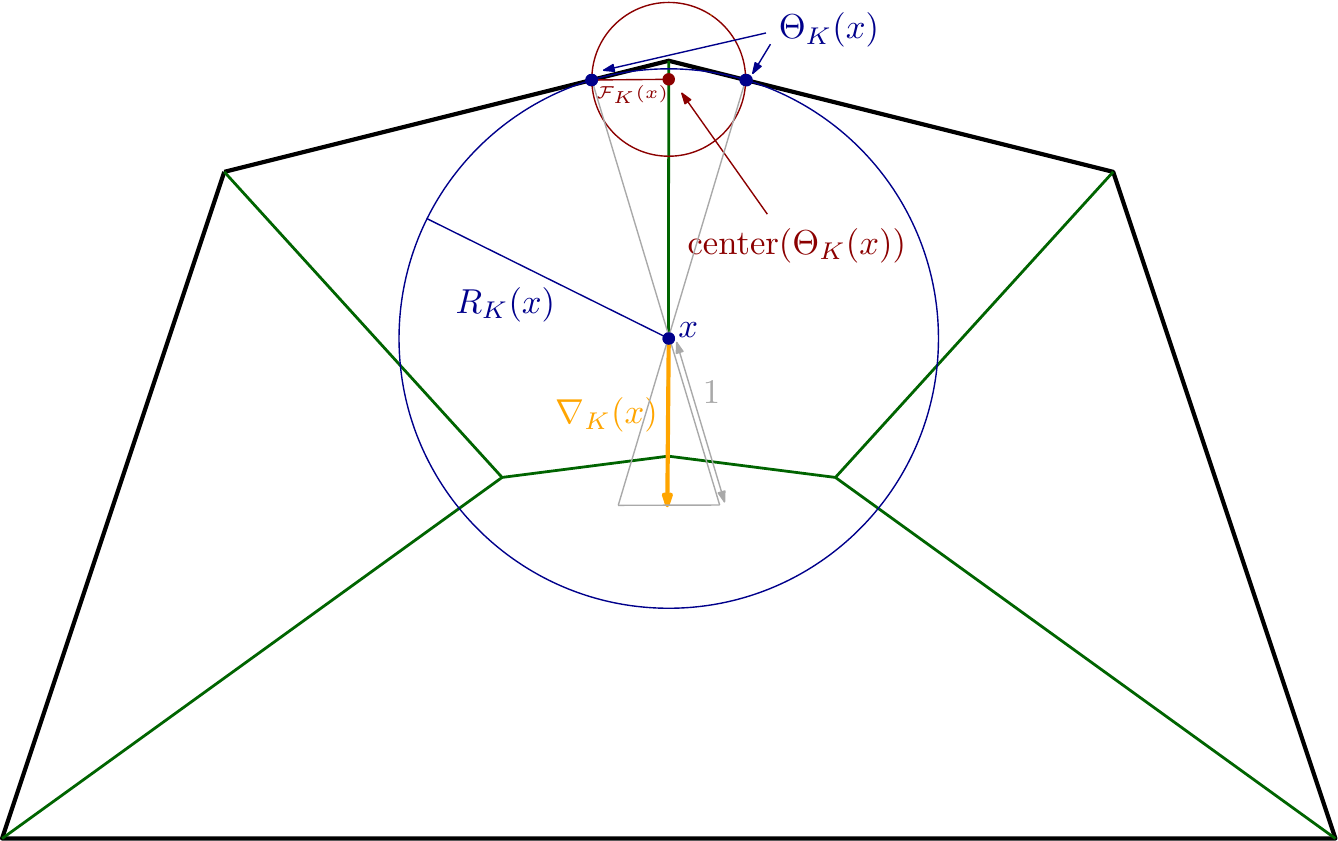}
\caption{Pictorial overview of the definitions and notation. The set $K$ is indicated in black and the medial axis in green. 
}
\label{figure:DefinitionsNotations81}
\Description{A pentagon and its medial axis with the notation indicated.}
\end{figure}

{When $x\in K$, one has $\Theta_K(x) = \{x\}$ and $\F_K (x) =0$.
 It follows that  $\ax (K) \subset K^c$.}

When $x \notin K \cup \ax (K)$, $\nabla_K (x)$ 
coincides with the gradient of the Lipschitz function $x \mapsto R_K(x)$.

For any $x \in K^c$, including when  $x \in \ax (K)$,  $\nabla_K (x)$ 
%\xout{ is the projection of $0$ on the Clarke gradient of  $x \mapsto R_K(x)$.}
could have been equivalently defined as the projection of $0$ on the Clarke gradient, see \cite{Clarke1990}, of  $x \mapsto R_K(x)$. 
However the definition \eqref{equation:DefinitionNabla}  is more convenient in our context and does not require the formal introduction of Clarke gradient, which is technical.
%{\color{blue} \xout{ where is the term Clarke gradient used first? Also cite \cite{Clarke1990}? }} 
For this reason, we call $\nabla_K (x)$ the \textit{generalized gradient}  of $x \mapsto R_K(x)$.
Thanks to the definition of $\nabla_K (x)$, and Pythagoras, see that  for  $x \notin K$ one has: 
\begin{equation}\label{equation:NormGradientRelatedToFandR}
\forall x \in K^c,   \quad \left\|\nabla_K (x) \right\| ^2 = 1 - \left(\frac{ \F_K(x)}{R_K(x)} \right)^2.
\end{equation}
In \cite{LIEUTIERhomotopytype} we have seen that there exists a locally Lipschitz, and therefore continuous, flow $\Phi_K: {\R_{\geq 0} } \times K^c \rightarrow K^c$ 
%{\color{red} [there was a complaint about the notation $\R^+$, perhaps $\R_{\geq 0}$ would make sense]} 
such that:
\begin{align}
 \forall x &\in K^c,  & \Phi_K (0,x ) &= x  
\nonumber 
\\
\forall x &\in K^c \: , \forall  t \geq 0, & \frac{d}{dt^+} \Phi_K (t,x ) &= \nabla_K \left( \Phi_K (t,x ) \right) 
\label{equation:DefinitionFlow}
\\
\forall x &\in K^c,\: \forall  t_1, t_2 \geq 0, &  \Phi_K \left(t_1, \Phi_K (t_2, \, x) \right) &=  \Phi_K (t_1 + t_2, \,  x).
\label{equation:FlowAsSemiGroup}
\end{align}

 \begin{figure}[h!]
\centering
\includegraphics[width=.5\textwidth]{%pictures/
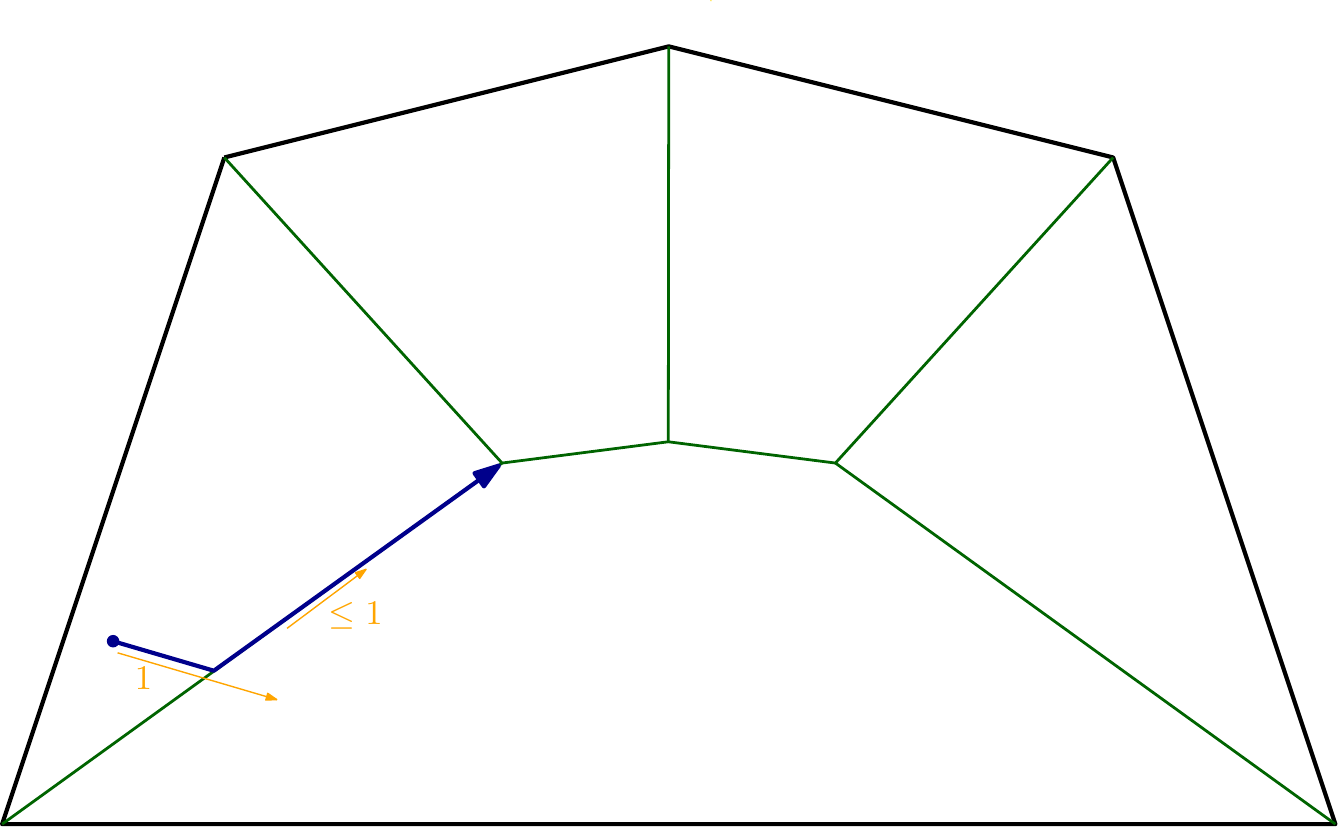}
\caption{The blue path follows the flow. The orange vectors indicate the vectors ($\nabla_K(x)$) whose flow $\Phi_K$ follows. They have been shifted slightly for visibility and lengths have been indicated. 
}
\Description{The same pentagon and medial axis as before with the flow of $\nabla_K(x)$.}
\label{figure:FlowExample}
\end{figure}	

It was established in \cite{LIEUTIERhomotopytype} that $\ax (K)$ has same homotopy type as $K^c$. 
This result is based on the fact that
the flow $\Phi_K$ realizes the homotopy equivalence 
  $K^c \rightarrow \ax (K)$ for a finite $t$ (see Definition \ref{definition:HomotopyEquivalenceDefinition}).
In the particular case when $K$ is a finite set, the flow $\Phi_K$ is equivalent to the flow that induces the
{\it flow complex} \cite{giesen2008flow} of $K$.

We further recall the following:
\begin{lemma}[Lemma 4.16 of \cite{LIEUTIERhomotopytype}] 
%(Lemma 4.16 in \cite{LIEUTIERhomotopytype}):
\begin{align}
\forall t &\geq 0, \forall x \in K^c,  & \frac{d}{dt^+} R_K( \Phi_K (t,x )) & = \left\| \nabla_K \left( \Phi_K (t,x ) \right) \right\|^2.
\label{equation:DerivativeR}
\end{align}
\end{lemma} 

\begin{lemma}[Corollary 4.7 of \cite{LIEUTIERhomotopytype}] \label{lemma:FIsUpperSemiContinuous}
The map $ \F_K$ is upper semi-continuous.
\end{lemma} 

\begin{lemma}[Lemma 4.17 of \cite{LIEUTIERhomotopytype}] \label{lemma:FIsINotDecreasing}
%(Lemma 4.17 in \cite{LIEUTIERhomotopytype}):
The map $ %\begin{equation}%\label{equation:FIsINotDecreasing}
t \mapsto \F_K( \Phi_K (t,x )) %\: \text{ is non-decreasing}
$ %\end{equation}
is non-decreasing (i.e. increasing but not necessarily strictly increasing)
and therefore, by Lemma \ref{lemma:FIsUpperSemiContinuous} %it is 
 right-continuous.

%By non-decreasing we mean increasing but may be not necessarily strictly increasing.
\end{lemma} 
Lemma 4.13 of \cite{LIEUTIERhomotopytype} immediately yields:
\begin{corollary} 
\begin{equation}
R_K(x_1), R_K(x_2) \geq \alpha \Rightarrow \|  \Phi_K (t,x_2 ) -  \Phi_K (t,x_1 ) \| \leq \| x_2 - x_1 \| \, e^{\frac{t}{\alpha}}.
\label{equation:BoundOnFlowExpansion}
\end{equation}
\end{corollary}

%%%%%%%%%%%%%%%%%%%%%%%%%%%%%%%%%%%%
\subsection{Critical function, \texorpdfstring{$\mu$}{mu}-reach and Weak Feature Size}\label{section:CriticalFunctionAndMuReach}

%{\color{red}[Do we want to add extra figures here?] }

%{\color{red} [The following is interesting, but it takes a lot of space (including the figure) it does not sell the paper, it also fits better with Section 6]} 
 
{ As we have seen in Section \ref{section:introCriticalLambdaMedialAxis}, the critical function $ \chi_K: (0 ,\infty) \to [0,1]$, is defined as 
\begin{equation} \tag{\ref{equation:CriticalFunctionFirstDefinition}}
 \chi_K (t) \defunder{=}  \inf_{R_K(x) = t} \sqrt{ 1 - \left( \frac{\F_K (x) }{ R_K (x)}\right)^2}.
\end{equation}
%{\color{red} 
In this setting we say that the infimum over the empty set yields $1$.% }
We'll now discuss the intuition behind this function.}  {%\color{red} [MW we have not done this before and perhaps there are some good reasons, but I would be tempted to add something like:] 
The critical function provides us in a certain sense with some lower bound on the norm of the vector field $\nabla_K (x)$ whose flow ($\Phi_K$) we follow.}
Figure \ref{figure:CriticalFunction} illustrates the critical function when $K$ is a { (hollow)}  square in $\R^3$. %{\color{red} [Question of opinion, are squares usually filled in or not?] } 
The medial axis of the square is an infinite prism, which is the product of the squares's diagonals and %whose basis are the  squares's diagonals and 
 the line orthogonal to the square supporting plane. The { infimum} %$\inf$ 
in \eqref{equation:CriticalFunctionFirstDefinition}
 is  first  { attained} along the square diagonals, where $\chi_K(t)= 1/\sqrt{2}$ since there $\F_K (x) = \frac{R_K(x)}{\sqrt{2}}$. Then when the offset reaches
 the square center we get $\F_K (x) = R_K(x)$ and therefore $\chi_K(t)= 0$. The topology of the offset then changes and,
 after this critical value,  the $\inf$ is then reached on the line through the square center and orthogonal to the square supporting plane.
 
The critical function enjoys some stability properties with respect to Hausdorff distance perturbation  \cite[Theorem 4.2]{chazal2009sampling}, illustrated in Figure \ref{figure:CriticalFunction}.
If two sets $K$ and $K'$ are close enough in Hausdorff distance, their critical functions $t\mapsto \chi_K(t)$ and  $t\mapsto \chi_{K'}(t)$ are close to each other, for $t$ large enough.

 \begin{figure}[ht!]
\centering
\includegraphics[width=0.5\textwidth]{%pictures/
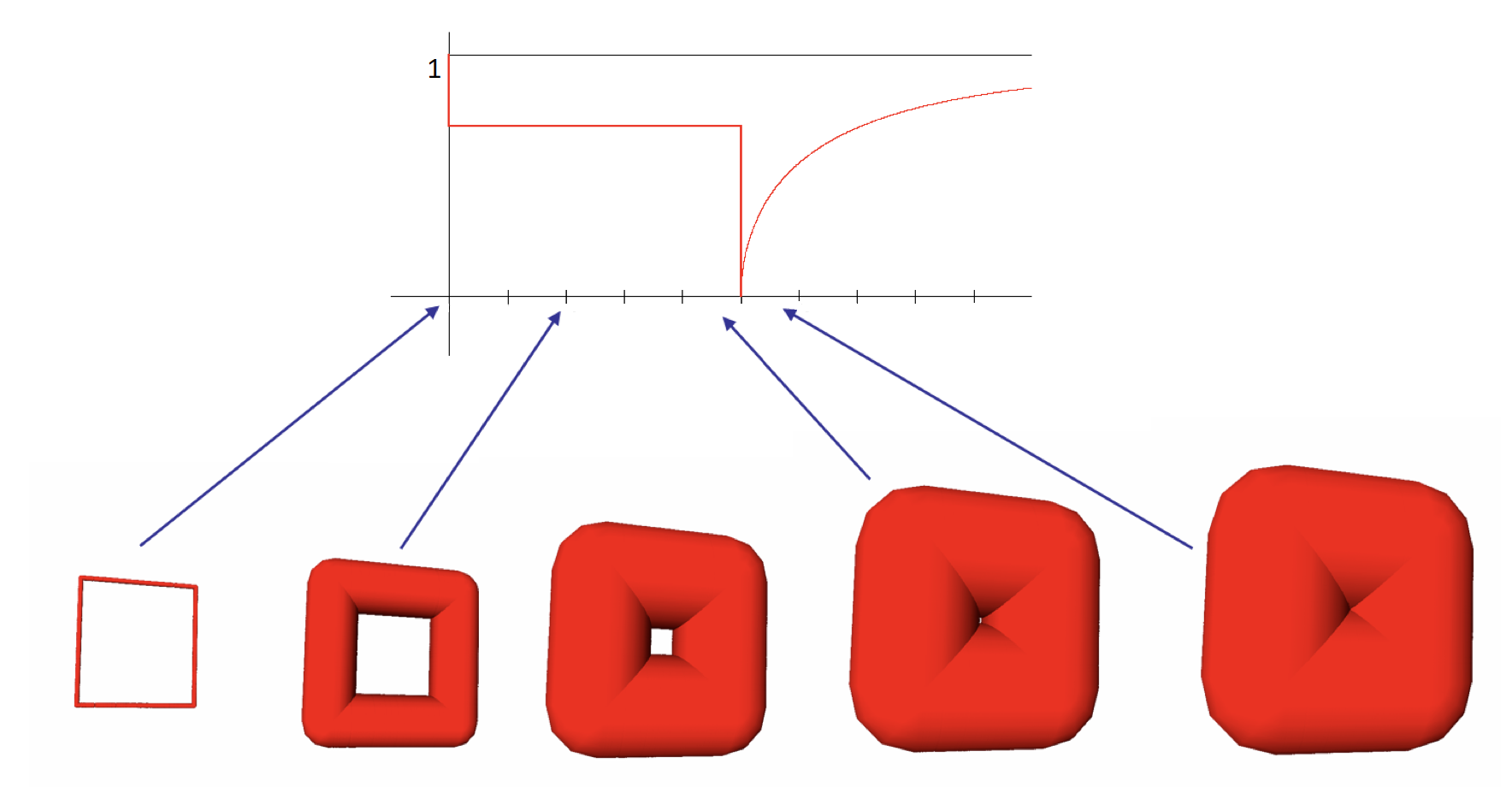}
\includegraphics[width=0.225\textwidth]{%pictures/
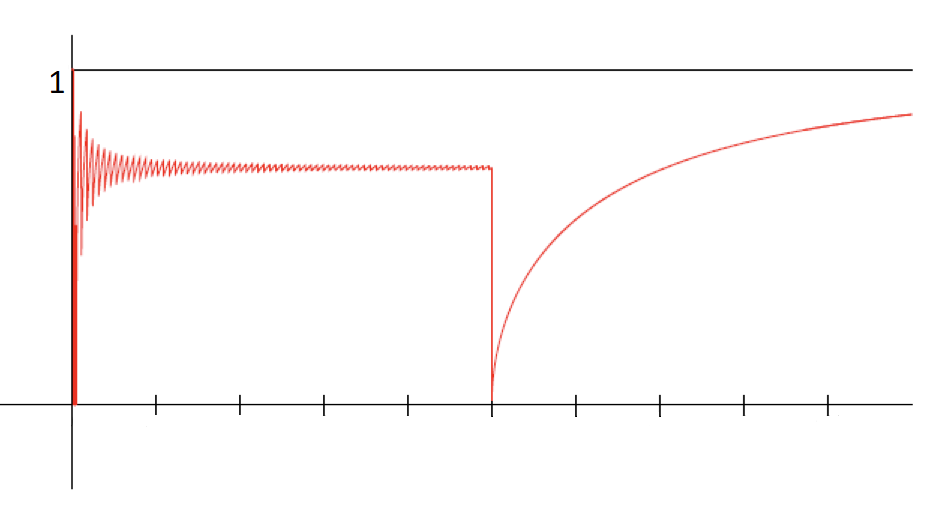}
\caption{
(Adapted  from \cite{boissonnat2018geometric})).
 On the top, the critical function $\chi_K$ of a square $K$ in $\R^3$ together with the corresponding level sets  $R_K^{-1}(t)$ of the distance to $K$,
on which the $\inf$ is taken in equation \eqref{equation:CriticalFunctionFirstDefinition}.
The topology of offsets changes only when $\chi_K(t)= 0$, that is when $t=t_{crit.}$  which is the  half of the square side. 
For $t\in (0,t_{crit.})$, $\chi_K(t)= 1/\sqrt{2}$ since, for $0<t<t_{crit.}$, the
$\inf$ in \eqref{equation:CriticalFunctionFirstDefinition} is {For $t\in (0,t_{crit.})$ we have that $\chi_K(t)= 1/\sqrt{2}$. This can be seen as follows: 
For $0<t<t_{crit.}$ the $\inf$ in \eqref{equation:CriticalFunctionFirstDefinition} is attained on the intersection of $R_K^{-1}(t)$, the  supporting plane of the square, and the medial axis. This intersection equals the diagonals of the square. It follows immediately from Pythagoras that $\F_K (x) = \frac{R_K(x)}{\sqrt{2}}$.
} 
This value is indeed strictly smaller  than $1$ as the reach of the square is $0$.  \\
At the bottom, the critical function of a set $K'$ (here a finite set), close, in Hausdorff distance, to $K$. 
For large enough offset $t$, $\chi_{K'}$ is close to $\chi_K$. 
In particular, if $d_H(K',K)$ is small enough with respect to  $t_{crit.}$, 
Theorem  \cite[Theorem 4.2]{chazal2009sampling} provides a lower bound on the critical function of $K'$,
which is  then guaranteed to not vanish on some interval subset of $(0, t_{crit.})$.   }
\Description{The offset of a square (not filled in) and the critical function $\chi_K$.}
\label{figure:CriticalFunction}  
\end{figure}

The critical function has been introduced in Section \ref{section:introCriticalLambdaMedialAxis} and illustrated in Figure \ref{figure:CriticalFunction}.
In \cite{chazal2009sampling} the critical function $\chi_K : (0 , \infty)  \rightarrow [0,1]$ was defined for a compact set $K$.
We adapt it to a closed set $K$ with bounded complement.

On a closed set $K$ with bounded complement the function $R_K$ attains a maximum:
\begin{equation}\label{equation:DefinitionRMax}
R_{\max} (K) \defunder{=} \sup_{{x \in K^{c}}}  R_K(x) 
= \max_{x \in K^{c}}  
R_K(x). 
\end{equation}

The critical function of $K$ is the function $\chi_K :  (0, R_{\max} (K)] \rightarrow [0,1]$ defined as:
\[
\chi_K (t) \defunder{=}  \inf_{R_K(x) = t} \| \nabla_K(x) \|.
\]
As a consequence of \eqref{equation:DerivativeR}, one has $\chi_K ( R_{\max}) = 0$.
For $\mu  \in (0,1]$, the $\mu$-reach of $K$, denoted $r_\mu (K)$ is  defined as:
\begin{align*}
r_\mu (K) &\defunder{=} \inf \{t  \mid \chi_K (t) < \mu \} \\&
 = \, \sup \left\{ r \mid R_K(x) < r \Rightarrow \| \nabla_K(x) \| \geq \mu \right\}.
\end{align*}
For $\mu=1$, $r_1(K)$ is also known as the reach of $K$. % introduced by Federer \cite{Federer}.
The \textit{Weak Feature Size} of $K$, denoted $\wfs(K)$ is defined as the first critical value of the distance function, that is the first value at which $\chi_K $ vanishes
\begin{align} 
\wfs(K) \defunder{=} \inf \{t \mid \chi_K (t) = 0 \}. 
\label{eq:DefWeakFeatureSize} 
\end{align} 
Since $x \mapsto \F_K(x)$ is upper semi-continuous, see Corollary 4.7 of \cite{LIEUTIERhomotopytype}, $\chi_K$ is lower semi-continuous.
%{\color{red} [OLD Alternative below:]} It follows that, if  $0<\wfs(K)$, then $\chi_K( \wfs(K) ) = 0$, as if we had $\chi_K( \wfs(K) ) > 0$, the lower semi-continuity of $\chi_K$ 
%would then be in contradiction with the definition of $\wfs$, so that
%\[
%\wfs(K) >0 \Rightarrow \wfs(K) = \min \{t \mid \chi_K (t) = 0 \} 
%\]
%also, one has by definition of $r_\mu (K)$ and lower semi-continuity of $\chi_K$ that, for any $\mu>0$,  $r_\mu (K) < \wfs(K)$
%and, again by the lower semi-continuity of $\chi_K$,  if $0< \sup \{ r_\mu (K) \mid \mu>0 \}  < \infty$, then $\chi_K(  \sup \{ r_\mu (K) \mid \mu>0 \} ) = 0$.
%Since $\wfs (K)>0 \Rightarrow  \sup \{ r_\mu (K) \mid \mu>0 \} >0$, it follows that:
%\begin{equation}\label{equation:WfsAndMuReach}
%\wfs(K) >0 \Rightarrow \wfs(K) = \sup \{ r_\mu (K) \mid \mu>0 \} .
%\end{equation}
%{\color{red} [Alternative: (please check carefully)]}
The lower semi-continuity of $\chi_K$ has a number of consequences:
\begin{itemize}
\item If $0<\wfs(K)$, then we have $\chi_K( \wfs(K) ) = 0$, so that 
\[
\wfs(K) >0 \Rightarrow \wfs(K) = \min \{t \mid \chi_K (t) = 0 \}. 
\] 
This follows by combining the lower semi-continuity with the definition of the weak feature size \eqref{eq:DefWeakFeatureSize}.
\item By definition of  $r_\mu (K)$ (and lower semi-continuity of $\chi_K$) we have that, for any $\mu>0$,  $r_\mu (K) < \wfs(K)$. 
\item  If $0< \sup \{ r_\mu (K) \mid \mu>0 \}  < \infty$, then $\chi_K(  \sup \{ r_\mu (K) \mid \mu>0 \} ) = 0$.
Because $\wfs (K)>0 \Rightarrow  \sup \{ r_\mu (K) \mid \mu>0 \} >0$, it follows that:
\begin{equation}\label{equation:WfsAndMuReach}
\wfs(K) >0 \Rightarrow \wfs(K) = \sup \{ r_\mu (K) \mid \mu>0 \} .
\end{equation}
\end{itemize}

\subsection{Gromov-Hausdorff distance}\label{section:GromovHausdorffDistance}
For $\alpha \geq 0$, the $\alpha$-offset of $K$ denoted $K^{\oplus \alpha}$
is the set of points lying at distance at most $\alpha$ from $K$:
\begin{align}
K^{\oplus \alpha}  \defunder{=} \{ x \in \R^n \mid d(x, K) \geq \alpha \} ,
\label{eqdef:Offset_Minkowski}
\end{align}
where $d(x,K) \defunder{=} \inf_{y\in K} d(x,y) = \min_{y\in K} d(x,y)$.
Alternatively, $K^{\oplus \alpha}$ can be defined as the Minkowski sum of $K$ 
with a ball of radius $\alpha$ centered at $0$, which motivates the $\oplus$ notation.

The Hausdorff distance between two compact subsets $A,B$ of the Euclidean  space is defined as (\cite[Section 5.30]{bridson2013metric}):
\begin{definition}[Hausdorff distance]\label{definition:HausdorffDistance}
The Hausdorff distance between closed sets $A$ and $B$ is defined as
\begin{equation}\label{equation:DefinitionHausdorffDistance}
d_H( A,B) = \inf \left\{ \varepsilon, \: A \subset B^{\oplus \varepsilon}\: \mathrm{and} \:  B \subset A^{\oplus \varepsilon}  \right\}.
\end{equation}
\end{definition}

One trivially has,
\begin{equation}\label{equation:HausdorffDistanceForNestedSets}
A \subset B  \quad \mathrm{and} \quad  B \subset A^{\oplus \varepsilon} \quad \Rightarrow \quad d_H(A,B) \leq \varepsilon.
\end{equation}

We now define the Gromov-Hausdorff distance between two metric spaces $X_1$ and $X_2$. We follow \cite[Section 5.33]{bridson2013metric}, with minor modifications in the formulation for compatibility.
%A subset $S$ of a metric space $X$ is said to be $\varepsilon$-dense if every point of $X$ is at distance 
%at most $\varepsilon$ from some point in $S$: $\forall x  \in x, \exists s\in S, d(x,s) \leq \varepsilon$.
\begin{definition}[Gromov-Hausdorff distance]\label{definition:GromovHausdorffDistance}
An $\varepsilon$-relation between two metric spaces $X_1$ and  and $X_2$ is a subset $\mathcal{R} \subset X_1 \times X_2$ such that:
\begin{enumerate}
\item[(1)] For $i=1,2$, the projection of $\mathcal{R}$ to $X_i$ is  surjective.
\item[(2)] If $(x_1, x_2), (x_1', x_2')   \in \mathcal{R}$ then,
\[
 \left|  d_{X_1} (x_1, x_1') - d_{X_2} (x_2, x_2') \right| < \varepsilon.
 \]
\end{enumerate}
If there exists an $\varepsilon$-relation between  $X_1$ and  and $X_2$ then we write $X_1 \simeq_{\varepsilon} X_2$.
We define the Gromov-Hausdorff distance between $X_1$ and $X_2$ to be
\begin{equation}\label{equation:DefinitionGromovHausdorffDistance}
d_{\textrm{GH}} (X_1, X_2) = \inf  \left\{ \varepsilon, \: X_1 \simeq_{\varepsilon} X_2 \right\}.
\end{equation}
\end{definition}

\subsection{The topology of the \texorpdfstring{$\lambda$}{lambda}-medial axis}
In \cite{cl2005lambda} the $\lambda$-medial axis was introduced for $\lambda >0$ as:
\[
\ax_\lambda(K) \defunder{=} \{ x \in \R^n \mid \F_K (x) \geq \lambda \}.
\]
By definition the $\lambda$-medial axis has the following properties:
\[
\lambda_1  \geq \lambda_2 \Rightarrow \ax_{\lambda_1}(K)  \subset \ax_{\lambda_2}(K),
\]
and
\[
 \bigcup_{\lambda >0} \ax_\lambda(K) = \ax(K).
\]
It follows from Lemma \ref{lemma:FIsUpperSemiContinuous} that:
\begin{lemma}\label{lemma:LmabdaMedialAxisIsCompact}
$\ax_\lambda(K)$ is closed  and therefore, since $K^c$ is bounded, it is compact.
\end{lemma}

As a consequence of Lemma \ref{lemma:FIsINotDecreasing}, 
using Definition \ref{definition:HomotopyEquivalenceDefinition},
the flow $\Phi_K$ can be used to show that: 
\begin{theorem}[Theorem 2 of \cite{cl2005lambda}]
If $\lambda  < \wfs(K)$, then $\ax_\lambda (K)$ has the homotopy type of $K^c$.
\end{theorem}

\subsection{Fundamental Theorem of calculus for Lipschitz functions}\label{section:FundamentalThmCalculus}

As in \cite{LIEUTIERhomotopytype, cl2005lambda} we apply  the fundamental theorem of
calculus to Lipschitz functions while it is usually stated in the context of differentiable functions. 

We recall that it follows trivially from the definitions that Lipschitz functions are in particular absolutely continuous and:
\begin{theorem}[Adapted from Theorem 6.4.2 in \cite{heil2019introduction}]\label{theorem:FundamentalThmCalculusACFunctions}
If $f:[a,b] \rightarrow \R$, then the following two statements are equivalent.
\begin{itemize}
\item[(a)] $f$ is absolutely continuous.
\item[(b)] $f$ is differentiable almost everywhere on $[a,b]$, $f'\in L^1[a,b]$, and
\[
f(x) - f(a) = \int_a^x f'(t) dt
\]
\end{itemize}
\end{theorem}

\subsection{Volterra integral inequalities} \label{sec:Volterra} 

Apart from a generalized fundamental theorem of calculus we'll also need a estimates on Volterra integrals. These are related to the solution of differential equations of \v{C}aplygin type. The simplest version of the result we recall first, see for example \cite[Theorem 2, Section 2, Chapter XI]{mitrinovic1991inequalities},   
\begin{theorem}[\v{C}aplygin]
Let $F(x,t)$ be a Lipschitz function and $x=x(t)$ a differentiable function such that $x(0) =x_0$, and
\[ 
\frac{d}{dt} x(t) \leq F(t,x(t) ) ,
\] 
then $x(t) \leq y(t)$, where $y(t)$ is the solution of the initial value problem
$y(0) = x_0$ and $\frac{d}{dt} y(t)  = F(t,y(t) )$. 
\end{theorem} 
We note that even though this result is ascribed to \v{C}aplygin the result was already known to Peano, we refer to \cite[page 316]{mitrinovic1991inequalities} and the reference mentioned there for more information. 

However because we deal with functions that are not differentiable we need an integral version of the statement. For this we'll adapt a number of definitions and results on Volterra integral inequalities from \cite[Chapter I]{walter2012differential}, see also \cite{Alekseev, Perov}.
We first make the following definitions: 
Integral equations of the form
\begin{align}
x(t) = g(t) +\int_0^t k(x(\tau)) \ud \tau  
\nonumber
\end{align}
are called Volterra integral equations and $k(x)$ its kernel. In general these kernels are allowed to depend on $t$, but we don't need this in our context. The class $Z_c(k)$ of admissible function for the kernel $k$ are the functions $\phi(\tau)$ such that $k(\phi(\tau))$ exists and is integrable. We say that the kernel is monotone increasing if $k(x) \leq k(\bar{x})$ for all $x \leq \bar{x}$ in its domain, and strictly monotone if this still holds when both bounds are replaced by strict inequalities. 
We have,
\begin{theorem}\label{VolterraInequality} 
Suppose that $k(x)$ is a monotone increasing kernel, $x(t),y(t) \in Z_c(K)$, and $C_1$ a constant. Further assume that  
\begin{align}
y(t) &= C_1 + \int_0^t k (y(\tau) ) \ud \tau 
\nonumber
\\
x(t) &\geq C_1+ \int_0^t k (x(\tau) ) \ud \tau, 
\nonumber
\end{align} 
for all $t$ in the domain, where equality in the second equation only occurs for $t=0$. Moreover assume that there exists a $\delta'>0$ such that for all $t' \in (0, \delta')$, we have $y(t') < x(t ')$. 
 Then for all $t$ in the domain,
\[ 
y(t) \leq x(t),
\] 
where equality occurs only for $t=0$. 
\end{theorem}
\begin{proof}
For any $t' \in (0, \delta')$, the result follows from the hypothesis. If the assertion would be false, there would be a $t_0>0$ such that $x(t_0) =y(t_0)$. However because $k$ is assumed to be monotone,
\begin{align}
y(t_0) &= C_1 + \int_0^{t_0}  k (y(\tau) ) \ud \tau \leq   C_1+ \int_0^{t_0} k (x(\tau) ) \ud \tau < x(t_0).
\nonumber
\end{align} 
The result now follows. 
\end{proof}

%%%%%%%%%%%%%%%%%%%%%%%%%%%%%%%%%%%%

%%%%%%%%%%%%%%%%%%%%%%%%%%%%%%%%%%%%
\section{The \texorpdfstring{$(\lambda,\alpha)$}{(lambda, alpha)}-medial axis}\label{section:LambdaAlphaMedialAxis}
In this section, we define the $(\lambda,\alpha)$-medial axis and prove its stability with respect to both $\lambda$ and $\alpha$.

\subsection{The definition}

It is easy to check that the subset of $\ax (K)$ that lies \ at distance greater than $\alpha$ from $K$ coincides with the medial axis  $\ax (K^{\oplus \alpha})$ 
of the $\alpha$-offset of $K$,
\[
\ax (K^{\oplus \alpha}) = \{ x \in \ax(K)\mid R_K(x) > \alpha \}.
\]
Moreover, we observe that 
\begin{equation}\label{equation:ROffsetK}
 R_K(x) > \alpha \Rightarrow  R_{K^{\oplus \alpha} } (x) = d \left( x, K^{\oplus \alpha} \right) = d(x, K) - \alpha = R_K(x) - \alpha.
\end{equation}
We further note that when $y \in  \Theta_K(x)$ one has ${\|y-x\|} = R_K(x)$, by definition of $ \Theta_K$, see \eqref{equation:DefinitionTheta}.
Using these two observation and the definition of $ \Theta_K$ we see that $\Theta_K(x) $ and $\Theta_{ K^{\oplus \alpha}}$ are related as follows,
\[
y \in \Theta_K(x) \iff x +  R_{K^{\oplus \alpha} } (x)  \frac{ y-x}{R_K(x)}  \in \Theta_{ K^{\oplus \alpha}}(x) ,
\]
assuming that $R_K(x) > \alpha$. 
So that, under the same condition, 
\begin{equation}\label{equation:ThetaOffsetK}
 \left(  \Theta_{ K^{\oplus \alpha}}(x)  - x \right) =  \frac{ R_{K^{\oplus \alpha} } (x)}{R_K(x)}  \left( \Theta_K(x) - x \right),
\end{equation}
which yields
\begin{equation}\label{equation:ExpressionFOfsset}
\F_{K^{\oplus \alpha}}(x) = \radius  \left( \Theta_{ K^{\oplus \alpha}}(x) \right) =  \frac{ R_{K^{\oplus \alpha} } (x)}{R_K(x)}  \radius  \left( \Theta_K(x) \right),
\end{equation}
and thus 
\[
\cent \left( \Theta_{ K^{\oplus \alpha}}(x)\right)   -x  =  \frac{ R_{K^{\oplus \alpha} } (x)}{R_K(x)}  \left( \cent \left( \Theta_K \right) -x \right).
\]
Thanks to the definitions \eqref{equation:DefinitionF} and \eqref{equation:DefinitionNabla}, we find that
\[
  R_K(x) > \alpha \Rightarrow \nabla_{K^{\oplus \alpha}} (x) =  \nabla_K(x),
\] 
which in turn implies that the flows $\Phi_K$ and $\Phi_{K^{\oplus \alpha}}$ coincide in $\left( K^{\oplus \alpha} \right)^c$.

 \begin{figure}[h!]
\centering
\includegraphics[width=.5\textwidth]{%pictures/
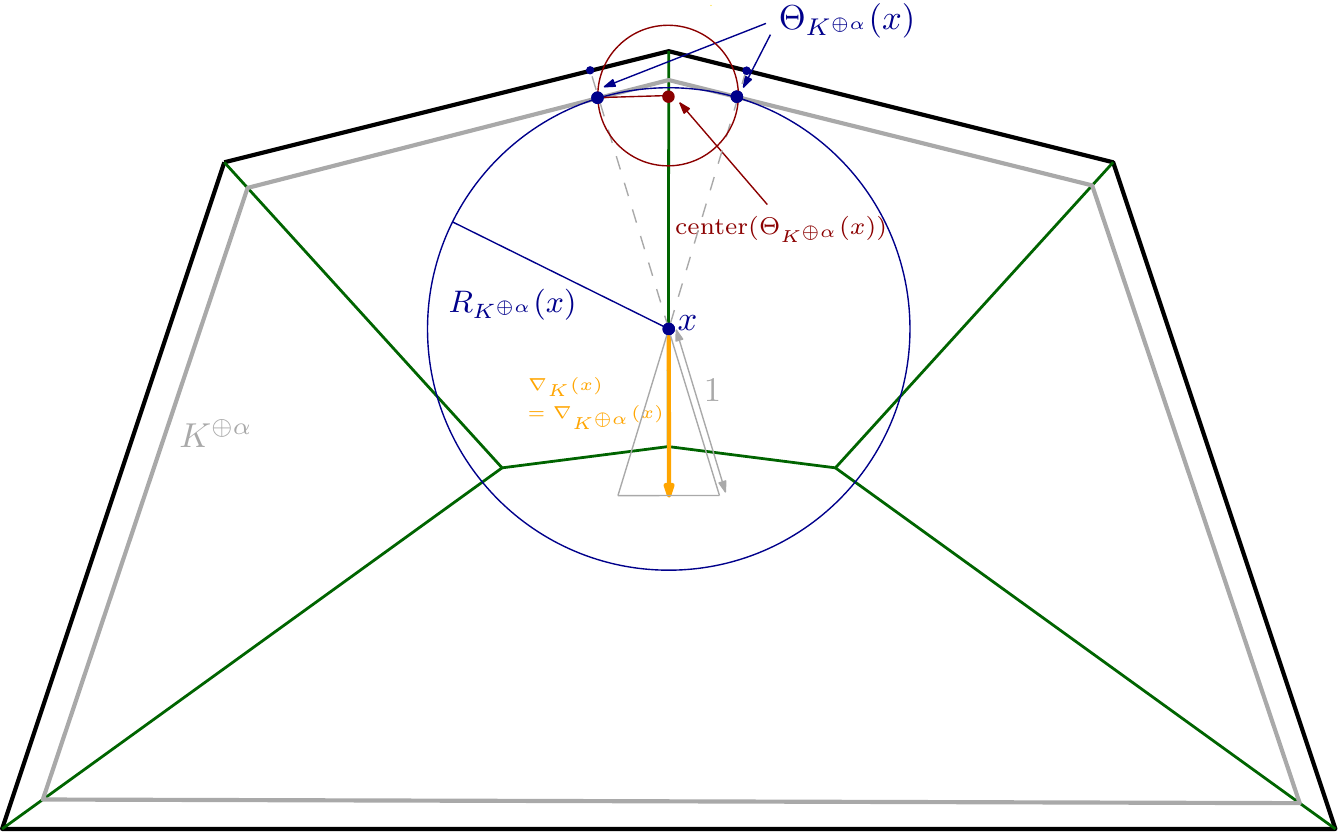}
\caption{Pictorial overview of the definitions and notation. The set $K^{\oplus \alpha}$ is indicated in grey.  
}
\Description{The same pentagon and medial axis as in Figures \ref{figure:DefinitionsNotations81} and \ref{figure:FlowExample} with an $\alpha$-offset and with notation indicated. }
\label{figure:DefinitionsNotations91}
\end{figure}	

We introduce the map $\F_K^\alpha : \left( K^{\oplus \alpha} \right)^c \rightarrow { \R_{\geq 0}}$ 
%{\color{red} [Notation positive numbers fixme]} 
as
\begin{align} 
\F_K^\alpha \defunder{=} \F_{K^{\oplus \alpha}}.
\label{eq:def:FKalpha}
\end{align} 
Therefore  \eqref{equation:DefinitionF}, \eqref{equation:ROffsetK} , \eqref{equation:ThetaOffsetK} and \eqref{equation:ExpressionFOfsset} give us
that for $x \in  \left( K^{\oplus \alpha} \right)^c$, that is  $R_K(x) > \alpha$, one has,
\begin{equation}\label{equation:ExpressionF_Alpha}
\F_K^\alpha(x) = \frac{R_K(x)- \alpha } {R_K(x)} \F_K(x).
\end{equation}

The $(\lambda,\alpha)$-medial axis  of $K$, denoted $\ax_\lambda^\alpha(K)$ is the $\lambda$-medial axis of the $\alpha$-offset of $K$
\begin{align}
\ax_\lambda^\alpha (K) & \defunder{=} \ax_\lambda (K^{\oplus \alpha})  
\nonumber
\\
&\, = \{x \in \R^n \mid \F_K^\alpha(x) \geq \lambda \}. 
 & & \textrm{(by \eqref{eq:def:FKalpha})} % used to be \tag, which looks better but this saves a line
\label{Eq:defAxAlphaLambdaK}
\end{align} 
We note that Lemma \ref{lemma:LmabdaMedialAxisIsCompact} extends to $(\lambda,\alpha)$-medial axis  of $K$ 
which is therefore compact as well.

Since, for $r\geq \alpha$, the map  $r \mapsto \frac{r-\alpha}{r}$ is increasing, we get from \eqref{equation:DerivativeR},
 Lemma \ref{lemma:FIsINotDecreasing}, and \eqref{equation:ExpressionF_Alpha} that
\begin{equation}\label{equation:FAlphaIsINotDecreasing}
t \mapsto \F_K^\alpha( \Phi_K (t,x )) \: \text{ is not decreasing.}
\end{equation}
In fact, when $\alpha>0$, $t\mapsto R_K( \Phi_K (t,x ))$ is strictly increasing as long as $\| \nabla_K(\Phi_K (t,x )) \| >0$, moreover 
the map  $t \mapsto \F_K^\alpha( \Phi_K (t,x )) $ 
 is  strictly increasing, which will be quantified in \eqref{equation:FAlphaIncreaseRate} below.
  
{The monotonicity in} \eqref{equation:FAlphaIsINotDecreasing}, together with the definition of $\ax_\lambda^\alpha (K)$, implies that $\ax_\lambda^\alpha (K)$
 is mapped to itself under the action of the flow:
\begin{equation}\label{equation:AXLambdaAlphaStableUnderFlow}
\forall t\geq 0, \: \Phi_K (t, \ax_\lambda^\alpha (K) ) \subset \ax_\lambda^\alpha (K).
\end{equation}
Observe that for $\lambda > \alpha$ one has 
\begin{align} 
x \in \ax_\lambda (K)  \Rightarrow  R_K(x) \geq \F_K(x) \geq \lambda.
\label{ShortLOWbndRKandFK} 
\end{align}
So if $x \in \ax_\lambda (K)$, we get from \eqref{equation:ExpressionF_Alpha} 
that
\[
\F_K^\alpha (x) = \frac{R_K(x)- \alpha } {R_K(x)} \F_K(x) \geq \frac{\lambda- \alpha }{\lambda} \F_K(x).
\]
So that,
\begin{align*} 
x \in \ax_\lambda (K)  &\Rightarrow \F_K(x) \geq \lambda \Rightarrow \F_K^\alpha (x) \geq  \frac{\lambda- \alpha }{\lambda} \lambda = \lambda- \alpha 
\\
 &\Rightarrow x \in \ax_{\lambda-\alpha}^\alpha  (K).
\end{align*} 
This in turn implies that if $\lambda >  \alpha$, one has,
\begin{align}
 \ax_{\lambda} (K)
\subset &\ax_{\lambda-\alpha}^\alpha (K)
\subset \ax_{\lambda-\alpha} (K) 
\label{equation:InclusionLambdaAlphaMedialAxe_A} 
\\
 \ax_{\lambda}^{\alpha} (K)
 \subset&  \ax_{\lambda} (K)
\subset \ax_{\lambda-\alpha}^{\alpha} (K) ,
\label{equation:InclusionLambdaAlphaMedialAxe_B}
\end{align}
and for any $c$ such that $0<c<1$ one has: 
\begin{equation}\label{equation:UnionOfLambdaAlphaMedialAxesGivesMedialAxis}
\bigcup _{\lambda >0} \: \bigcup _{0< \alpha < \lambda}   \ax_{\lambda}^\alpha (K) = \bigcup _{\lambda >0}   \ax_{\lambda}^{c \lambda} (K) = \ax(K).
\end{equation}

{ It follows  from \eqref{equation:ExpressionF_Alpha} and \eqref{Eq:defAxAlphaLambdaK} 
that $(\lambda, \alpha) \mapsto  \ax_{\lambda}^\alpha (K)$ is decreasing for inclusion order:
\begin{equation}\label{equation:LambdaAlphaMWDecreasing}
\lambda_1 \leq \lambda_2, \alpha_1 \leq \alpha_2 \Rightarrow \ax_{\lambda_2}^{\alpha_2}(K) \subset \ax_{\lambda_1}^{\alpha_1} (K).
\end{equation}
}

\subsection{The \texorpdfstring{$(\lambda,\alpha)$}{(lambda, alpha)}-medial axis is Hausdorff-stable under \texorpdfstring{$\lambda$}{lambda}  perturbation}

The purpose of this section is to show that the $(\lambda,\alpha)$-medial axis does not suffer much from instabilities. 
In this subsection we treat the stability with respect to $\lambda$ (Lemma \ref{corollary:LambdaMapstoAXLambdaAlphaIsHausdorffLipchitz} ),
 the following subsection is dedicated to the stability with respect to $\alpha$ (Lemma \ref{lemma:corollary:AlphaMapstoAXLambdaAlphaIsHausdorffLipchitz}).

We start with introducing the $(\alpha, \mu)$-reach.
\begin{definition}
For $\mu  \in (0,1]$ and $\alpha\geq 0$, the $(\alpha, \mu)$-reach of $K$, denoted by $r_\mu^\alpha (K)$, is defined as:
\[
r_\mu^\alpha (K) \defunder{=} \inf \{t> \alpha, \: \chi_K (t) < \mu \} = \alpha + r_\mu(K^{\oplus \alpha}).
\]
\end{definition} 

%%%%%%%%%%%%%%%%%%%%%%%%%%%%%%%%%%%%%%%%%%%%%%%%%%
%%%%%%%%%%%%%%%%%%%%%%%%%%%%%%%%%%%%%%%%%%%%%%%%%%
We need an easy lemma that gives a lower bound on the norm of $\| \nabla_K (x)\|$
for $x \in  \left( K^{\oplus \alpha'} \right)^c \setminus \ax_\lambda^\alpha$, for $0< \alpha'  \leq \alpha$.
To be able to state the result we define,
\begin{align}\label{equation:DefinitionBeta}
\tilde{\mu}_{\mu, \lambda}^{\alpha, \alpha'}(K) \defunder{=}  \min \left( \mu, \sqrt{1- \left( \frac{\lambda}{r_\mu^{\alpha'} (K)- \alpha}\right)^2 } \right).
\end{align}

%%%%%%%%%%%%%%%%%%%%%%%%%%%%%%%%%%%%%%%%%%%%%%%%%%
\begin{lemma}\label{lemma:NablaLowerBoundedByBeta}
Let $K\subset \R^n$ be the complement of a bounded open  set $K^c$ and 
$\alpha, \lambda > 0$, $\mu \in (0,1]$ and $\alpha'$ such that $0< \alpha' \leq \alpha$
and  $r_\mu^{\alpha'} (K) > \alpha + \lambda $.
\newline For any $x  \in  \left( K^{\oplus \alpha'} \right)^c \setminus \ax_\lambda^\alpha$, one has
\[
\| \nabla_K(x) \| \geq \tilde{\mu},
\]
 where
\[
\tilde{\mu} = \tilde{\mu}_{\mu, \lambda}^{\alpha, \alpha'}(K) > 0.
\]
\end{lemma}
%%%%%%%%%%%%%%%%%%%%%%%%%%%%%%%%%%%%%%%%%%%%%%%%%%

\begin{proof}%[of Lemma \ref{lemma:NablaLowerBoundedByBeta}] 
Because the lower bound $\tilde{\mu}$ is defined as a minimum over two values we distinguish two cases:
\begin{itemize}
\item If $R_K\left(x \right) < r_\mu^{\alpha'} (K)$ then  by definition of $ r_\mu^{\alpha'}$ one has \newline $\| \nabla_K \left(x \right) \| \geq \mu$.
\item 
If $R_K\left(x \right) \geq r_\mu^{\alpha'} (K)$ then  $R_{K^{\oplus \alpha}}(x) \geq r_\mu^{\alpha'} (K) - \alpha$.
Since $x\notin  \ax_\lambda^\alpha$, one has $ \F_{K^{\oplus \alpha}} (x) = \F_\lambda^\alpha(x) < \lambda$.
Combining this with \eqref{equation:NormGradientRelatedToFandR} we get,
\begin{align*}
\| \nabla_K \left(x \right) \| &= \| \nabla_{K^{\oplus \alpha}}  \left( x \right) \|  \\
& =  \sqrt{ 1 -  \left( \frac{ \F_{K^{\oplus \alpha}} (x)}{R_{K^{\oplus \alpha}}(x)}\right)^2 } \\
 &>  \sqrt{ 1 -   \left( \frac{ \lambda} {r_\mu^{\alpha'} (K) - \alpha} \right)^2 } >0.
\end{align*}
\end{itemize} 
So that in both cases, we get:
\begin{equation}\label{equation:LowerBoundOnNabla}
\| \nabla_K\left(x \right) \|  \geq \tilde{\mu}_{\mu, \lambda}^{\alpha, \alpha'}(K) >0.
\end{equation}
\end{proof}

%%%%%%%%%%%%%%%%%%%%%%%%%%%%%%%%%%%%%%%%%%%%%%%%%%
%%%%%%%%%%%%%%%%%%%%%%%%%%%%%%%%%%%%%%%%%%%%%%%%%%

As observed in \cite{cl2005lambda}, the 
%\xout{{\color{blue} Is there a reason why to switch back from $\mu$ to $\lambda$? } } 
$\lambda$-medial axis  seen as a function $\lambda \mapsto \ax_\lambda (K)$
is not continuous: $\ax_\lambda (K)$ may ``increase'' abruptly at some ``singular'' values of $\lambda$, 
 even when $\lambda$ is small with respect to $\wfs(K)$. 
 This is illustrated in Figures \ref{figure:MLambdaDiscontinuous} and \ref{figure:MLambdaContinuous}.
This is related to the fact that, for $x \in \ax_\lambda(K)$, the map $t \mapsto \F_K(\Phi_K(t,x))$ may remain constant on some intervals.
   
 \begin{figure}[ht!]
\centering
\includegraphics[width=.5\textwidth]{%pictures/
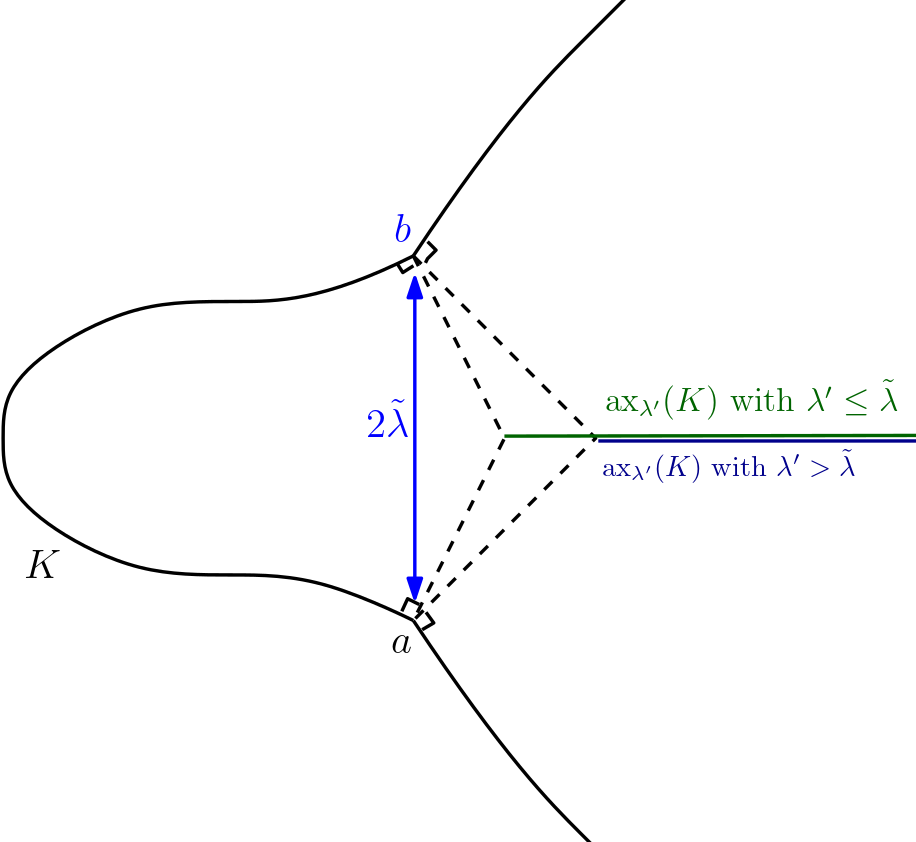}
\caption{$\lambda \mapsto \ax_\lambda (K)$ is not continuous, because $K$ is non-smooth at the points $a$ and $b$. 
}
\Description[A curve that is smooth with the exception of two points.]{A curve that is smooth with the exception of two points, $a$ and $b$, that lie directly above each other and a distance $2\tilde{\lambda}$ apart. This gives that the $\lambda$-medial axis is not continuous at $\tilde{\lambda}$.}
\label{figure:MLambdaDiscontinuous}
\end{figure}
 \begin{figure}[h!]
\centering
\includegraphics[width=.4\textwidth]{%pictures/
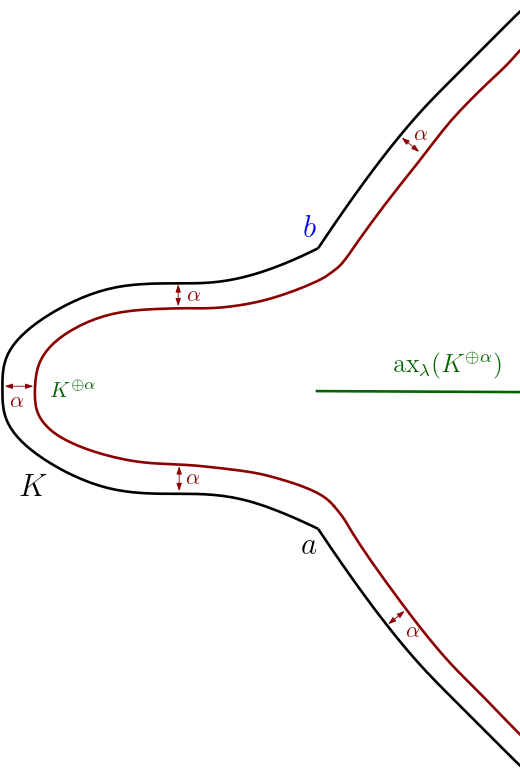}
\caption{$\lambda \mapsto \ax_\lambda^\alpha (K)$ is  continuous.}
\label{figure:MLambdaContinuous}
\Description[The same curve as in the previous figure (Figure \ref{figure:MLambdaDiscontinuous}), but with an $\alpha$ offset. ]{The same curve as in the previous figure (Figure \ref{figure:MLambdaDiscontinuous}), but with an $\alpha$ offset. Taking the offset removes the discontinuity. }
\end{figure}	
	
In contrast, for $\alpha>0$, the map  $\lambda \mapsto \ax^\alpha_\lambda (K)$  is continuous with respect to the Hausdorff distance, 
as long as $\alpha+ \lambda < \wfs(K)$, %{\color{red} [I think it would make more sense to have it inline.]}\footnote{ {\color{blue}  
or, more generally, as long as  $\alpha+ \lambda < r_\mu^\alpha(K)$ for some $\mu>0$.
The continuity follows from the fact that 
the rate of increase of the map $t \mapsto \F^\alpha_K(\Phi_K(t,x))$ is lower bounded as soon as $\F^\alpha_K(x) >0$.
More precisely, we have:
%%%%%%%%%%%%%%%%%%%%%%%%%%%%%%%%%%%%%%%%%%%
\begin{lemma}\label{lemma:BoundOnLengthToLambdaPlusDelta}
Let $K\subset \R^n$ be the complement of a bounded open  set $K^c$ and 
$\alpha,  \delta, \lambda > 0$, $\mu \in (0,1]$
such that $r_\mu^\alpha (K) > \alpha + \lambda + \delta$.
One has:
\begin{equation}\label{equation:FlowSendsInLargerLambda}
\Phi_K \left( T, \ax_\lambda^\alpha (K) \right) \subset  \ax_{\lambda+\delta}^\alpha (K),
\end{equation}
with
\[
T = \frac{R_{\max}^2}{\alpha \lambda \tilde{\mu}^2} \delta,
\]
{where $R_{\max} = R_{\max}(K)  < \infty$ has been defined in \eqref{equation:DefinitionRMax}.}
For any $x \in \ax_\lambda^\alpha (K)$, the path
\[
[0,T] \ni t \rightarrow \Phi_K(t, x) \in  \ax_{\lambda}^\alpha (K),
\]
has length $S$ upper bounded by
\[
S \leq \frac{R_{\max}^2}{\alpha \lambda \tilde{\mu}^2} \delta,
\]
where %$R_{\max} = R_{\max}(K)  < \infty$ and 
$\tilde{\mu} = \tilde{\mu}_{\mu,\lambda+ \delta}^{\alpha, \alpha}(K) > 0$,
and
\begin{equation}\label{equation:InclusionOffsetLambdaMedialInLargerLambdaMedial}
\ax_\lambda^\alpha (K)  \subset  \left(       \ax_{\lambda+\delta}^\alpha (K)  \right)^{\oplus S} .
\end{equation}

\end{lemma}
%%%%%%%%%%%%%%%%%%%%

\begin{proof}%[of Lemma \ref{lemma:BoundOnLengthToLambdaPlusDelta}]
For $r\geq \alpha$ the map  $t\mapsto \frac{r-\alpha}{r}$ is increasing, this together with \eqref{equation:DerivativeR}, Lemma \ref{lemma:FIsINotDecreasing} and \eqref{equation:ExpressionF_Alpha} yields
\begin{equation}\label{equation:FirstBoundOnFAlphaPhiOfTandX}
 \F^\alpha_K\left(\Phi_K(t,x)\right) \geq \frac{R_K\left(\Phi_K(t,x)\right)- \alpha } {R_K\left(\Phi_K(t,x)\right)} \F_K(x).
\end{equation}
Observe that, using \eqref{equation:DerivativeR} again,
\begin{equation}\label{equation:FirstExpressionOfDerivativeRKOnRkMinusAlpha}
\frac{d}{dt^+} \frac{R_K\left(\Phi_K(t,x)\right)- \alpha } {R_K\left(\Phi_K(t,x)\right)}  = \frac{\alpha}{R_K\left(\Phi_K(t,x)\right)^2} \nabla_K\left(\Phi_K(t,x)\right)^2.
\end{equation}
Consider now $T$ such that  $\Phi_K( T,x) \notin \ax_{\lambda+\delta}^\alpha (K)$, and then $ \F^\alpha_K\left(\Phi_K(T,x)\right)   < \lambda+\delta$.
From \eqref{equation:FAlphaIsINotDecreasing} this gives:
\begin{equation*}%%\label{equation:SecondBoundOnFAlphaPhiOfTandX}
\forall t\in [0, T], \:  \F^\alpha_K\left(\Phi_K(t,x)\right)   < \lambda+\delta.
\end{equation*}
Lemma \ref{lemma:NablaLowerBoundedByBeta} gives a lower bound on $\| \nabla_K(\Phi_K(t,x)) \|$, that is, 
\[
\| \nabla_K(\Phi_K(t,x)) \| \geq \tilde{\mu} = \tilde{\mu}_{\mu,\lambda+ \delta}^{\alpha, \alpha}(K) > 0
\]
and \eqref{equation:FirstExpressionOfDerivativeRKOnRkMinusAlpha} then gives,
\begin{equation}\label{equation:LowerBoundOnDerivativeRKOnRkMinusAlpha}
t\in [0, T] \Rightarrow \frac{d}{dt^+} \frac{R_K\left(\Phi_K(t,x)\right)- \alpha } {R_K\left(\Phi_K(t,x)\right)}  > \frac{\alpha \tilde{\mu}^2}{R_{\max}^2}.
\end{equation}
This, thanks to Theorem \ref{theorem:FundamentalThmCalculusACFunctions}, leads us to the following bound, 
\begin{align}
\F^\alpha_K & \left(\Phi_K(T,x)\right)  - \F^\alpha_K(x) 
\nonumber
\\ 
& \geq  \F_K(x)  \left( \frac{R_K\left(\Phi_K(T,x)\right)- \alpha } {R_K\left(\Phi_K(T,x)\right)}  - \frac{R_K(x )- \alpha } {R_K(x)} \right) 
\tag{by \eqref{equation:FirstBoundOnFAlphaPhiOfTandX}}
\\
& =  \F_K(x)  \int_0^{T}  \left( \frac{d}{dt^+} \frac{R_K\left(\Phi_K(t,x)\right)- \alpha } {R_K\left(\Phi_K(t,x)\right)} \right) dt   \nonumber
\\
& > \F_K(x) T   \frac{\alpha \tilde{\mu}^2}{R_{\max}^2}.
\tag{by \eqref{equation:LowerBoundOnDerivativeRKOnRkMinusAlpha}}  
\\
\label{equation:FAlphaIncreaseRate}
\end{align}
%{\color{blue} We use the fundamental theorem of calculus a lot but with only a one-sided derivative. It seems clear that this works, but it would be nice if we could add a reference. Do you know one or should I start looking for one? } 
Because $x \in \ax_\lambda^\alpha (K)$, $\F^\alpha_K(x) \geq \lambda$, 
and \eqref{equation:InclusionLambdaAlphaMedialAxe_B} we find that $x \in \ax_\lambda (K)$. The fact that  $x \in \ax_\lambda (K)$ in turn implies that $\F_K(x) \geq \lambda$. This then yields,
\[
 \F^\alpha_K\left(\Phi_K(T,x)\right) > \lambda + \frac{\lambda \alpha \tilde{\mu}^2}{R_{\max}^2} T. 
\]
We have shown that:
\begin{eqnarray*}
x \in \ax_\lambda^\alpha (K) \: & \text{and} & \: \Phi_K( T,x) \notin \ax_{\lambda+\delta}^\alpha (K) \\
& \Rightarrow& \:\lambda + \frac{\lambda \alpha \tilde{\mu}^2}{R_{\max}^2} T < \lambda+ \delta  
\Rightarrow  T <  \frac{\delta R_{\max}^2} {\lambda \alpha \tilde{\mu}^2}.
\end{eqnarray*}
By contraposition we have:
\[
x \in \ax_\lambda^\alpha (K)\: \text{and} \: T \geq  \frac{\delta R_{\max}^2} {\lambda \alpha \tilde{\mu}^2}
 \Rightarrow \Phi_K( T,x) \in \ax_{\lambda+\delta}^\alpha (K).
\]

Since $\forall y,\, \| \nabla_K (y)\| \leq 1$, one has
\[
S =  \length  \left(\Phi_K([0,T], x \right)  =  \int_0^T  \| \nabla_K \left(\Phi_K(t,x)\right) \| dt \leq T  \leq \frac{\delta R_{\max}^2} {\lambda \alpha \tilde{\mu}^2}.
\]
Since the Euclidean distance between $x \in \ax_\lambda^\alpha (K)$ and  $ \Phi_K\left(T, x \right)  \in  \ax_{\lambda+\delta}^\alpha (K) $
cannot be larger than the length of the path connecting them we get \eqref{equation:InclusionOffsetLambdaMedialInLargerLambdaMedial}.

\end{proof}

%%%%%%%%%%%%%%%%%%%%%%%%%%%%%%%%%%%%%%%
{
Since for $\delta \geq 0$  one has 
$  \ax_{\lambda+\delta}^\alpha (K)   \subset    \ax_{\lambda}^\alpha (K)   $, 
\eqref{equation:InclusionOffsetLambdaMedialInLargerLambdaMedial} and  the upper bound
 \eqref{equation:HausdorffDistanceForNestedSets} immediately give us:
}
\begin{lemma}\label{corollary:LambdaMapstoAXLambdaAlphaIsHausdorffLipchitz}
Let $K\subset \R^n$ be the complement of a bounded open  set $K^c$ and 
$\alpha,   \lambda_{\max} > 0$, $\mu \in (0,1]$
such that $r_\mu^\alpha (K) > \alpha +  \lambda_{\max}$.
The map
\[
\lambda \mapsto \ax_{\lambda}^\alpha (K)
\]
is $\left( \frac{R_{\max}^2}{\alpha \lambda_{\min} \tilde{\mu}^2} \right)$-Lipschitz in the interval $[\lambda_{\min}, \lambda_{\max}]$
 for  Hausdorff distance,
where $R_{\max} = R_{\max}(K) < \infty$ and $\tilde{\mu} = \tilde{\mu}_{\mu,  \lambda_{\max}}^{\alpha, \alpha}> 0$.
\end{lemma}
%%%%%%%%%%%%%%%%%%%%

\subsection{The \texorpdfstring{$(\lambda,\alpha)$}{(lambda, alpha)}-medial axis is {Hausdorff-}stable under \texorpdfstring{$\alpha$}{ alpha} perturbation}
In this subsection we complete the proof of the stability of the $(\lambda,\alpha)$-medial axis under perturbations of $\lambda$ and $\alpha$. 
%{\color{red} 
%[I don't know if I can give much intuition, but I can give an overview of the proof]
It is not difficult to see that $\ax_\lambda^{\alpha+\delta}(K) \subset \ax_\lambda^{\alpha}(K)$ for $\delta\geq 0$. We establish that $\ax_\lambda^{\alpha}(K)$ is not very far from $\ax_\lambda^{\alpha+\delta}(K) $ by proving that $\ax_\lambda^{\alpha}(K)$ flows (fast enough) into $\ax_\lambda^{\alpha+\delta}(K)$. This is done by determining that the radius of the enclosing ball of the closest points on $K$ increases sufficiently fast.  
%} 

\begin{lemma}\label{lemma:corollary:AlphaMapstoAXLambdaAlphaIsHausdorffLipchitz}
Let $K\subset \R^n$ be the complement of a bounded open  set $K^c$ and 
$\alpha,  \lambda > 0$, $\mu \in (0,1]$
if $r_\mu^{\alpha_{\min}} (K) > \alpha_{\max} + \lambda$.
The map 
\[
\alpha \mapsto \ax_{\lambda}^\alpha (K)
\]
is $\left(\frac{R_{\max}}{\alpha_{\min} \tilde{\mu}^2}\right)$-Lipschitz in the interval $[\alpha_{\min} , \alpha_{\max} ]$
 for Hausdorff  distance,
 where  $\tilde{\mu} =  \tilde{\mu}_{\mu, \lambda}^{\alpha_{\max}, \alpha_{\min}} > 0$.
 \end{lemma}

\begin{proof}%[of Lemma \ref{lemma:corollary:AlphaMapstoAXLambdaAlphaIsHausdorffLipchitz}] 
From \eqref{equation:ExpressionF_Alpha} we conclude that,
\[
\delta \geq 0 \Rightarrow \forall x \in  \left(K^{\oplus \alpha + \delta} \right)^c, \,  \F_K^\alpha(x)  \geq \F_K^{\alpha + \delta} (x) . 
\]
It follows that
\begin{equation}\label{equation:AXLambdaAlphaDecreasingWithAlpha}
\delta \geq 0 \Rightarrow \ax_\lambda^{\alpha+\delta}(K) \subset \ax_\lambda^{\alpha}(K).
\end{equation}
Let us assume that for some $\mu>0$ one has  $r_\mu^\alpha(K) > \alpha+ \delta + \lambda$ and take $x \in \ax_\lambda^{\alpha} (K)$. 
Using Lemma \ref{lemma:NablaLowerBoundedByBeta}, we see that
\begin{align}
\Phi_K(t,x)  \notin \ax_\lambda^{\alpha+ \delta}(K)  \Rightarrow \| \nabla_K \left( \Phi_K(t,x)  \right) \| > \tilde{\mu} >0,
\label{eq:LowBoundNablaPhiKbyBeta}
\end{align}
where $\tilde{\mu} = \tilde{\mu}_{\mu, \lambda}^{\alpha+\delta, \alpha}(K) > 0$.

It then follows from \eqref{equation:DerivativeR} that,
%{\color{blue} 
\begin{align} 
R_K \left( \Phi_K \left(\frac{R_{\max}}{\alpha \tilde{\mu}^2} \delta, x \right)  \right) 
&= \int_{0}^{\frac{R_{\max}}{\alpha \tilde{\mu}^2} \delta } \frac{d}{dt'^+} R_K( \Phi_K (t',x )) \ud t'+ R_K(x) 
\nonumber
\\
& =  \int_{0}^{\frac{R_{\max}}{\alpha \tilde{\mu}^2} \delta } \left\| \nabla_K \left( \Phi_K (t',x ) \right) \right\|^2 \ud t'+ R_K(x)
\nonumber
\\
&> R_K(x) + \frac{R_{\max}}{\alpha} \delta 
\tag{by \eqref{eq:LowBoundNablaPhiKbyBeta}}
\\
&\geq R_K(x) \frac{\alpha + \delta}{\alpha} .
\nonumber
\end{align}
%}
Using \eqref{equation:ExpressionF_Alpha}, we see
\begin{align}
\F_K^{\alpha+\delta} &\left( \Phi_K \left(\frac{R_{\max}}{\alpha \tilde{\mu}^2} \delta, x \right) \right) 
\nonumber
\\
&= \frac{R_K \left( \Phi_K \left(\frac{R_{\max}}{\alpha \tilde{\mu}^2} \delta, x \right) \right)- (\alpha+ \delta) } {R_K\left(\Phi_K \left(\frac{R_{\max}}{\alpha \tilde{\mu}^2} \delta, x \right)\right)} 
\F_K\left(\Phi_K \left(\frac{R_{\max}}{\alpha \tilde{\mu}^2} \delta, x \right) \right) 
\nonumber
\\
&> \frac{ R_K(x) \frac{\alpha + \delta}{\alpha}  - (\alpha+ \delta)}{R_K(x) \frac{\alpha + \delta}{\alpha} } \F_K(x)
\nonumber
\\
& = \frac{ R_K(x) -\alpha}{R_K(x)} \F_K(x) 
\nonumber
\\
&= \F_K^\alpha(x) 
\tag{by definition} 
\\
&\geq \lambda, \tag{because $x \in \ax_\lambda^{\alpha}$} 
\end{align}
where in the second line we used that for $r\geq \alpha + \delta$, the map $r\mapsto \frac{r-(\alpha+ \delta) }{r}$ is increasing, as well as the previous bound
and Lemma \ref{lemma:FIsINotDecreasing}.

With this we have shown that
\begin{equation}\label{equation:FlowSendsInLargerAlpha}
 r_\mu^\alpha(K) > \alpha+ \delta + \lambda \Rightarrow \Phi_K \left(\frac{R_{\max}}{\alpha \tilde{\mu}^2} \delta,   \ax_\lambda^{\alpha} \right) 
 \subset  \ax_\lambda^{\alpha+\delta}(K) .
\end{equation}
{
As for Lemma \ref{corollary:LambdaMapstoAXLambdaAlphaIsHausdorffLipchitz},
}
this together with \eqref{equation:AXLambdaAlphaStableUnderFlow} and \eqref{equation:AXLambdaAlphaDecreasingWithAlpha} completes the proof of the lemma. 
\end{proof}

%%%%%%%%%%%%%%%%%%%%
\subsection{The \texorpdfstring{$(\lambda,\alpha)$}{(lambda, alpha)}-medial axis has the right homotopy type and a finite geodesic diameter} \label{sec:finiteDiam} 

In this section, we recover that the $(\lambda,\alpha)$-medial axis preserves the homotopy type, as for the usual medial axis \cite{LIEUTIERhomotopytype}. We then prove that there exists a path of finite length between any two points in the $(\lambda,\alpha)$-medial axis, that is it has a finite geodesic diameter. Both results are again based on manipulations of the flow $\Phi_K$. 

The next lemma gives an upper bound on the time needed for the flow $\Phi_K$ to map $K^c$ inside $\ax_{\lambda}^\alpha (K)$.
It is instrumental for characterizing the homotopy type of $\ax_{\lambda}^\alpha (K)$ (Theorem \ref{theorem:AxLambdaAlphaHasRightHomotopyType})
and establishing that it has finite geodesic diameter (Theorem \ref{theorem:AXLambdaAlphaIsGeodesicAndFiniteDiameter}).

\begin{lemma}\label{lemma:KCFlowsToLambdaPlusDelta}
Let $K\subset \R^n$ be the complement of a bounded open  set $K^c$ and 
$\alpha \geq 0$, $\lambda > 0$, $\mu \in (0,1] $ and $\alpha' $ such that $0< \alpha' \leq \alpha$ and
$r_\mu^{\alpha'} (K) > \alpha + \lambda$.
Then,
$ % \[
\Phi_K \left( \frac{R_{\max}}{\tilde{\mu}^2}, \left( K^{\oplus \alpha'}\right)^c \right) \subset \ax_\lambda^\alpha (K), 
$ %\]
where $R_{\max} = R_{\max} (K) < \infty$ and $\tilde{\mu} =  \tilde{\mu}_{\mu, \lambda}^{ \alpha, \alpha'} > 0$.
\end{lemma}

\begin{proof} %[of Lemma \ref{lemma:KCFlowsToLambdaPlusDelta}]
Consider $x \in \left(K^{\oplus \alpha'}\right)^c$. 
Using \eqref{equation:AXLambdaAlphaStableUnderFlow} and Lemma \ref{lemma:NablaLowerBoundedByBeta} again, one finds,
\[
\Phi_K(T,x) \notin \ax_\lambda^\alpha \Rightarrow \forall t \in [0,T],  \| \nabla_K \left( \Phi_K(t,x) \right)  \| > \tilde{\mu}.
\]
Applying \eqref{equation:DerivativeR} as before yields, 
\begin{align*} 
&\Phi_K(T,x) \notin \ax_\lambda^\alpha(K) 
\\
&\Rightarrow R_K(\Phi_K(T,x) )  =
\int_{0}^{T} \frac{d}{dt^+} R_K( \Phi_K (t,x )) \ud t + R_K(\Phi_K(0,x) ) > T \tilde{\mu}^2. 
\end{align*} 
Since $R_K(\Phi_K(T,x) )  \leq R_{\max}$ one gets,
\[
\Phi_K(T,x) \notin \ax_\lambda^\alpha(K)  \Rightarrow T \tilde{\mu}^2 < R_{\max}  \Rightarrow T < \frac{R_{\max}}{\tilde{\mu}^2}.
\]
The contraposition gives,
\[
 \Phi_K \left( \frac{R_{\max}}{\tilde{\mu}^2}, x \right) \in  \ax_\lambda^\alpha(K).
\]
\end{proof}

We now define the map 
\begin{align}
H : [0,1] \times \left( K^{\oplus \alpha'}\right)^c  &\rightarrow \left( K^{\oplus \alpha'}\right)^c 
\nonumber
\\
(t,x)  &\mapsto   \Phi_K \left( \frac{R_{\max}}{\tilde{\mu}^2} t, x \right) =H(t,x).  
\nonumber
\end{align}
Lemma \ref{lemma:KCFlowsToLambdaPlusDelta}, \eqref{equation:AXLambdaAlphaStableUnderFlow}, and the natural inclusion $\ax_\lambda^\alpha(K) \subset \left( K^{\oplus \alpha'}\right)^c $ together imply that $H$ gives a homotopy equivalence between $\left( K^{\oplus \alpha'}\right)^c$ and $\ax_\lambda^\alpha(K)$,
as it satisfies Definition \ref{definition:HomotopyEquivalenceDefinition}.
%and ,
 %in the conditions of the lemma, the homotopy:
%
%realizes, together with the natural inclusion $\ax_\lambda^\alpha(K) \subset \left( K^{\oplus \alpha'}\right)^c $, 
%a homotopy equivalence between $\left( K^{\oplus \alpha'}\right)^c$ and $\ax_\lambda^\alpha(K)$.
Since, by \eqref{equation:WfsAndMuReach},
\[r < \wfs(K) \Rightarrow \exists \mu>0, \, \textrm{such that } r_\mu^\alpha (K) \geq r,\] 
one sees that:
 %%%%%%%%%%%%%%%%%%%%%%%%%%%%%%%%%%%%%%%%%%%%%%%%%
\begin{theorem}\label{theorem:AxLambdaAlphaHasRightHomotopyType}
Let $K\subset \R^n$ be the complement of a bounded open  set $K^c$, $\alpha \geq 0$, $\lambda > 0$ and  $\alpha'$ 
such that $0< \alpha' \leq \alpha$ and $r_\mu^{\alpha'} (K) > \alpha + \lambda$  
then  $\ax_\lambda^\alpha(K)$ has the same homotopy type as
$\left( K^{\oplus \alpha'}\right)^c$.

% In particular, if $\wfs(K) > \alpha + \lambda$
%then  $\ax_\lambda^\alpha(K)$ has the same homotopy type as $K^c$.
\end{theorem}

Pushing a path $\gamma$ along the flow $\Phi$, while keeping its endpoints fixed, will be instrumental in several subsequent proofs.
We introduce therefore a specific notation for such pushed paths.

\begin{definition}[Pushed paths]\label{definition:PushedPaths}
Let $K\subset \R^n$ be a closed  set,  $\gamma:[0,1] \rightarrow K^c$  a rectifiable  path and $T\geq0$.
The path $\gamma$ pushed along a time $T$ by the flow $\Phi_K$ with fixed end points,  denoted  $\PushedPath{\gamma}{T}$
is  the path $\PushedPath{\gamma}{T} :[0,1] \rightarrow K^c$ from $\gamma(0)$ to $\gamma(1)$ defined by:
\begin{equation}
\PushedPath{\gamma}{T} (t) = 
\begin{cases}
	\Phi_K(3 t T, \gamma(0)) &\textrm{if} \quad t \in [0, 1/3] \\
	\Phi_K(T, \gamma(3t -1 )) &\textrm{if}\quad t \in [1/3, 2/3] \\
	\Phi_K((3(1- t) T, \gamma(1)) &\textrm{if} \quad t \in [2/3, 1] .
     \end{cases}
\nonumber
\end{equation}
\end{definition}
The next simple lemma says that $\PushedPath{\gamma}{T} (t)$ is rectifiable and allows to upper bound path lengths in several subsequent proofs (Theorems \ref{theorem:AXLambdaAlphaIsGeodesicAndFiniteDiameter} and \ref{theorem:GromovHausdorffStability}, Lemmas \ref{lemma:AxLambdaAlphaGHStableWRTLambda} 
and \ref{lemma:AxLambdaAlphaGHStableWRTAlpha} ). 

\begin{lemma}\label{lemma:PushedPathLengthExpansion}
Let $K\subset \R^n$ be a closed  set 
and $\alpha > 0$. Let $\gamma:[0,1] \rightarrow \left(K^{\oplus \alpha}\right)^c$ be a rectifiable path with length $L(\gamma)$ and $T\geq0$.
$\PushedPath{\gamma}{T}$ is rectifiable and has length $L\left(\PushedPath{\gamma}{T}  \right)$ upper bounded by:
\[
L\left(\PushedPath{\gamma}{T}  \right) \leq 2 \, T + L(\gamma) e^{\frac{T}{\alpha}}.
\]
\end{lemma}

\begin{proof} %[of Lemma \ref{lemma:PushedPathLengthExpansion}]
Since the norm of $\nabla_K(x)$ is upper bounded by $1$, $\PushedPath{\gamma}{T} $ is $3$-Lipschitz 
on the first and third intervals $[0, 1/3]$ and $[2/3, 1]$ so that the length of each of these parts is upper bounded by  $T$.

Because $\Image(\gamma)$ is included in  $\left(K^{\oplus \alpha}\right)^c$, the length of 
$\PushedPath{\gamma}{T} \left(  [1/3, 2/3]   \right)$  is bounded by 
\[ 
 L(\gamma) e^{\frac{T}{\alpha}}.
 \]
This follows, because \eqref{equation:BoundOnFlowExpansion} can be applied to arbitrary small subdivisions of $\mathbf{\gamma}$,
the bound on the expansion factor extends to the length of the curve through the definition of the length of rectifiable curves.
\end{proof}

The next theorem  shows that connected components of $(\lambda,\alpha)$-medial axes have finite geodesic diameter.
In particular, geodesic distances inside connected components of $(\lambda,\alpha)$-medial axes are finite, and are realized by minimal paths.

%%%%%%%%%%%%%%%%%%%%%%%%%%%%%%%%%%%%%%%%%%
{Before we can go into the precise statement we need to introduce some notation. } 
For a Borel set $X\subset \R^n$ we  denote by $\vol (X)$ its $n$-Lebesgue measure, or volume. We write $V_n = \vol(\B^n(0,1))$ for the volume of the unit $n$-ball,
$ %\[
V_n = \frac{\pi^{\frac{n}{2}}}{\Gamma\left(\frac{n}{2} +1\right)},
$ % \]
where $\Gamma$ is the Euler Gamma function, see for example \cite[page 622]{DuistermaatKolkIntegration}.

\begin{remark}
Thanks to Theorem  \ref{theorem:AxLambdaAlphaHasRightHomotopyType},  connected components of 
 $\ax_\lambda^\alpha(K)$ are in one to one correspondance with connected components of  $\left(K^{\oplus \alpha}\right)^c$.
 For this reason, the  next statements related to geodesic diameters and Gromov-Hausdorff distance  assumes $\left(K^{\oplus \alpha}\right)^c$ 
 connected without real loss of generality.
 \end{remark}

%{\color{red} [I very much like the idea of the example, but the notation got me thoroughly confused. ]}

While  any connected open  subset of $\R^n$ is path-wise connected \cite [Proposition 12.25]{sutherland2009introduction}, some 
connected open bounded subset may not have a finite geodesic diameter. 
%{\color{red} \sout{, such as:} 
We'll illustrate this in the following example. Let %} 
%{\color{red} I don't get notation $\{t\} \times$ wouldn't this get us to $\R^3$?} 
%{\color{blue}
%Well , it is not my fault if the english use the same notation for a couple and an open interval !!! {\color{red} I too prefer the $] a,b[$ notation. }
%
%Do you prefer:
%\[
%S = \bigcup_{0< t < 1} \{t\} \times \Big] -t + \sin \frac{1}{t} ,\quad   t  + \sin \frac{1}{t} \Big[ \subset \R^2
%\]
{
\[ S=\left \{ (x,y) \in \mathbb{R}^2   \big\vert \:  0 < x< 1, \: -\frac{1}{2} +  \sin \frac{1}{x} < y <  \frac{1}{2}  +   \sin \frac{1}{x}  \right \} 
\]}
%} 
Observe that $S$ is indeed a connected, bounded, open subset of $\R^2$ with infinite geodesic diameter.
Moreover, for some $\alpha >1$, taking
\[ K = \Big\{ (x,y,z) \in \R^3 \big\vert \: (x,y) \notin S \:\:\:  \text{and} \:\:\: d\big( (x,y,z) , \:  S \times \{0\} \big) \geq \alpha \Big\},
\]
we get that the open bounded set $\left(K^{\oplus \alpha}\right)^c$ deform retracts by vertical projection on $S \times \{0\}$
along trajectories of small length $h < 1$.
Thanks to this vertical deformation retract we have that $d_{GH} (\left(K^{\oplus \alpha}\right)^c, S) \leq 2h$. 
% Pushing geodesics in $\left(K^{\oplus \alpha}\right)^c$ by this deformation retract shows that  $d_{GH} (\left(K^{\oplus \alpha}\right)^c, S) \leq 2h$ {\color{red} [TODO fix text]} {\color{blue}. Indeed, we have that }
More precisely the argument goes as follows: 
We have that
\[
\left(K^{\oplus \alpha}\right)^c = \left\{ (x,y,z) \middle | (x,y) \in S \textrm{ and } |z| < \bar{h}(x,y) \right \}, 
\]  
with $\bar{h}(x,y) \leq h \leq 1$. It follows by concatenating a geodesic $S$ with vertical line segments, that 
\[
d( (x_1,y_1,z_1), (x_2,y_2,z_2)) \leq d( (x_1,y_1), (x_2,y_2)) +2 h.
\] 
Here we wrote $(x_1,y_1,z_1), (x_2,y_2,z_2)$ for elements of $\left(K^{\oplus \alpha}\right)^c$ and $(x_1,y_1), (x_2,y_2)$ for the elements of $S$ with the same first two coordinates. Conversely, $d( (x_1,y_1), (x_2,y_2)) \leq d( (x_1,y_1,z_1), (x_2,y_2,z_2))$ because for any rectifiable curve $\gamma(t)\subset \R^3$ we have $L(\pi_{\R^2}(\gamma(t))) \leq L(\gamma(t))$, where $\pi_{\R^2}$ denotes the projection onto $\R^2$. 
It follows from the bound $d_{GH} (\left(K^{\oplus \alpha}\right)^c, S) \leq 2h$ that $\left(K^{\oplus \alpha}\right)^c$ also has infinite geodesic diameter.
%$K\subset \R^3$  to be the set of points at distance at most $\alpha$ from the closure of $S\times \{0\}$, we get that $\left(K^{\oplus \alpha}\right)^c = S$.
So to be the complement of an offset does not imply  finite geodesic diameter either. 

However, we have the following:
\begin{lemma}\label{lemma:OpenSetWithGeodesicFiniteDiameter}
Let $K\subset \R^n$ be the complement of a bounded, open  set $K^c$.
Assume there are $\mu>0$ and $\alpha > \alpha' > 0$  such that $r_\mu^{\alpha'}(K) > \alpha$.

Then, if   $\left(K^{\oplus \alpha}\right)^c$ is connected it has finite geodesic diameter.
In particular, if for some $z\in \R^n$ and $R_{\operatorname{bound}}>0$, one has $K^c \subset \B (z, R_{\operatorname{bound}} )$, then: 
\begin{align*}
\Gdiam \left(\left(K^{\oplus \alpha}\right)^c\right) \leq &2   \frac{\alpha-\alpha'}{\mu} \\& + 2\left( \left( \frac{ 4 R_{\operatorname{bound}} } {\alpha- \alpha' }\right)^n+1\right) (\alpha - \alpha')e^{\frac{\alpha-\alpha'}{\alpha \mu}}.
\end{align*} 

\end{lemma}
This diameter is bounded by constructing a very dense graph inside 
 the set $\left(K^{\oplus \alpha'}\right)^c$ 
and bounding the distance between any two points in the graph. The result then follows by pushing this path inside   $\left(K^{\oplus \alpha}\right)^c$. 

\begin{proof}[of  Lemma \ref{lemma:OpenSetWithGeodesicFiniteDiameter}]
Since $\left(K^{\oplus \alpha}\right)^c$ is compact, 
we can cover it by a finite number of open balls $\left( \B^\circ (y_i, \frac{\alpha- \alpha'}{2}) \right)_{i \in I_0}$, where $I_0$ is a finite set  and
 $y_i \in \left(K^{\oplus \alpha}\right)^c$.
It is possible to extract an $\varepsilon$-net $\{ y_i, i\in I\}$ from $\{ y_i, i\in I_0\}$, that is $I \subset I_0$ and
\begin{equation} \label{equation:coveringA}
 \left(K^{\oplus \alpha}\right)^c \subset \bigcup_{i \in I} \B^\circ \left(y_i, \alpha- \alpha'\right) \subset  \left(K^{\oplus \alpha'}\right)^c ,
\end{equation}
with 
\begin{equation} \label{equation:packingA}
i,j \in I \: \text{and} \: i\neq j \Rightarrow \B \left(y_i, \frac{\alpha- \alpha'}{4} \right) \cap  \B \left(y_i,\frac{\alpha- \alpha'}{4} \right) =   \emptyset.
\end{equation}
Let us denote the cardinality  of $I$ by $N$.
Since  by \eqref{equation:packingA} the balls $\B \left(y_i,\frac{\alpha- \alpha'}{4} \right)$ 
are disjoint and $\cup_{i\in I} \B \left(y_i,\frac{\alpha- \alpha'}{4} \right)  \subset K^c \subset \B (z, R_{\operatorname{bound}} )$ one has
\begin{equation} \label{equation:UpperBoundN}
N \leq \frac{\vol (\B (z, R_{\operatorname{bound}}) )}{\vol \left( \B \left(0,\frac{\alpha- \alpha'}{4} \right) \right) } 
\leq \left( \frac{ 4 R_{\operatorname{bound}} } {\alpha- \alpha' }\right)^n.
\end{equation}

Let $G$ be the graph with one vertex for each $y_i\in \{1,\ldots,N\}$ and one edge for each pair $(i,j)$  such that 
$\B^\circ (y_i,  \alpha- \alpha')\cap \B^\circ (y_j,  \alpha- \alpha') \neq \emptyset$.
%{\color{blue} This is perhaps a minor issue, but I believe we should pick one notation for the empty set and stick with it. I have a slight preference for $\emptyset$. } 
This graph is connected as otherwise, from \eqref{equation:coveringA},
$\left(K^{\oplus \alpha}\right)^c$ would be the disjoint union of two non-empty disjoint open sets, which would contradict its connectedness.

Consider $x,x' \in \left(K^{\oplus \alpha}\right)^c$ and let  $k,k'$ be such that $x\in \B^\circ (y_{k}, \alpha- \alpha')$ and  $x'\in \B^\circ (y_{k'}, \alpha- \alpha')$.
We define a piecewise linear path  $\mathbf{P}$ from $x$ to $x'$ as follows.
First let $y_{k}= y_{i_0} \ldots y_{i_M}= y_{k'}$ a shortest path between $y_k$ and $y_{k'}$ in the graph $G$.
We define the piecewise linear path $\mathbf{P}$ in  $\left(K^{\oplus \alpha'}\right)^c$ as the concatenation  of the segments 
\[
[x y_{i_0}], [y_{i_0} y_{i_1}],\ldots, [y_{i_{M-1}} y_{i_M}],[y_{i_M},x'].
\]
The length of $P$ is then lower bounded: $L(\mathbf{P}) \leq 2(N+1) (\alpha - \alpha')$.
By definition of $r_\mu^{\alpha'}(K)$, since $r_\mu^{\alpha'}(K) > \alpha$ we have that 
$\| \nabla_K (x) \| \geq \mu$ for $x \in \left(K^{\oplus \alpha'}\right)^c \setminus \left(K^{\oplus \alpha} \right)^c$,
so that $\PushedPath{P}{\frac{\alpha-\alpha'}{\mu}} \subset  \left(K^{\oplus \alpha} \right)^c$.
Applying Lemma \ref{lemma:PushedPathLengthExpansion} we get:
\[
L\left(\PushedPath{P}{\frac{\alpha-\alpha'}{\mu}}  \right) \leq 2   \frac{\alpha-\alpha'}{\mu} + L(P) e^{\frac{\alpha-\alpha'}{\alpha \mu}}.
\]
\end{proof}
%
%\begin{lemma}\label{lemma:PushedPathLengthExpansion}
%Let $K\subset \R^n$ be a closed  set 
%and $\alpha > 0$. Let $\gamma:[0,1] \rightarrow \left(K^{\oplus \alpha}\right)^c$ be a rectifiable path with length $L(\gamma)$ and $T\geq0$.
%$\PushedPath{\gamma}{T}$ is rectifiable and has length $L\left(\PushedPath{\gamma}{T}  \right)$ upper bounded by:
%\[
%L\left(\PushedPath{\gamma}{T}  \right) \leq 2 \, T + L(\gamma) e^{\frac{T}{\alpha}}.
%\]
%\end{lemma}

\begin{theorem}\label{theorem:AXLambdaAlphaIsGeodesicAndFiniteDiameter}
Let $K\subset \R^n$ be the complement of a bounded, open  set $K^c$
and $\alpha, \lambda > 0$ such that $r_\mu^{\alpha}(K) > \alpha + \lambda$ 
and assume $\left(K^{\oplus \alpha}\right)^c$ to be connected with finite geodesic diameter:
 \[
 \Gdiam \left(\left(K^{\oplus \alpha}\right)^c \right) < \infty.
 \]

Then $\ax_\lambda^\alpha(K)$ is a geodesic space with finite geodesic diameter given by \eqref{equation:DiameterAXLambdaAlpha} below. 
More precisely,
if $x,x' \in \ax_\lambda^\alpha(K)$, then there exists a path $\gamma : [0,1] \rightarrow \ax_\lambda^\alpha(K)$ of minimal length 
 such that $\gamma(0) = x$ and $\gamma(1) = x'$. 
This path satisfies, 
\begin{align}
 \length(\gamma) 
&\leq \Gdiam( \ax_\lambda^\alpha(K)) 
\nonumber \\  &
\leq 2 \frac{R_{\max}}{\tilde{\mu}^2} +  \Gdiam \left(\left(K^{\oplus \alpha}\right)^c\right) e^{\frac{R_{\max}}{\alpha \tilde{\mu}^2}}, \label{equation:DiameterAXLambdaAlpha}
\end{align}
where $R_{\max} = R_{\max}(K)$, $\tilde{\mu} = \tilde{\mu}_{\mu,\lambda}^{\alpha, \frac{\alpha}{2}}$.
The length $\length(\gamma)$ is the geodesic distance between $x$ and $x'$ in $\ax_\lambda^\alpha(K)$.
\end{theorem}
%%%%%%%%%%%%%%%%%%%%%%%%%%%%%%%%%%%%%%%%%%

%{\color{red} 
The proof follows by pushing paths into the medial axis, using Lemma \ref{lemma:PushedPathLengthExpansion}. 
%}

\begin{proof} %[of Theorem \ref{theorem:AXLambdaAlphaIsGeodesicAndFiniteDiameter}]
Consider $x,x' \in \ax_\lambda^\alpha(K) \subset \left(K^{\oplus \alpha}\right)^c$ and $\Gamma$ a path in $\left(K^{\oplus \alpha}\right)^c$
from $x$ to $x'$ with length upper bounded by \newline $\Gdiam \left(\left(K^{\oplus \alpha}\right)^c\right)$. 

As in Lemma \ref{lemma:BoundOnLengthToLambdaPlusDelta} we set $\tilde{\mu} = \tilde{\mu}_{\mu,\lambda}^{\alpha, \alpha}$.
Using Definition \ref{definition:PushedPaths} we consider now the path  $\PushedPath{\Gamma}{ \frac{R_{\max}}{\tilde{\mu}^2} }$
with  end points $x$ and $x'$.
Because $\Gamma$ is included in $\left(K^{\oplus \alpha}\right)^c$, Lemma \ref{lemma:PushedPathLengthExpansion} gives: %, using \eqref{equation:LengthOfP_2}:
\begin{equation}\label{equation:BoundOnDiameter}
 L\left(\PushedPath{\Gamma}{ \frac{R_{\max}}{\tilde{\mu}^2} }\right)  \leq  2 \frac{R_{\max}}{\tilde{\mu}^2} +  L(\Gamma) e^{\frac{R_{\max}}{\alpha \tilde{\mu}^2}}
  %\leq \frac{2 R_{\max}}{\tilde{\mu}^2} +   \Gdiam \left(\left(K^{\oplus \alpha}\right)^c \right) e^{\frac{R_{\max}}{\alpha \tilde{\mu}^2}}.
\end{equation}

Due to \eqref{equation:AXLambdaAlphaStableUnderFlow}, $\PushedPath{\Gamma}{ \frac{R_{\max}}{\tilde{\mu}^2} }\left( [0, 1/3]\right)$
 and $\PushedPath{\Gamma}{ \frac{R_{\max}}{\tilde{\mu}^2} }\left( [2/3, 1]\right)$  are inside $\ax_\lambda^\alpha(K)$.
As a consequence of Lemma \ref{lemma:KCFlowsToLambdaPlusDelta}, $\PushedPath{\Gamma}{ \frac{R_{\max}}{\tilde{\mu}^2} }\left( [1/3, 2/3]\right)$
also lies in $\ax_\lambda^\alpha(K)$.

This upper bound \eqref{equation:BoundOnDiameter} on the length is independent of $x,x' \in \ax_\lambda^\alpha(K)$ and it therefore gives an upper bound on the 
geodesic diameter of $\ax_\lambda^\alpha(K)$.

It is known that  if a finite length path exists between two points in a compact subset of Euclidean space, then there exists a path of minimal length, see the second paragraph of Part III, Section 1: {\it Die Existenz geodätischer Bogen in metrischen R{\"a}umen},  in \cite{menger1930untersuchungen}.
  \end{proof}

As a consequence of Lemma \ref{lemma:OpenSetWithGeodesicFiniteDiameter} and Theorem \ref{theorem:AXLambdaAlphaIsGeodesicAndFiniteDiameter}
we get:
\begin{corollary}\label{corollary:AXLambdaAlphaIsGeodesicAndFiniteDiameter}
Let $K\subset \R^n$ be the complement of a bounded, open  set $K^c$,
 $\alpha> \alpha'>0$ and $\lambda > 0$ such that $r_\mu^{\alpha'}(K) > \alpha + \lambda$ and  assume $\left(K^{\oplus \alpha}\right)^c$ to be connected. 

Then $\ax_\lambda^\alpha(K)$ is a geodesic space with finite geodesic diameter.
\end{corollary}

%%%%%%%%%%%%%%%%%%%%
\subsection{The \texorpdfstring{$(\lambda,\alpha)$}{(lambda, alpha)}-medial axis is Gromov-Hausdorff-stable under \texorpdfstring{$\lambda$}{lambda}  perturbation}\label{section:GHStableUnderLamdaPerturbation}

Thanks to Theorem \ref{theorem:AXLambdaAlphaIsGeodesicAndFiniteDiameter}, $\ax_\lambda^\alpha(K)$ is a
geodesic space with finite length inside connected components. Equipped with this geodesic distance  $\ax_\lambda^\alpha(K)$
is a metric space which is stable under $\lambda$ perturbation in the following sense:

\begin{lemma}\label{lemma:AxLambdaAlphaGHStableWRTLambda}

Let  $K\subset \R^n$ be the complement of a bounded,  open  set $K^c$ and 
$\alpha,  \delta, \lambda > 0$, $\mu \in (0,1]$
such that, for some $\alpha' < \alpha$ one has $r_\mu^{\alpha'} (K) > \alpha + \lambda + \delta$ and  $\left(K^{\oplus \alpha}\right)^c$ is connected.
%
%{\color{blue} Since we use here  Theorems \ref{theorem:AxLambdaAlphaHasRightHomotopyType} and \ref{theorem:AXLambdaAlphaIsGeodesicAndFiniteDiameter} in the proof,
% we may need a slightly stronger assumption:
% $r_\mu^{\alpha'} (K) > \alpha + \lambda + \delta$, for some $\alpha'< \alpha$ or
%  $r_\mu^{\alpha/2} (K) > \alpha + \lambda + \delta$ ? }

{Then, the map
%\[
$
\lambda \mapsto \ax_{\lambda}^\alpha (K)
%\]
$
is locally Lipschitz 
 for the Gromov-Hausdorff  distance with respect to the intrinsic metric.

 More precisely, }one has
\[
d_{\textrm{GH}} (\ax_{\lambda+\delta}^\alpha (K), \ax_{\lambda}^\alpha (K)) \leq 2 T + D \left( e^\frac{ T }{\alpha}  - 1 \right) = O\left( \delta\right),
\]
 with
\[
T = \frac{R_{\max}^2}{\alpha \lambda \tilde{\mu}^2} \delta.
\]
 and $D<\infty$ is the geodesic diameter of $\ax_{\lambda}^\alpha (K)$.
%\[
%D = 2 \frac{R_{\max}}{\tilde{\mu}^2} +  \alpha  \left(\frac{ \vol(K^c)}{V_N  \, \left( \frac{\alpha}{4}\right)^n}+1\right)  e^{\frac{2 R_{\max}}{\alpha \tilde{\mu}^2}},
%\]
%where $R_{\max} = R_{\max}(K)  < \infty$ and $\tilde{\mu} = \tilde{\mu}_{\mu,\lambda+ \delta}^{\alpha, \alpha}(K) > 0$.

\end{lemma}
%%%%%%%%%%%%%%%%%%%%

{The core of the proof consists of pushing the geodesics in $\ax_{\lambda}^\alpha$ into  $\ax_{\lambda +\delta}^\alpha$, which is again achieved by the flow. The other properties we need to verify to establish a bound on the Gromov-Hausdorff distance are relatively straightforward.} 
\begin{proof} %[of Lemma \ref{lemma:AxLambdaAlphaGHStableWRTLambda}]
The value of the geodesic diameter $D$ is given by Theorem \ref{theorem:AXLambdaAlphaIsGeodesicAndFiniteDiameter}.
In order to upper bound the Gromov-Hausdorff distance between $\ax_{\lambda+\delta}^\alpha (K)$ and $\ax_{\lambda}^\alpha (K)$
we use Definition \ref{definition:GromovHausdorffDistance}.

Consider the  relation $\mathcal{R} \subset \ax_{\lambda}^\alpha (K) \times \ax_{\lambda+\delta}^\alpha (K)$ defined by:
\[
\mathcal{R}  = \left\{ (x_1, x_2) \in \ax_{\lambda}^\alpha (K) \times \ax_{\lambda+\delta}^\alpha (K),
\:  \exists t \in [0, T],\, x_2 = \Phi_K(t,x_1)  \right\},
\]
where $T  = \frac{R_{\max}^2}{\alpha \lambda \tilde{\mu}^2} \delta$.

We have to check the two conditions of Definition  \ref{definition:GromovHausdorffDistance}.
Recall that, by Lemma \ref{lemma:BoundOnLengthToLambdaPlusDelta}, one has
 \[
 \Phi_K \left( T, \ax_\lambda^\alpha (K) \right) \subset  \ax_{\lambda+\delta}^\alpha (K) \subset \ax_\lambda^\alpha (K).
 \]

{
\paragraph{Condition (1): $\mathcal{R}$ is surjective.} This condition follows because, if $x_1 \in \ax_{\lambda}^\alpha (K)$ then $\Phi_K(T, x_1)  \in \ax_{\lambda+\delta}^\alpha (K)$, thanks to
%% \eqref{equation:FlowSendsInLargerLambda} in 
Lemma \ref{lemma:BoundOnLengthToLambdaPlusDelta}. 
This is in turn equivalent to  $\left(x_1 ,  \Phi_K(T, x_1) \right) \in \mathcal{R}$.
Conversely, if $x_2 \in \ax_{\lambda+\delta}^\alpha (K)$, then $x_2 \in \ax_{\lambda}^\alpha (K)$ and  $(x_2, x_2) \in \mathcal{R}$
since $\Phi_K(0, x_2) = x_2$.
}

\paragraph{Condition (2): The bound on the distance distortion.}  
Consider $(x_1, x_2), (x_1', x_2')   \in \mathcal{R}$, with $x_2 = \Phi_K(t, x_1)$ and
 $x_2' = \Phi_K(t', x_1')$ with $t, t' \in [0, T]$.
 Denote  by $d_1$ and $d_2$
the respective intrinsic distances in $\ax_{\lambda}^\alpha (K)$ and $\ax_{\lambda+\delta}^\alpha (K)$. 
{By Theorem \ref{theorem:AxLambdaAlphaHasRightHomotopyType},
since  $\left(K^{\oplus \alpha}\right)^c$ is connected (which is equivalent to be path-wise connected for an open set), $\ax_{\lambda}^\alpha (K)$ is path-wise connected 
and $x_1$ and $x_1'$ are in the same connected component of $\ax_{\lambda}^\alpha (K)$.
Thanks to Theorem \ref{theorem:AXLambdaAlphaIsGeodesicAndFiniteDiameter}, 
 $d_1(x_1,x_1') <\infty$ } and there is a path $\gamma_1: [0,1]\rightarrow \ax_{\lambda}^\alpha (K)$ such that $d_1(x_1,x_1') = L(\gamma_1)$.
Thanks to Lemma \ref{lemma:PushedPathLengthExpansion}, the path $\PushedPath{\gamma_1}{T}$
has length upper bounded by $2 \, T + L(\gamma_1) e^{\frac{T}{\alpha}}$. 

{ %[TODO Get rid of at least one also]} Also, using 
By }Definition \ref{definition:PushedPaths}, one has $\PushedPath{\gamma_1}{T}\left(\frac{t}{3T}\right) = x_2$
and  $\PushedPath{\gamma_1}{T}\left(1 - \frac{t'}{3T} \right)= x_2'$. %Also, from 
{Moreover, \eqref{equation:FAlphaIsINotDecreasing} 
and \eqref{equation:FlowSendsInLargerLambda} yields}  %we have
\[
\PushedPath{\gamma_1}{T} \left( \left[\frac{t}{3T}, 1 - \frac{t'}{3T}\right] \right) \subset \ax_{\lambda}^\alpha (K),
\]
so that $\PushedPath{\gamma_1}{T} \left( \left[\frac{t}{3T}, 1 - \frac{t'}{3T}\right] \right)$ 
is a path from $x_2$ to $x_2'$ inside $\ax_{\lambda+\delta}^\alpha (K)$
and its length is therefore lower bounded by $d_2(x_2,x_2')$. It follows that
\begin{align*}
d_2(x_2,x_2') &\leq L\left( \PushedPath{\gamma_1}{T}\left( \left[\frac{t}{3T}, 1 - \frac{t'}{3T}\right] \right)  \right)\\
& \leq L\left( \PushedPath{\gamma_1}{T} \right) \leq 2 \, T + d_1(x_1, x_1')  e^{\frac{T}{\alpha}},
\end{align*}
so that
\begin{equation}\label{equation:d2BoundedFromd1}
d_2(x_2,x_2')  - d_1(x_1, x_1') \leq 2 \, T +  d_1(x_1, x_1') \left( e^\frac{ T }{\alpha}  - 1 \right)
\leq 2 \, T +  D \left( e^\frac{ T }{\alpha}  - 1 \right).
\end{equation}
{Since $\ax_{\lambda+\delta}^\alpha (K)$ is connected, there is a path 
 $\gamma_2: [0,1]\rightarrow \ax_{\lambda+\delta}^\alpha (K)$ such that $d_2(x_2,x_2') = L(\gamma_2)$.}

Consider now the path $\Gamma: [0,1] \rightarrow \ax_{\lambda}^\alpha (K)$ defined as:
\begin{equation}
\Gamma (u) = 
\begin{cases}
	\Phi_K(3 u t, x_1) &\textrm{if} \quad u \in [0, 1/3] \\
	\gamma_2(3u-1 )) &\textrm{if}\quad u \in [1/3, 2/3] \\
	\Phi_K((3(1- u) t', x_1' ) &\textrm{if} \quad u \in [2/3, 1] .
     \end{cases}
\nonumber
\end{equation}
{ Here we used that $\ax_{\lambda+ \delta}^{\alpha} (K) \subset \ax_{\lambda}^\alpha$ for $u \in [1/3, 2/3]$.
}  
By \eqref{equation:FAlphaIsINotDecreasing} 
$\Gamma$ is  a path from $x_1$ to $x_1'$ inside $\ax_{\lambda}^\alpha (K)$ so that
 \[
d_1(x_1,x_1') \leq  L\left( \Gamma \right) \leq 2 \, T + L\left( \gamma_2 \right) = 2 \, T + d_2(x_2,x_2')
 \]
 and with \eqref{equation:d2BoundedFromd1} we get as { required} %MW: this was purely a matter of taste and there was nothing wrong with the previous
 \[
\left| d_2(x_2,x_2')  - d_1(x_1, x_1') \right| 
\leq 2 \, T +  D \left( e^\frac{ T }{\alpha}  - 1 \right) .\]
Note that by Corollary \ref{corollary:AXLambdaAlphaIsGeodesicAndFiniteDiameter} (using the assumption $r_\mu^{\alpha'} (K) > \alpha + \lambda + \delta$) the geodesic diameter $D$
of $\ax_{\lambda}^\alpha (K)$ is finite.
\end{proof}

%%%%%%%%%%%%%%%%%%%%
\subsection{The \texorpdfstring{$(\lambda,\alpha)$}{(lambda, alpha)}-medial axis is Gromov-Hausdorff-stable 
under  \texorpdfstring{$\alpha$}{ alpha} perturbation}\label{section:GHStableUnderAlphaPerturbation}
{ Using almost identical arguments as in the previous section we also get the Gromov-Hausdorff stability with respect the offset $\alpha$. 
}

\begin{lemma}\label{lemma:AxLambdaAlphaGHStableWRTAlpha}
Let $K\subset \R^n$ be the complement of a bounded open  set $K^c$ and 
$\alpha,  \delta,  \lambda > 0$, $\mu \in (0,1]$
such that,  for some $\alpha' < \alpha$ one has $r_\mu^{\alpha'} (K) > \alpha + \lambda + \delta$  and $\left(K^{\oplus \alpha}\right)^c$ is connected.
Then, the map
\[
\alpha \mapsto \ax_{\lambda}^\alpha (K)
\]
is { locally Lipschitz }
 for the Gromov-Hausdorff  distance with respect to the intrinsic metric.

 More precisely, one has
\[
d_{\textrm{GH}} (\ax_{\lambda}^{\alpha+\delta} (K), \ax_{\lambda}^\alpha (K)) \leq 2 T + D \left( e^\frac{ T }{\alpha}  - 1 \right) = O\left( \delta\right),
\]
with
\[
T = \frac{R_{\max}}{\alpha  \tilde{\mu}^2} \delta,
\]
and $D< \infty$ is the geodesic diameter of $\ax_{\lambda}^\alpha (K)$.
%\[
%D = 2 \frac{R_{\max}}{\tilde{\mu}^2} +  \alpha  \left(\frac{ \vol(K^c)}{V_N  \, \left( \frac{\alpha}{4}\right)^n}+1\right)  e^{\frac{2 R_{\max}}{\alpha \tilde{\mu}^2}},
%\]
%where $R_{\max} = R_{\max}(K)  < \infty$ and $\tilde{\mu} = \tilde{\mu}_{\mu,\lambda+ \delta}^{\alpha, \alpha}(K) > 0$.
 \end{lemma}

\begin{proof}%[of Lemma \ref{lemma:AxLambdaAlphaGHStableWRTAlpha}] 
This is similar to Section \ref{section:GHStableUnderLamdaPerturbation}
  using \eqref{equation:FlowSendsInLargerAlpha} instead of \eqref{equation:FlowSendsInLargerLambda}.
\end{proof}

\section{Hausdorff stability of the \texorpdfstring{$(\lambda,\alpha)$}{(lambda, alpha)}-medial 
axis under Hausdorff perturbation of \texorpdfstring{$K$}{K}}\label{section:HausdorffStability}

In this section, we prove one of the main stability theorems of this paper, namely stability in the Hausdorff sense of the $(\lambda,\alpha)$-medial axis under Hausdorff perturbations of $K$. This requires some further results on the flow that are proven in the first subsections, while the main result is proven in the final subsection.  
%{\color{red} [I added a more extensive version of the previous sentence, which also serves as an overview of the results: It is not strictly necessary, but some of the reviewers complained] 
As in the previous section we use the flow to establish our main result, namely the Hausdorff stability. Intuitively this may seem straightforward, because near the medial axis the flow we follow points towards the medial axis. However, establishing that the flow is fast enough requires a number of technical estimates. Firstly we need that the distance to the closest points ($R_K$) increases sufficiently fast as we follow the flow. This is proven in Section \ref{sec:lowBndIncreaseRk}. Based on this result we can prove that
\begin{itemize} 
\item Points close to $\ax_{\lambda}^\alpha(K)$ flow inside it after a short amount of time (Section \ref{subsec:NearPointFlowIntoAxQuickly}).
\item If you perturb $K$ into $K'$ (near in Hausdorff distance) then $\ax^{\alpha}_{\lambda}(K')$ flows into $\ax_{\lambda-\delta}^{\alpha} (K)$ after a short amount of time (Section \ref{subsec:QuickFlow2}).   
\end{itemize}
The bound on the Hausdorff distance is finally established based on this and Lemma \ref{lemma:BoundOnLengthToLambdaPlusDelta}. % }

%%%%%%%%%%%%%%%%%%%%%%%%%%%%%%%%%%%%%%%%%%%%%%%%%
\subsection{A lower bound on \texorpdfstring{$R_K$}{Rk} along the flow trajectories} \label{sec:lowBndIncreaseRk}

The technical result that underpins the lemma in this section is the Volterra integral inequality, as discussed in Section \ref{sec:Volterra}.  
%}
%%%%%%%%%%%%%%%%%%%%%%%%%%%%%%%%%%%%%%%%%%%%%
\begin{lemma}\label{lemma:RadiusTightLowerBoundWhenTrajectoryNotInAXLambdaAlpha}
Let $K\subset \R^n$ be the complement of a bounded open  set $K^c$, 
$\alpha \geq 0$ and $\lambda > 0$.
Consider $y \in    \left( K^{\oplus \alpha} \right)^c \setminus \ax_\lambda^\alpha(K)$ and 
denote by $s\mapsto y(s)$ the trajectory of $t \mapsto \Phi_K \left( t, y \right)$, parametrized by arc length. We stress that $y(0)=y$. 
Assume that $R_K(y) -\alpha > \lambda$ and that for some $S>0$ one has
\[
y(S) \notin \ax_{\lambda}^\alpha (K),
\]  
then 
\begin{equation}
 \left( R_K(y(S)) - \alpha \right)^2  \geq   (S_0+S)^2 +  \lambda^2,
\nonumber
\end{equation}
 where $S_0 > 0$ is defined as
\begin{equation}
S_0^2 = \left(R_K(y)-\alpha \right)^2 - \lambda^2.
\nonumber
\end{equation}
\end{lemma}

%%%%%%%%%%%%%%%%%%%%%%%%%%%%%%%%%%%%%%%%%%%%%%%%%

%An alternate proof of this lemma is also given in Appendix \ref{section:RadiusTightLowerBoundWhenTrajectoryNotInAXLambdaAlphaAlternateProof}.

\begin{proof} %[of Lemma \ref{lemma:RadiusTightLowerBoundWhenTrajectoryNotInAXLambdaAlpha}]
We define the map $R_0: [0,S]\rightarrow \R$ as
\[
R_0(s)  =  \sqrt{ (S_0+s)^2 + \lambda^2 }{  + \alpha}. 
\]
We need to prove that $R_K(y(s)) \geq R_0(s)$, which we'll do by means of Theorem \ref{VolterraInequality}. 
By definition of $S_0$, one has $R_0(0) = R_K(y)$.
For $s\in [0,S]$ we get 
\[
\frac{d}{ds} (R_0(s) - \alpha) ^2 = 2 (S_0+s).
\]
It follows that
\begin{align}
\frac{d}{ds} R_0(s)  &= \frac{ (S_0+s) }{R_0(s) - \alpha} 
\nonumber 
\\
&= \frac{\sqrt{(R_0(s) - \alpha) ^2 -   \lambda^2  }}{R_0(s) - \alpha}  \nonumber 
\\
&= \sqrt{ 1 - \left( \frac{\lambda }{ R_0(s) - \alpha } \right)^2 }
\nonumber
\end{align}
and thus we have the Volterra integral inequality, 
\begin{align}
R_0(s) = R_0(0) + \int_{0}^s k(R_0(\tau)) \ud \tau,
\nonumber
\end{align}
where the kernel $k$ is
\begin{align} 
k(x) = \sqrt{ 1 - \left( \frac{\lambda }{ x - \alpha } \right)^2 }.
\nonumber
\end{align} 
Combining \eqref{equation:DefinitionFlow}, and \eqref{equation:DerivativeR} 
gives
\begin{align}
\frac{d}{ds^+} R_K( y(s) )  &=  \left\|  \nabla_K \left( y(s) \right) \right\|   \nonumber  \\
&= \sqrt{ 1 -  \left( \frac{\F^\alpha(y(s)) }{ R_K(y(s)) - \alpha } \right)^2 },
\nonumber
\end{align} 
see also \cite[Equation (5)]{cl2005lambda}. Because $ R_K( y(s) ) $ is Lipschitz, Theorem~\ref{theorem:FundamentalThmCalculusACFunctions} yields that 
\begin{align} 
R_K( y(s) )&= R_K( y )+  \int_0^s \sqrt{ 1 -  \left( \frac{\F^\alpha(y(\tau)) }{ R_K(y(\tau)) - \alpha } \right)^2 } \ud \tau 
\nonumber 
\\
&= R_0(0)+  \int_0^s \sqrt{ 1 -  \left( \frac{\F^\alpha(y(\tau)) }{ R_K(y(\tau)) - \alpha } \right)^2 } \ud \tau.
\nonumber
\end{align} 
By assumption $y(S) \notin \ax_{\lambda}^\alpha$, so \eqref{equation:FAlphaIsINotDecreasing} implies 
that $s \leq S \Rightarrow \F^\alpha(y(s)) < \lambda$. 
Moreover, for sufficiently small $0<\delta'\leq S$ we can assume that $\F^\alpha(y(s)) \leq \lambda'<\lambda$, for $s\in [0,\delta]$, by Lemma \ref{lemma:FIsINotDecreasing}.
This means that $R_K( y(s) )$ satisfies the following integral inequality of Volterra type,
\begin{align} 
R_K( y(s) ) &\geq R_0(0)+  \int_0^s \sqrt{ 1 -  \left( \frac{ \lambda}{ R_K(y(\tau)) - \alpha } \right)^2 } \ud \tau 
\nonumber 
\\
& = R_0(0) + \int_{0}^s k(R_K(y(\tau))) \ud \tau,
\nonumber
\end{align} 
where the equality occurs only when $s=0$. We note that, because $R_K(y) -\alpha > \lambda>0$ and $R_K(y)$ is monotone, 
\begin{align} 
\frac{d}{dx} k(x) = \frac{\lambda ^2}{(x-\alpha )^3 \sqrt{1-\frac{\lambda ^2}{(x-\alpha )^2}}} \geq 0
\nonumber
\end{align} 
on the domain and thus the kernel is monotone. Because $R_K( y(s) )$ is 1-Lipschitz in $s$ and we can assume that there is some $\delta'\leq \delta$ such that for all $s\in (0,\delta')$, we have $R_K( y(s) )> R_0(s)$. In fact $\delta'$ is determined by the condition
\begin{align} 
\sqrt{ 1 -  \left( \frac{ \lambda'}{ R_0(0)-\delta' - \alpha } \right)^2 }&\geq \sqrt{ 1 -  \left( \frac{ \lambda}{ R_0(0)+\delta' - \alpha } \right)^2 }
\nonumber
\\
\frac{ \lambda'}{ R_0(0)-\delta' - \alpha } &\leq \frac{ \lambda}{ R_0(0)+\delta' - \alpha }.
\nonumber
\end{align} 
The result is now a direct consequence of the application of Theorem \ref{VolterraInequality}. 
\end{proof}

%%%%%%%%%%%%%%%%%%%%%%%%%%%%%%%%%%%%%%%%%%%%%%%%%

\subsection{Points close to \texorpdfstring{$\ax_\lambda^\alpha(K)$}{lambda alpha medial axis of K}  flow  inside it after a short time}
\label{subsec:NearPointFlowIntoAxQuickly} 
%%%%%%%%%%%%%%%%%%%%%%%%%%%%%%%%%%%%%%%%%%%%%
\begin{lemma}\label{lemma:EpsilonNeighborFlowsToLambdaMinusDelta}
Let $K\subset \R^n$ be the complement of a bounded open  set $K^c$ and 
$\alpha > 0$, $\mu >0$ and $\lambda > 0$,
such that $r_\mu^\alpha (K)  > \alpha + \lambda$.
Then, if $\delta <  \lambda$ and $\epsilon < \min \left(2 \alpha,   \frac{\left( 2 \lambda -\delta\right)  \delta}{8 R_{\max}} \right)$ one has
\[
\Phi_K \left(\frac{8 R_{\max}^2}{\left( 2 \lambda -\delta\right)  \delta \tilde{\mu}}  \:  \epsilon \, , \: \ax_\lambda^\alpha(K)^{\oplus \epsilon}  \right) \subset \ax_{\lambda-\delta}^\alpha (K) ,
\]
where $R_{\max} = R_{\max}(K) < \infty$ and $\tilde{\mu} = \tilde{\mu}_{\mu,\lambda}^{\alpha, \alpha}>0$.
Moreover, for $y \in \ax_\lambda^\alpha(K)^{\oplus \epsilon}$,
the length of the trajectory $\Phi_K \left( \left[0, \frac{8 R_{\max}^2}{ \left( 2 \lambda -\delta\right)  \delta \tilde{\mu}} \:  \epsilon \right] , y \right)$ is upper bounded by
\[
\frac{8 R_{\max}^2}{\left( 2 \lambda -\delta\right)  \delta}\:  \epsilon.
\]
\end{lemma}

\begin{proof} %[of Lemma \ref{lemma:EpsilonNeighborFlowsToLambdaMinusDelta}]
Consider $y \in    \ax_\lambda^\alpha(K)^{\oplus \epsilon}$ and $a \in \ax_\lambda^\alpha(K)$ such that $\| y-a\| \leq \epsilon$.
Denote by $s\mapsto y(s)$ the trajectory of $t \mapsto \Phi_K \left( t, y \right)$, parametrized by arc length, 
so that $y(0)=y$. Let us now assume that for some $s>0$ one has
\[
y(s) \notin \ax_{\lambda-\delta}^\alpha (K).
\]
{ Since $a \in \ax_\lambda^\alpha(K)$ \eqref{equation:InclusionLambdaAlphaMedialAxe_B} yields that $a \in \ax_\lambda(K)$. The bound \eqref{ShortLOWbndRKandFK} gives that $R_K (a) \geq F_K(a) $. This together with \eqref{Eq:defAxAlphaLambdaK} and \eqref{equation:ExpressionF_Alpha} gives
\begin{equation}\label{equation:R_KOfALowerBounded}
R_K(a) -\alpha \geq \lambda.
\end{equation}
} 
By the conditions of the theorem we have $\epsilon < \delta$ and $\| y-a\| \leq \epsilon$, \eqref{equation:R_KOfALowerBounded} therefore implies 
\begin{equation}\label{equation:RYMinusAlphaGreaterThanLambdaMinusDelta}
R_K(y) -\alpha > \lambda - \delta,
\end{equation}
by the triangle inequality. This means that the condition of  Lemma \ref{lemma:RadiusTightLowerBoundWhenTrajectoryNotInAXLambdaAlpha} are satisfied. 

We can then apply Lemma \ref{lemma:RadiusTightLowerBoundWhenTrajectoryNotInAXLambdaAlpha} with $\lambda $ replaced by $\lambda - \delta$,
which gives us
 \begin{equation}\label{equation:LowerBoundRK}
\left( R_K(y(s)) - \alpha \right)^2  \geq   (S_0+s)^2 +  (\lambda - \delta)^2,
 \end{equation}
where
\begin{equation}\label{equation:DefinitionS0}
S_0^2 = \left(R_K(y)-\alpha \right)^2 - (\lambda - \delta)^2.
\end{equation}

Since $R_K(a) -\alpha  \geq \lambda$  there is  $S_a \geq 0$
such that
 \begin{equation}\label{equation:DefinitionSA}
S_a^2 + \lambda^2 = \left( R_K(a) -\alpha \right)^2.
 \end{equation}
%{\color{purple} Instead of :
%
%We recall Lemma 4.15 in \cite{LIEUTIERhomotopytype}, where we chose $\mathcal{O}$ to be equal to $\left(K^{\oplus \alpha}\right)^c$,
% $x$ to be $a$ and $y$ to be $y(s)$,
% \begin{equation}\label{equation:UpperBoundRK_1}
%( R_K(y(s)) -\alpha)^2 \leq     ( S_a + \|y(s)-a\|)^2   +    \lambda^2.
% \end{equation}
%}

  We recall Lemma 4.15 of \cite{LIEUTIERhomotopytype}. The correspondence between the notation is the following: 

%whose inequality can rewritten
%\[
%\mathscr{R}(y)^2 \leq  \left(\sqrt{\mathscr{R}(x)^2 - \mathscr{F}(x)^2} + d(x,y)\right)^2 + \mathscr{F}(x)^2
%\]
%using the notations of \cite{LIEUTIERhomotopytype}.
 
We pick $\mathcal{O}$ to be equal to $\left(K^{\oplus \alpha}\right)^c$,
 $x$ to be $a$ and $y$ to be $y(s)$, then $\mathscr{R}(x)$ is $R_K(a) -\alpha$, 
 $\mathscr{R}(y)$ is $R_K(y(s)) -\alpha$, and $\mathscr{F}(x)$ is $ \F_K^{\alpha} (a)$.   So that
  the inequality of Lemma 4.15 of \cite{LIEUTIERhomotopytype} reads (using our notation):
\begin{align*}%\label{equation:UpperBoundRK_1}
&( R_K(y(s)) -\alpha)^2 
\\ &\leq    (R_K(a) -\alpha)^2 + 2 \|y(s)-a\| \sqrt{(R_K(a) -\alpha)^2 -  \F_K^{\alpha} (a)^2}   +  \|y(s)-a\|^2 .
 \end{align*} 
Using \eqref{equation:DefinitionSA}  we have:
 \begin{align*}
 &( R_K(y(s)) -\alpha)^2 
\\
&\leq \lambda^2   +  S_a^2 + 2 \|y(s)-a\| \sqrt{(R_K(a) -\alpha)^2 -  \F_K^{\alpha} (a)^2}    +  \|y(s)-a\|^2 .
\end{align*}
and since   $\F_K^{\alpha} (a) \geq \lambda$ one has, using  \eqref{equation:DefinitionSA} again, 
\[\sqrt{(R_K(a) -\alpha)^2 -  \F_K^{\alpha} (a)^2}  \leq S_a\] 
and
 \begin{equation}%\label{equation:UpperBoundRK_1}
( R_K(y(s)) -\alpha)^2 \leq     ( S_a + \|y(s)-a\|)^2   +    \lambda^2.
\nonumber
 \end{equation}

Because $\|y(s)-a\| \leq \|y-a\| + \|y(s)-y\| \leq \epsilon +s$, we get
 \begin{equation}\label{equation:UpperBoundRK_2}
( R_K(y(s)) -\alpha)^2 \leq     ( S_a + s + \epsilon)^2   +    \lambda^2.
 \end{equation}
Combining this with \eqref{equation:LowerBoundRK} yields,
\[
 ( S_a + s + \epsilon)^2   +    \lambda^2  \geq  (S_0+s)^2 +  (\lambda - \delta)^2,
\]
which can be rewritten as
%
%{\color{blue} I do not get why the intermediate step helps, going directly to the equation after this seems to be just as hard or I did something stupid checking. TODO MERGE} 
%\[
% ( S_a + S_0 + 2s + \epsilon) (S_a - S_0 + \epsilon)  + (2 \lambda - \delta)\delta   \geq 0
%\]
%that is:
 \begin{equation}\label{equation:ConstraintOnS_1}
2 (S_0 -  S_a  - \epsilon )  s  \leq    (2 \lambda - \delta)\delta - \left( S_0^2 - S_a^2\right) +  \left(2 S_a + \epsilon \right) \,\epsilon. 
\end{equation}  
To recover an upper bound on $s$ from \eqref{equation:ConstraintOnS_1}, we need a lower bound on 
 $(S_0 -  S_a  - \epsilon )$ and an upper bound on the right hand side of the previous inequality, that is an upper bound on $ (2 \lambda - \delta)\delta - \left( S_0^2 - S_a^2\right) +  \left(2 S_a + \epsilon \right) \,\epsilon$.
 
From \eqref{equation:DefinitionS0} and \eqref{equation:DefinitionSA} we get
\begin{eqnarray*}
S_0^2 - S_a^2 &=& \left(R_K(y)-\alpha \right)^2 - (\lambda - \delta)^2 - \left( R_K(a) -\alpha \right)^2 +\lambda^2 \\
&=& \left(R_K(y)-\alpha \right)^2 - \left( R_K(a) -\alpha \right)^2 + (2 \lambda - \delta)\delta   ,
\end{eqnarray*}
so that,
\begin{equation}\label{equation:FirstBoundOnS02MinusSA2MinusTwoLambdaDeltaDelta}
S_0^2 - S_a^2 -   (2 \lambda - \delta)\delta   =  \left(R_K(y)-\alpha \right)^2 - \left( R_K(a) -\alpha \right)^2.
\end{equation}

Since $ | R_K(y) -R_K(a) | < \epsilon$  we get,
\[
| \left(R_K(y)-\alpha \right)^2 - \left( R_K(a) -\alpha \right)^2 | < | R_K(y) +   R_K(a)  - 2 \alpha |   \,\epsilon    \leq 2 R_{\max} \, \epsilon,
\]
where we used that $R_K(y), R_K(a) \leq R_{\max}$ by definition of $R_{\max}$ and $\alpha\leq R_{\max}$. 
Now \eqref{equation:FirstBoundOnS02MinusSA2MinusTwoLambdaDeltaDelta}, in turn  gives,
\begin{equation}\label{equation:SecondBoundOnS02MinusSA2MinusTwoLambdaDeltaDelta}
| (2 \lambda - \delta)\delta  - (S_0^2 - S_a^2)   |  < 2 R_{\max} \, \epsilon
\end{equation}
or, equivalently
\begin{equation}\label{equation:SecondBoundOnS02MinusSA2MinusTwoLambdaDeltaDelta_2}
(2 \lambda - \delta)\delta  - 2 R_{\max} \epsilon  <  S_0^2 - S_a^2 < ( 2 \lambda - \delta)\delta  + 2 R_{\max} \,\epsilon  .
\end{equation}
%Since  $0< \delta \leq \lambda$ one has  $\lambda  \delta \leq (2 \lambda - \delta)\delta \leq   2 \lambda  \delta$ and,
% with the assumption $\epsilon < \frac{ \lambda \delta}{8 R_{\max}}$ of the lemma one gets :
 With the assumption $\epsilon < \frac{(2 \lambda -\delta) \delta}{8 R_{\max}}$, which gives  
 $2 R_{\max} \,\epsilon   <  \frac{1}{4} (2 \lambda -\delta) \delta$,
  \eqref{equation:SecondBoundOnS02MinusSA2MinusTwoLambdaDeltaDelta_2} yields,
\begin{equation}\label{equation:UpperBoundOnS02MinusSA2}
  0 < \frac{3}{4} (2 \lambda -\delta) \delta <  S_0^2 - S_a^2 <  \frac{5}{4} (2 \lambda -\delta) \delta .
\end{equation}
This in turn implies that 
\begin{equation*}
  S_0 - S_a= \frac{S_0^2 - S_a^2 }{S_0+S_a}  >  \frac{3}{4}  \frac{\left(2 \lambda-\delta\right)  \delta}{S_0+S_a  }  .
\end{equation*} 
Combining  \eqref{equation:DefinitionSA} and \eqref{equation:DefinitionS0} one has $ S_0+S_a < 2 R_{\max}$ 
and using again $\epsilon < \frac{\left(2 \lambda-\delta\right)  \delta}{8 R_{\max}}$ we  get
\begin{equation}\label{equation:LowerBoundS0MinusSA_2}
  S_0 - S_a - \epsilon    >   \frac{3}{8}    \frac{\left(2 \lambda-\delta\right)  \delta}{R_{\max} }   - \epsilon  >   \frac{\left(2 \lambda-\delta\right)  \delta}{4 R_{\max} } .
\end{equation}

We have from  $\epsilon < 2 \alpha$ and  \eqref{equation:DefinitionSA}   that  $ 2 S_a  + \epsilon < 2(S_a + \alpha) \leq 2 R_{\max}$. Therefore
\eqref{equation:ConstraintOnS_1} together with \eqref{equation:SecondBoundOnS02MinusSA2MinusTwoLambdaDeltaDelta} gives us
\begin{align}
2 (S_0 -  S_a  - \epsilon )  s  &<   2 R_{\max} \, \epsilon + \left( 2 S_a  + \epsilon \right) \, \epsilon  
\nonumber 
\\
&\leq 4 R_{\max}  \epsilon.
\label{equation:LinearInequationOnS}
 \end{align}
 
\begin{remark}\label{remark:EpsilonLessThanAlphaNotNecessaryInEpsilonNeighborFlowsToLambdaMinusDelta}
Note that the assumption  $\epsilon < 2 \alpha$ is not really necessary as, here,  we could merely upper bound $\epsilon$ by $R_{\max}$,
so that $ 2 S_a  + \epsilon < 3R_{\max}$ and we would get $5 R_{\max}  \epsilon$ instead of $ 4 R_{\max}  \epsilon$ as upper bound in \eqref{equation:LinearInequationOnS}.
\end{remark}

Equations \eqref{equation:LowerBoundS0MinusSA_2} and \eqref{equation:LinearInequationOnS} gives us
\[
s < \frac{8 R_{\max}^2}{\left(2 \lambda-\delta\right)  \delta} \:  \epsilon
\]
We have obtained this inequality by assuming $y(s) \notin \ax_{\lambda-\delta}^\alpha$.
By contraposition we get
\[
y \left( \frac{8 R_{\max}^2}{ \left(2 \lambda-\delta\right)  \delta} \:  \epsilon \right) \in \ax_{\lambda-\delta}^\alpha.
\]
This proves the last statement of the lemma that upper bounds the length of the trajectory.
Since by Lemma \ref{lemma:NablaLowerBoundedByBeta}, as long as 
$y(s) \notin \ax_{\lambda-\delta}^\alpha$ the modulus of the right derivative of 
$t \mapsto \Phi_K \left( t, y \right)$, which is $\frac{ds}{dt} = \| \nabla_K\left(  \Phi_K \left( t, y \right)\right) \| $,  
is lower bounded by $\tilde{\mu}$ we get, still using Theorem \ref{theorem:FundamentalThmCalculusACFunctions},
 the first statement of  the lemma. 
\end{proof}

%%%%%%%%%%%%%%%%%%%%%%%%%%%%%%%%%%%%%%%%%%%%%
%%%%%%%%%%%%%%%%%%%%%%%%%%%%%%%%%%%%%%%%%%%%%

\begin{remark}\label{remark:alphaCanBeZeroInEpsilonNeighborFlowsToLambdaMinusDelta}
{\color{red} } While we do not need it in subsequent proofs, thanks to remark \ref{remark:EpsilonLessThanAlphaNotNecessaryInEpsilonNeighborFlowsToLambdaMinusDelta}, we could omit the condition
$\epsilon < 2 \alpha$ in Lemma \ref{lemma:EpsilonNeighborFlowsToLambdaMinusDelta}, so that the lemma holds as well without this condition and then also for $\alpha =0$,
that is for the $\lambda$-medial axis, at the mild price of a larger constant, replacing the flow time $\frac{8 R_{\max}^2}{\left( 2 \lambda -\delta\right)  \delta \tilde{\mu}}  \:  \epsilon$
by $\frac{10 R_{\max}^2}{\left( 2 \lambda -\delta\right)  \delta \tilde{\mu}}  \:  \epsilon$.
\end{remark}

%%%%%%%%%%%%%%%%%%%%%%%%%%%%%%%%%%%%
%%%%%%%%%%%%%%%%%%%%%%%%%%%%%%%%%%%%
%%%%%%%%%%%%%%%%%%%%%%%%%%%%%%%%%%%%
%%%%%%%%%%%%%%%%%%%%%%%%%%%%%%%%%%%%
\subsection{Flow to the medial axis after a perturbation of the set \texorpdfstring{$K$}{K}. 
} \label{subsec:QuickFlow2}
%%%%%%%%%%%%%%%%%%%%%%%%%%%%%%%%%%%%%%%%%%%%%

We write $d_H(C_1,C_2)$ for the Hausdorff distance between two compact sets $C_1, C_2 \subset  \R^n$.

%%%%%%%%%%%%%%%%%%%%%%%%%%%%%%%%%%%%
%%%%%%%%%%%%%%%%%%%%%%%%%%%%%%%%%%%%
\begin{lemma}\label{lemma:AXLambdaAlphaOfKPrimeFlowsToAXLambdaMinusDeltaAlphaOfK}
Let $K, K'\subset \R^n$ be  complements of bounded open  sets $K^c$ and $K'^c$ and 
$\alpha \geq 0$ ,$\mu >0$, $\lambda > 0$,
such that $r_\mu^\alpha (K)  > \alpha + \lambda$.

If
$d_H( K,  K') < \epsilon $
 then, if $\delta <  \lambda$ and $\epsilon < \frac{\left( 2 \lambda -\delta\right)  \delta}{8 R_{\max}}$ one has:
%$%
\[
\Phi_K \left(\frac{8 R_{\max}^2}{\left( 2 \lambda -\delta\right)  \delta \tilde{\mu}}  \:  \epsilon \, , \: \ax_\lambda^\alpha(K')  \right) \subset \ax_{\lambda-\delta}^\alpha(K),
%$ %
\]
where $R_{\max} = {\max \{ R_{\max}(K) ,R_{\max} (K') \}  } < \infty$
 and $\tilde{\mu} = \tilde{\mu}_{\mu,\lambda}^{\alpha, \alpha}>0$.

Moreover, for $y \in \ax_\lambda^\alpha(K')$,
the length of the trajectory \newline $\Phi_K \left( \left[0, \frac{8 R_{\max}^2}{ \left( 2 \lambda -\delta\right)  \delta \tilde{\mu}} \:  \epsilon \right] , y \right)$ is upper bounded by
\[
\frac{8 R_{\max}^2}{\left( 2 \lambda -\delta\right)  \delta}\:  \epsilon.
\]
\end{lemma}
%%%%%%%%%%%%%%%%%%%%%%%%%%%%%%%%%%%%

\begin{proof} %[ of Lemma \ref{lemma:AXLambdaAlphaOfKPrimeFlowsToAXLambdaMinusDeltaAlphaOfK}]
The proof is similar to the proof of Lemma \ref{lemma:EpsilonNeighborFlowsToLambdaMinusDelta}, except that now we start with a point in 
$\ax_\lambda^\alpha(K')$ and flow to $\ax_{\lambda-\delta}^\alpha(K)$.
We consider $y  \in \ax_\lambda^\alpha(K')$.
Denote by $s\mapsto y(s)$ the trajectory of $t \mapsto \Phi_K \left( t, y \right)$, parametrized by arc length, 
so that $y(0)=y$. Assume that for some $s>0$ one has:
\[
y(s) \notin \ax_{\lambda-\delta}^\alpha { (K) }
\]
Recall that $d_H( K,  K') < \epsilon $ implies that 
\begin{align} 
\left| R_{K'}(x) - R_{K}(x) \right| < \epsilon,
\label{eq:RKRKprimeBound}
\end{align}  
for all $x$.
Similarly to \eqref{equation:R_KOfALowerBounded}, once again has $ y  \in \ax_\lambda^\alpha(K')  \Rightarrow R_{K'}(y) -\alpha > \lambda \Rightarrow R_{K}(y) - \alpha > \lambda -\epsilon$. 
Moreover due to the hypothesis of the lemma one has $\epsilon < \delta$, so one sees, 
\begin{equation*}\tag{\ref{equation:RYMinusAlphaGreaterThanLambdaMinusDelta}} 
R_K(y) -\alpha > \lambda - \delta.
\end{equation*}
As in the proof of Lemma \ref{lemma:EpsilonNeighborFlowsToLambdaMinusDelta} 
we can  apply Lemma \ref{lemma:RadiusTightLowerBoundWhenTrajectoryNotInAXLambdaAlpha} with $\lambda $ replaced by $ \lambda - \delta$,
which gives us,
 \begin{equation}\label{equation:LowerBoundRK_2}
\left( R_K(y(s)) - \alpha \right)^2  \geq   (S_0+s)^2 +  (\lambda - \delta)^2,
 \end{equation}
where:
\begin{equation}\label{equation:DefinitionS0_2}
S_0^2 = \left(R_K(y)-\alpha \right)^2 - (\lambda - \delta)^2.
\end{equation}
Since $y \in \ax_\lambda^\alpha(K')$, $R_{K'}(y) -\alpha  \geq \lambda$ and there is  $S_y \geq 0$
such that
 \begin{equation}\label{equation:DefinitionSY}
S_y^2  = \left( R_{K'}(y) -\alpha \right)^2 - \lambda^2.
 \end{equation}
Again following the same steps as in the proof of Lemma \ref{lemma:EpsilonNeighborFlowsToLambdaMinusDelta}, we use Lemma 4.15 of \cite{LIEUTIERhomotopytype} to see that,
 \begin{equation}\label{equation:UpperBoundRK_2_1}
\left( R_{K'}(y(s)) -\alpha\right)^2 \leq     ( S_y + \|y(s)-y\|)^2   +    \lambda^2.
 \end{equation}
Because $y(s)$ is parametrized by arc length, we have $\|y(s)-y\| \leq s$, which together with \eqref{eq:RKRKprimeBound} yields,
 \begin{equation}\label{equation:UpperBoundRK_2_2}
\left( R_K(y(s)) - \epsilon -\alpha \right)^2 \leq     ( S_y + s )^2   +    \lambda^2.
 \end{equation}
% and 
% \begin{equation}\label{equation:UpperBoundOnSY}
%S_y^2  \leq  \left( R_{K}(y) + \epsilon - \alpha \right)^2 - \lambda^2
% \end{equation}
Subtracting \eqref{equation:LowerBoundRK_2} from \eqref{equation:UpperBoundRK_2_2} yields,
\begin{align*}
 0  \leq  &( S_y + s )^2   - (S_0+s)^2   +    \lambda^2 -  (\lambda - \delta)^2 +  \left( R_K(y(s))  -\alpha \right)^2 
\\  & -\left( R_K(y(s)) - \epsilon -\alpha \right)^2,
\end{align*}
which can be rewritten (in two steps) as,
\begin{align}
 0  \leq  ( S_y + s )^2   - (S_0+s)^2   +    \left(2 \lambda - \delta \right) \delta + \left( 2 R_K(y(s))  - 2 \alpha  - \epsilon \right) \, \epsilon
\nonumber
\\
2 \left(S_0 - S_y \right) s \leq \left(2 \lambda - \delta \right) \delta - \left(S_0^2  - S_y^2 \right) + 2 R_{\max} \, \epsilon.
\label{equation:FirstLinearInequationOnS}
\end{align}
Again as in the proof of Lemma \ref{lemma:EpsilonNeighborFlowsToLambdaMinusDelta}, combining 
\eqref{equation:DefinitionS0_2} and \eqref{equation:DefinitionSY} gives us
\begin{align}
\left| \left(2 \lambda - \delta \right) \delta - \left(S_0^2  - S_y^2 \right) \right| 
 &= \left|  \left( R_{K'}(y) -\alpha \right)^2 -  \left( R_{K}(y) -\alpha \right)^2 \right| 
\nonumber 
\\
&= \left( R_K(y) + R_{K'}(y) - 2 \alpha \right) |  R_K(y) - R_{K'}(y) | \nonumber 
\\
 &\leq \left( R_K(y) + R_{K'}(y) - 2 \alpha \right) \, \epsilon 
\tag{by \eqref{eq:RKRKprimeBound}}  
\\
 &\leq 2 R_{\max} \, \epsilon ,
\label{equation:R02MinusSY2CloseToTwoDeltaMinusDelatTimesDelta}
\end{align}
and \eqref{equation:FirstLinearInequationOnS} gives
\begin{equation}\label{equation:SecondtLinearInequationOnS}
 \left(S_0 - S_y \right) s \leq  2 R_{\max} \, \epsilon.
\end{equation}
Since $\epsilon < \frac{\left( 2 \lambda -\delta\right)  \delta}{8 R_{\max}}$ , 
\eqref{equation:R02MinusSY2CloseToTwoDeltaMinusDelatTimesDelta} yields,
\begin{equation*}
\frac{3}{4}  \left(2 \lambda - \delta \right) \delta < S_0^2  - S_y^2  < \frac{5}{4}  \left(2 \lambda - \delta \right) \delta
\end{equation*}
and
\begin{equation*}
S_0 - S_y = \frac{S_0^2  - S_y^2 }{ S_0 + S_{ y } } 
> \frac{3}{4}  \frac{ \left(2 \lambda - \delta \right)\delta }{ S_0 + S_{ y } } >   \frac{ \left(2 \lambda - \delta \right) \delta}{ 4 R_{\max}} , 
\end{equation*} 
 where we used that \eqref{equation:DefinitionS0_2} implies that $S_0 \leq R_{\max}$, and \eqref{equation:DefinitionSY}  implies $S_y \leq R_{\max}$. 
With \eqref{equation:SecondtLinearInequationOnS} we get:
\begin{equation*}
s  < \frac{8 R_{\max} }{\left(2 \lambda - \delta \right) \delta } \, \epsilon
\end{equation*}
and we conclude exactly as in the proof of Lemma \ref{lemma:EpsilonNeighborFlowsToLambdaMinusDelta} 
\end{proof}

%%%%%%%%%%%%%%%%%%%%%%%%%%%%%%%%%%%%
%{\color{blue} I have added a suggestion for a different section title. Now the statement does not really reflect the title. }

\begin{remark}
The previous lemma can be interpreted as a stability result with respect to the one-sided Hausdorff distance. Moreover, the lemma applies  when $\alpha =0$, in which case the statement can be compared to Theorem 3 of \cite{cl2005lambda}. 
Theorem 3 of \cite{cl2005lambda} says that the $\lambda$-medial axis is $\frac{1}{2}$-H{\"o}lder stable in the following sense: 
If $d_H(K,K') < \epsilon$, then for $x \in \ax_\lambda (K')$ there is $y\in ax_{\lambda - \delta} (K)$ 
with $\|y-x\| =O(\epsilon^{\frac{1}{2}})$. In other words the one sided Hausdorff distance between $\ax_\lambda (K')$ and $\ax_\lambda (K)$ is $O(\epsilon^{\frac{1}{2}})$.
 Lemma
 \ref{lemma:AXLambdaAlphaOfKPrimeFlowsToAXLambdaMinusDeltaAlphaOfK}
 proves the stronger linear bound $\|y-x\| =O(\epsilon)$.
The effect of a translation on $K$ shows that one cannot expect a bound better than linear.

The proofs of Theorem 3 in \cite{cl2005lambda} and Lemma \ref{lemma:AXLambdaAlphaOfKPrimeFlowsToAXLambdaMinusDeltaAlphaOfK}
are based on the same idea. However, here we get a better bound by using the
lower bound on $R_K(y(S))$ given by Lemma 
\ref{lemma:RadiusTightLowerBoundWhenTrajectoryNotInAXLambdaAlpha}
which is tighter than the one used in \cite{cl2005lambda}.
\end{remark}

\subsection{Hausdorff distance between  
%$\ax_\lambda^\alpha(K')$ 
\texorpdfstring{$\ax_\lambda^\alpha(K')$}{lambda alpha medial axis of K'} 
and %$\ax_\lambda^\alpha(K)$
\texorpdfstring{$\ax_\lambda^\alpha(K)$}{lambda alpha medial axis of K}
}

Combining Lemmas \ref{lemma:AXLambdaAlphaOfKPrimeFlowsToAXLambdaMinusDeltaAlphaOfK}
and \ref{lemma:BoundOnLengthToLambdaPlusDelta} allows to give a more symmetric statement.

%We have proven:

%%%%%%%%%%%%%%%%%%%%%%%%%%%%%%%%%%%%%
%%%%%%%%%%%%%%%%%%%%%%%%%%%%%%%%%%%%%
%%%%%%%%%%%%%%%%%%%%%%%%%%%%%%%%%%%%%
%%%%%%%%%%%%%%%%%%%%%%%%%%%%%%%%%%%%%
%%%%%%%%%%%%%%%%%%%%%%%%%%%%%%%%%%%%%

\begin{lemma}\label{lemma:AXLambdaHausdorffStableSymmetric}
Let $K, K'\subset \R^n$ be  complements of a bounded open  sets $K^c$ and $K'^c$ and 
$\alpha > 0$, $\mu >0$, $\lambda > 0$,
such that $r_\mu^\alpha (K)   > \alpha + \lambda$.
If
\[
d_H( K,  K') < \epsilon,
\] 
and, if  $\delta = 2 \sqrt{ \alpha \tilde{\mu} \epsilon}<  \lambda$, and $\epsilon < \min \left(  \frac{\left( 2 \lambda -\delta\right)  \delta}{8 R_{\max}} , \frac{\lambda^2}{16 \alpha \tilde{\mu}} \right)$ then,
\begin{align*}
\Phi_K \left( C \epsilon^{\frac{1}{2}}
 \, , \: \ax_\lambda^\alpha(K')  \right) %&
\subset \ax_\lambda^\alpha(K),
%%  \\
%%\Phi_{K'} \left(\frac{22}{3} \frac{R_{\max}^2 }{\alpha^{\frac{1}{2}}  \tilde{\mu}^{\frac{3}{2}} \lambda } \epsilon^{\frac{1}{2}}
 %%\, , \: \ax_\lambda^\alpha(K)  \right) &\subset& \ax_\lambda^\alpha(K')
\end{align*}
where  $R_{\max} = \max \{ R_{\max}(K) ,R_{\max} (K') \} < \infty$ %{\color{blue} See remarks in the previous section TODO MW}
, $\tilde{\mu} = \tilde{\mu}_{\mu,\lambda}^{\alpha, \alpha} (K) >0$,
 and $C$ is defined as
\begin{equation}\label{equation:DefinitionC}
C = \frac{22}{3} \frac{R_{\max}^2 }{\alpha^{\frac{1}{2}}  \tilde{\mu}^{\frac{3}{2}} \lambda } .
\end{equation}

Moreover, if one has also symmetrically  $r_\mu^\alpha (K)  ,   r_\mu^\alpha (K')  > \alpha + \lambda$  and $R_{\max} = \max \left( R_{\max}(K),  R_{\max}(K')\right)$,
and $\tilde{\mu} = \min \left( \tilde{\mu}_{\mu,\lambda}^{\alpha, \alpha} (K) , \tilde{\mu}_{\mu,\lambda}^{\alpha, \alpha} (K')  \right)$
 we have
\begin{equation}\label{equation:HausdorffDistanceBetweenLambdaAlphaMA}
d_H \left(  \ax_\lambda^\alpha(K) ,  \ax_\lambda^\alpha(K') \right) < C \epsilon^{\frac{1}{2}}. 
\end{equation}
\end{lemma}
%%%%%%%%%%%%%%%%%%%%%%%%%%%%%%%%%%%%%
%%%%%%%%%%%%%%%%%%%%%%%%%%%%%%%%%%%%%

\begin{proof}%[of Lemma \ref{lemma:AXLambdaHausdorffStableSymmetric}] 
{
%More precisely, w
We first flow $\ax_\lambda^\alpha(K')$ into $\ax_{\lambda-\delta}^\alpha(K)$, and then we flow from $\ax_{\lambda-\delta}^\alpha(K)$ to $\ax_{\lambda}^\alpha(K)$.   } 
Indeed these Lemmas \ref{lemma:AXLambdaAlphaOfKPrimeFlowsToAXLambdaMinusDeltaAlphaOfK}
and \ref{lemma:BoundOnLengthToLambdaPlusDelta} give that if  $K, K'\subset \R^n$ 
are  complements of a bounded open  sets $K^c$ and $K'^c$ and 
$\alpha >0 $, $\mu >0$, $\lambda > 0$,
are such that $r_\mu^\alpha (K)   > \alpha + \lambda$ and 
$d_H( K,  K') < \epsilon$ then for $0< \delta < \lambda$
and  $\epsilon <  \frac{\left( 2 \lambda -\delta\right)  \delta}{8 R_{\max}}$ one has,
\begin{equation}\label{equation:PhiKSendAXLambdaOfKPrimeInAXLambdaOfK}
\Phi_K \left(  \frac{R_{\max}^2}{\alpha (\lambda-\delta)  \tilde{\mu}^2} \delta
     +  \frac{8 R_{\max}^2}{\left( 2 \lambda -\delta\right)  \delta \tilde{\mu}}  \: 
      \epsilon \, , \: \ax_\lambda^\alpha(K')  \right) \subset \ax_{\lambda}^\alpha(K),
\end{equation}
where we use \eqref{equation:FlowAsSemiGroup} and $R_{\max} = R_{\max}(K)< \infty$ and $\tilde{\mu} = \tilde{\mu}_{\mu,\lambda}^{\alpha, \alpha}>0$. 

%\begin{equation}
%\begin{cases}
%\Phi_K \left(  \frac{R_{\max}^2}{\alpha (\lambda-\delta)  \tilde{\mu}^2} \delta
%     +  \frac{8 R_{\max}^2}{\left( 2 \lambda -\delta\right)  \delta \tilde{\mu}}  \: 
%      \epsilon \, , \: \ax_\lambda^\alpha(K')  \right) \subset \ax_{\lambda}^\alpha(K)  \\
%      \Phi_{K'} \left(  \frac{R_{\max}^2}{\alpha (\lambda-\delta)  \tilde{\mu}^2} \delta
%     +  \frac{8 R_{\max}^2}{\left( 2 \lambda -\delta\right)  \delta \tilde{\mu}}  \: 
%      \epsilon \, , \: \ax_\lambda^\alpha(K)  \right) \subset \ax_{\lambda}^\alpha(K') 
%      \end{cases}
%\end{equation}
Here we still have to choose a value of $\delta$ that minimizes
the first argument of $\Phi_K$ 
in \eqref{equation:PhiKSendAXLambdaOfKPrimeInAXLambdaOfK}, that is, 
\begin{equation}\label{equation:TimeToMinimizeInFlowFromAXLambdaKPrimeToAXLambdaKP}
\delta \mapsto \frac{R_{\max}^2}{\alpha (\lambda-\delta)  \tilde{\mu}^2} \delta
     +  \frac{8 R_{\max}^2}{\left( 2 \lambda -\delta\right)  \delta \tilde{\mu}}  \: 
      \epsilon.
\end{equation}
To this end we first observe that when $\epsilon$ is small, 
the value of $\delta$ that minimizes, \eqref{equation:TimeToMinimizeInFlowFromAXLambdaKPrimeToAXLambdaKP}
is small.
When $\delta$ is small 
\eqref{equation:TimeToMinimizeInFlowFromAXLambdaKPrimeToAXLambdaKP}
is well approximated by
\begin{equation}\label{equation:ApproxTimeToMinimizeInFlowFromAXLambdaKPrimeToAXLambdaKP}
\delta \mapsto \frac{R_{\max}^2}{\alpha \lambda  \tilde{\mu}^2} \delta
     +  \frac{4 R_{\max}^2}{ \lambda  \delta \tilde{\mu}}  \: 
      \epsilon.
\end{equation}
The value of $\delta$ that minimizes \eqref{equation:ApproxTimeToMinimizeInFlowFromAXLambdaKPrimeToAXLambdaKP}
is
\begin{equation}\label{equation:AChoiceForDelta}
\delta = 2 \sqrt{ \alpha \tilde{\mu} \epsilon}.
\end{equation}
We now substitute \eqref{equation:AChoiceForDelta} in \eqref{equation:TimeToMinimizeInFlowFromAXLambdaKPrimeToAXLambdaKP}. If we also observe that if $\epsilon < \frac{\lambda^2}{16 \alpha \tilde{\mu}}$, \eqref{equation:AChoiceForDelta} gives $\delta < \lambda/2$ and thus $2 \lambda -\delta > 3/2 \lambda$ and  $\lambda-\delta > 1/2 \lambda$ we find the following upper bound
%If $\epsilon < \frac{\lambda^2}{16 \alpha \tilde{\mu}}$, \eqref{equation:AChoiceForDelta} gives $\delta < \lambda/2$ and then,
%substituting \eqref{equation:AChoiceForDelta} in \eqref{equation:TimeToMinimizeInFlowFromAXLambdaKPrimeToAXLambdaKP}
%gives, using $2 \lambda -\delta > 3/2 \lambda$ and  $\lambda-\delta > 1/2 \lambda$, the following upper bound
\[
\frac{R_{\max}^2}{\alpha (\lambda-\delta)  \tilde{\mu}^2} \delta
     +  \frac{8 R_{\max}^2}{\left( 2 \lambda -\delta\right)  \delta \tilde{\mu}}  \: 
      \epsilon
      < \frac{22}{3} \frac{R_{\max}^2 }{\alpha^{\frac{1}{2}}  \tilde{\mu}^{\frac{3}{2}} \lambda } \epsilon^{\frac{1}{2}}.
\]
\end{proof} 

%%%%%%%%%%%%%%%%%%%%%%%%%%%%%%%%%%%%%
%%%%%%%%%%%%%%%%%%%%%%%%%%%%%%%%%%%%%
%%%%%%%%%%%%%%%%%%%%%%%%%%%%%%%%%%%%%
%%%%%%%%%%%%%%%%%%%%%%%%%%%%%%%%%%%%%

\begin{remark} The symmetric condition $r_\mu^\alpha (K)  , r_\mu^\alpha (K')  > \alpha + \lambda$ can be replaced by a condition on $K$ only. More precisely, 
in the limit where some $\delta$ tends to zero as $d_H( K,  K') \rightarrow 0$ we have that $r_\mu^{\alpha - \delta} (K)  >   \alpha + \lambda + \delta$ implies $r_\mu^\alpha (K)  , r_\mu^\alpha (K')  > \alpha + \lambda$ , see \cite[Theorem 3.4]{chazal2009sampling} where also the dependencies of $\delta$ and $d_H( K,  K')$ are made precise.
\end{remark}

%Remark: the symmetric condition $r_\mu^\alpha (K)  , r_\mu^\alpha (K')  > \alpha + \lambda$ can be replaced by 
%a similar condition on  $K$ only as it
%is  implied by a condition in the form $r_\mu^{\alpha - \delta} (K)  >   \alpha + \lambda + \delta$
%for some $\delta$ that tends to zero as $d_H( K,  K') \rightarrow 0$ \cite{chazal2009sampling}
%{\color{red} or rather should we better explain this remark  in previous work section ??} {\color{blue} I would move the stuff from the previous section here. Some more details would help}

%%%%%%%%%%%%%%%%%%%%%%%%%%%

\begin{lemma}\label{lemma:EpsilonNeighborFlowsToLambda}
Let $K\subset \R^n$ be the complement of a bounded open  set $K^c$ and 
$\alpha > 0$, $\mu >0$ and $\lambda > 0$,
such that $r_\mu^\alpha (K)  > \alpha + \lambda$.
Then,
if  $\delta = 2 \sqrt{ \alpha \tilde{\mu} \epsilon}<   \lambda$ and $\epsilon < \min \left(2 \alpha,   \frac{\left( 2 \lambda -\delta\right)  \delta}{8 R_{\max}} , \frac{\lambda^2}{16 \alpha \tilde{\mu}} \right)$ one has
\[
\Phi_K \left(C \epsilon^{\frac{1}{2}}
 \, , \:  \left(\ax_\lambda^\alpha(K) \right)^{\oplus \epsilon}  \right) \subset \ax_\lambda^\alpha(K),
\]
where  $R_{\max} = R_{\max}(K)< \infty$, $\tilde{\mu} = \tilde{\mu}_{\mu,\lambda}^{\alpha, \alpha} (K)>0$, and $C$ is defined by \eqref{equation:DefinitionC}.
\end{lemma}

%%%%%%%%%%%%%%%%%%%%%%%%%%%%%%%%%%%%
%%%%%%%%%%%%%%%%%%%%%%%%%%%%%%%%%%%%

\begin{proof}%[of Lemma \ref{lemma:EpsilonNeighborFlowsToLambda}] 
The proof is identical to the proof of Lemma \ref{lemma:AXLambdaHausdorffStableSymmetric}, however this 
 this time we combine Lemmas \ref{lemma:EpsilonNeighborFlowsToLambdaMinusDelta}
and \ref{lemma:BoundOnLengthToLambdaPlusDelta} to achieve the more symmetric statement. We note that we choose the same optimal value for $\delta$ from \eqref{equation:AChoiceForDelta}. 
\end{proof}

%%%%%%%%%%%%%%%%%%%%%%%%%%%%%%%%%%%%
%%%%%%%%%%%%%%%%%%%%%%%%%%%%%%%%%%%%
%%%%%%%%%%%%%%%%%%%%%%%%%%%%%%%%%%%%

%%%%%%%%%%%%%%%%%%%%%%%%%%%%%%%%%%%%
\section{Gromov-Hausdorff stability of the \texorpdfstring{$(\lambda,\alpha)$}{(lambda, alpha)}-medial 
axis under Hausdorff perturbation of \texorpdfstring{$K$}{K}} \label{section:GromovHausdorffStability}

\begin{figure}[h!]
\centering
\includegraphics[width=0.48\textwidth]{%pictures/
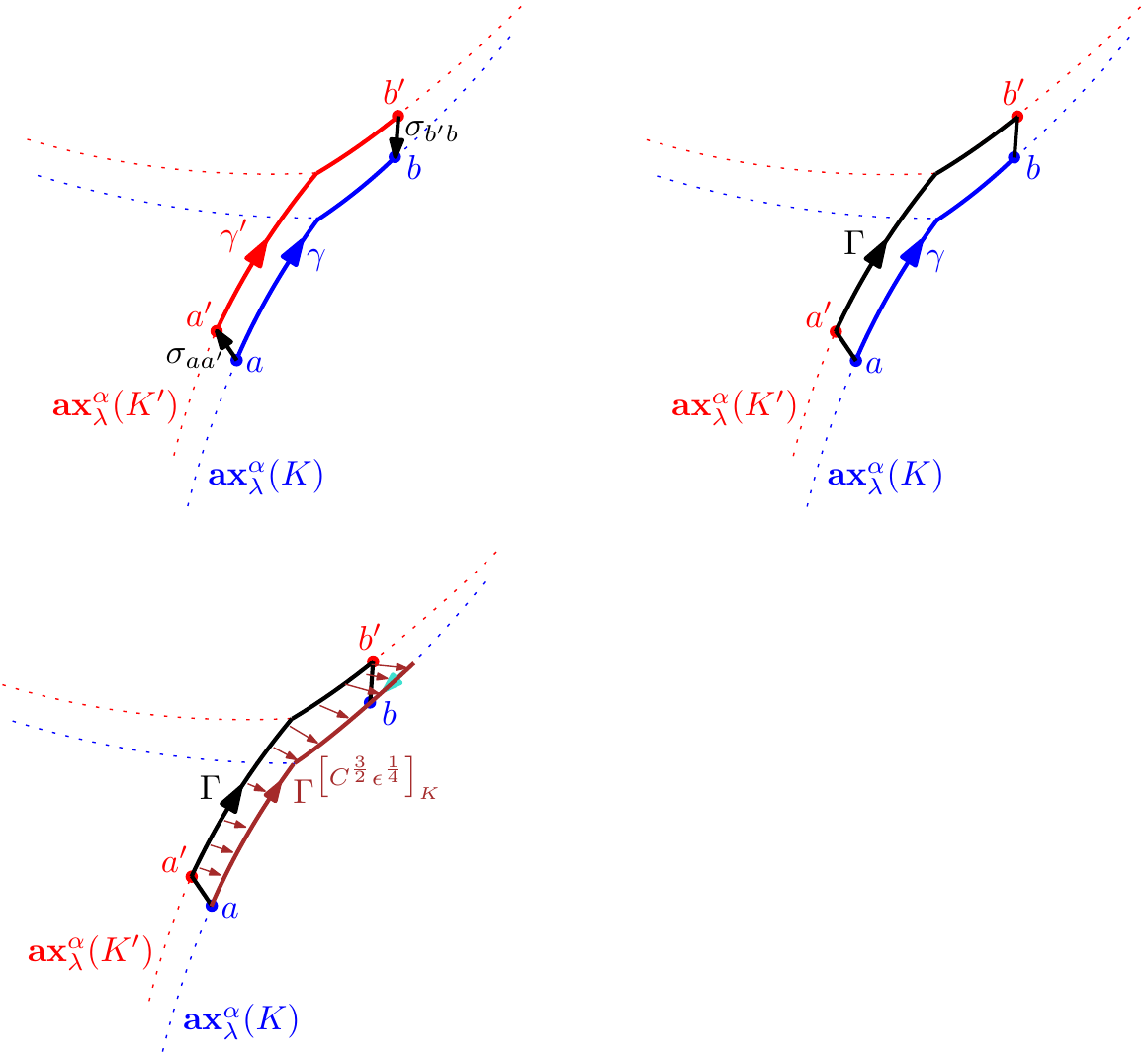} 
\caption{Illustration of the proof of Theorem \ref{theorem:GromovHausdorffStability}. The path $\Gamma$,
concatenation of  $\sigma_{aa'}$, $\gamma'$ and $\sigma_{b'b}$, is pushed through the flow $\Phi_K$ into a path  $\PushedPath{\Gamma}{C^\frac{3}{2} \epsilon^\frac{1}{4}}$
 that realizes a path in $\ax_\lambda^\alpha(K)$ between $a$ and $b$.}
\label{figure:PathPushedGamma}
\Description{The pushing of a path on one $(\lambda,\alpha)$-medial axis onto another $(\lambda,\alpha)$-medial axis.}
\end{figure}

In this section, we bound the Gromov-Hausdorff distance between the $(\lambda,\alpha)$-medial axis of a set $K$ and the medial axis of perturbation of the set, where the perturbation is small in the Hausdorff sense. 
Gromov-Hausdorff distance is understood to be with respect to the intrinsic distance on the medial axis, that is the metric on the space $\ax_{\lambda}^\alpha (K)$ 
is the metric induced by the geodesic distance on the set. As we have seen in Section \ref{sec:finiteDiam} this metric is well defined. 

{We assume that we are in the symmetric setting, that is the conditions of Lemma \ref{lemma:AXLambdaHausdorffStableSymmetric} are satisfied.} 
We assume moreover that $K^c$ and $(K')^c$ are connected.

Figure   \ref{figure:PathPushedGamma} illustrates the idea of the proof of Theorem \ref{theorem:GromovHausdorffStability}. % (in  appendix \ref{section:proofs}).
{Consider } the pairs 
\[(a,a'), (b,b') \in \mathcal{R}= \left\{ (x, x') \in \ax_\lambda^\alpha(K) \times \ax_\lambda^\alpha(K'), \|x - x'\| <  C \epsilon^{\frac{1}{2}} \right\}.\] 
In order  to compare the length of a geodesic  $\gamma$ from $a$ to $b$ 
to the length of a geodesic  $\gamma'$ from $a'$ to $b'$  (left), we first create a path $\Gamma$ (middle), concatenation of  $\gamma$ with two straight segments $\sigma_{aa'}$ and $\sigma_{b' b}$.
Then $\Gamma$ is ``pushed'' (right) along the flow $\Phi_K$, which, after a ``time'' $t=\mathcal{O}\big(\epsilon^{\frac{1}{4}}\big)$, belongs to $\ax_\lambda^\alpha (K)$. The  pushed path can then be shown to be not much longer than the path $\gamma'$.

% We have proven, without forgetting the bound  \eqref{equation:complementaryConditionsOnEpsilon} on $\epsilon$:

\begin{theorem}\label{theorem:GromovHausdorffStability}
Let $K, K' \subset \R^n$ be  complements of  bounded open  sets $K^c,(K')^c$,
$\alpha > 0$, $\mu >0$ and $\lambda > 0$,
such that, for some $\alpha' < \alpha$ one has $r_\mu^{\alpha'} (K) > \alpha + \lambda $ and
$r_\mu^{\alpha'} (K') > \alpha + \lambda $,
and  $\left(K^{\oplus \alpha}\right)^c$ and  $\left(K'^{\oplus \alpha}\right)^c$ are connected.
Denote $R_{\max} = \max \left( R_{\max}(K),  R_{\max}(K')\right)$,
and $\tilde{\mu} = \min \left( \tilde{\mu}_{\mu,\lambda}^{\alpha, \alpha} (K) , \tilde{\mu}_{\mu,\lambda}^{\alpha, \alpha} (K')  \right)$.

Assume that  $d_H( K,  K') < \epsilon$. 
%Then, 
{ If }
  \[
  \epsilon < \min \left( \frac{  \lambda^2 \alpha \tilde{\mu}}{16 R_{\max}^2} , \frac{\lambda^2}{16 \alpha \tilde{\mu}},\frac{9 \lambda^4 \alpha \tilde{\mu}^3}{400 R_{\max}^4},\left( \frac{ 2 \alpha} {C }\right)^2, \left(\frac{\lambda^2 \alpha \tilde{\mu}}{16 R_{\max}^2 C}\right)^2, \left( \frac{\lambda^2}{16 \alpha \tilde{\mu} C}  \right)^2 \right),
  \]
then the Gromov-Hausdorff  distance between $\ax_\lambda^\alpha(K)$ and $\ax_\lambda^\alpha(K')$ { with respect to the intrinsic metric}
is upper bounded by
\[
2  C^{\frac{3}{2}} \epsilon^\frac{1}{4}   +  2 C \epsilon^{\frac{1}{2}}    e^\frac{ C^\frac{3}{2} \epsilon^\frac{1}{4} }{\alpha} 
+ D \left( e^\frac{ C^\frac{3}{2} \epsilon^\frac{1}{4} }{\alpha}  - 1 \right) =O\left(\epsilon^\frac{1}{4} \right),
\]
 where 
\[
C = \frac{22}{3} \frac{R_{\max}^2 }{\alpha^{\frac{1}{2}}  \tilde{\mu}^{\frac{3}{2}} \lambda } 
\]
 and
\begin{align}
D &= \max \left(  \Gdiam( \ax_\lambda^\alpha(K)),  \Gdiam( \ax_\lambda^\alpha(K')) \right) 
\nonumber
%\\
%& \leq 2 \frac{R_{\max}}{\tilde{\mu}^2} +  \alpha  \left(\frac{ \max\left( \vol(K^c),\vol(K'^c) \right) }{V_N  \, \left( \frac{\alpha}{2}\right)^n}+1\right)  e^{\frac{2 R_{\max}}{\alpha \tilde{\mu}^2}} 
%%\nonumber
%%\\
%%&
<\infty.
\nonumber
\end{align}
\end{theorem}

\begin{proof} %[of Theorem \ref{theorem:GromovHausdorffStability}]  
In order to lower bound the Gromov-Hausdorff distance between $\ax_\lambda^\alpha(K)$
and $\ax_\lambda^\alpha(K')$
under the assumptions of Lemma  \ref{lemma:AXLambdaHausdorffStableSymmetric} we use Definition \ref{definition:GromovHausdorffDistance}.

Consider the  relation $\mathcal{R} \subset \ax_\lambda^\alpha(K) \times \ax_\lambda^\alpha(K')$ defined by:
\begin{equation}\label{equation:RelationForGromovHausdorff}
\mathcal{R}  = \left\{ (x, x') \in \ax_\lambda^\alpha(K) \times \ax_\lambda^\alpha(K'), \|x - x'\| <  C \epsilon^{\frac{1}{2}} \right\},
\end{equation} 
where $C$ is defined by \eqref{equation:DefinitionC}.

We know from Lemma \ref{lemma:AXLambdaHausdorffStableSymmetric} that this relation is surjective, which means that for any $x \in \ax_\lambda^\alpha(K)$
there is $x' \in \ax_\lambda^\alpha(K')$ such that $(x,x') \in \mathcal{R}$ and, reciprocally, 
if $x' \in \ax_\lambda^\alpha(K')$
there is $x \in \ax_\lambda^\alpha(K)$ such that $(x,x') \in \mathcal{R}$.

Consider $(a,a'), (b,b') \in \mathcal{R}$.
Since $K^c$ and $(K')^c$ are connected,  Theorem \ref{theorem:AXLambdaAlphaIsGeodesicAndFiniteDiameter} yields
that there are paths $\gamma : [0,L] \rightarrow  \ax_\lambda^\alpha(K)$ 
and $\gamma' : [0,L'] \rightarrow  \ax_\lambda^\alpha(K')$, parametrized by arc length such that:
\[
\gamma(0)=a, \: \gamma(L)=b,  \:  \gamma'(0)=a',  \:   \gamma'(L')=b'
\]
and  where $L$ and $L'$ are the respective
geodesic distances in $\ax_\lambda^\alpha(K)$ and  $\ax_\lambda^\alpha(K')$ between $a,b$ and $a',b'$, that is,
\begin{align}
L = d_{\ax_\lambda^\alpha(K)} (a,b) \quad \text{and} \quad L' = d_{\ax_\lambda^\alpha(K')} (a',b') . 
\nonumber
\end{align} 
Denote  by $\sigma_{aa'}$ and $\sigma_{b'b}$ the linear path from $a$ to $a'$ and from
$b'$ to $b$, respectively. By definition of the relation $\mathcal{R}$, we have that their lengths are upper bounded by $C \epsilon^{\frac{1}{2}} $
and, since $R_K(a), R_K(b) \geq \lambda + \alpha$, we have the inclusions 
$\Image(\sigma_{aa'}), \Image(\sigma_{b'b}) \subset \left( K^{\oplus \lambda + \alpha - C \epsilon^{\frac{1}{2}} }\right)^c$.

From now on we assume:
\begin{equation}\label{equation:FourthBoundOnEpsilon}
\epsilon < \frac{\lambda^2}{C^2} = \frac{9 \lambda^4 \alpha \tilde{\mu}^3}{400 R_{\max}^4}
\end{equation}
then $C \epsilon^{\frac{1}{2}} < \lambda$ and it follows that  $\Image(\sigma_{aa'}), \Image(\sigma_{b'b})\subset \left( K^{\oplus \alpha}\right)^c$.
Also $\Image(\gamma') \subset \ax_\lambda^\alpha(K') \subset \left( K'^{\oplus \lambda + \alpha }\right)^c$.
Because $\alpha, \lambda  < R_{\max}$ and $\tilde{\mu} \leq 1$, \eqref{equation:FourthBoundOnEpsilon} implies $\epsilon < \lambda$, and $d_H(K,K') < \epsilon$ gives
$\left( K'^{\oplus \lambda + \alpha }\right)^c \subset \left( K^{\oplus \alpha }\right)^c$ and we get $\Image(\gamma') \subset \left( K^{\oplus \alpha}\right)^c$.

From \eqref{equation:HausdorffDistanceBetweenLambdaAlphaMA} 
one has as well 
$\Image(\gamma')\subset  \ax_\lambda^\alpha (K') \subset    \left(\ax_\lambda^\alpha (K) \right)^{\oplus C \epsilon^{\frac{1}{2}}}$ and
$a,b \in \ax_\lambda^\alpha (K)$ with $(a,a'), (b,b') \in \mathcal{R}$ gives as well \newline
$\Image(\sigma_{aa'}), \Image(\sigma_{b'b})\subset   \ax_\lambda^\alpha \left( K\right)^{\oplus C \epsilon^{\frac{1}{2}}}$.

Now consider the path $\Gamma$ from $a$ to $b$, which we define as the concatenation of $\sigma_{aa'}$, $\gamma'$ and $\sigma_{b'b}$.
One has,
\begin{equation}\label{equation:LengthGAMMA}
\length \left(\Image(\Gamma) \right) < L' + 2 C \epsilon^{\frac{1}{2}} .
\end{equation}
$\Gamma$  is included in $\left( K^{\oplus \alpha}\right)^c$ and $\left(\ax_\lambda^\alpha (K)\right)^{\oplus C \epsilon^{\frac{1}{2}}}$, that is, 
\begin{equation}\label{equation:GammaIncludedInTubularNeihghborAXLambdaAlpha}
\Image(\Gamma) \subset \left( K^{\oplus \alpha}\right)^c  \cap \left(\ax_\lambda^\alpha (K)\right)^{\oplus C \epsilon^{\frac{1}{2}}}. 
\end{equation}
For $T \geq 0$, we consider now the path $\PushedPath{\Gamma}{T}$,
% {\color{blue}not sure this is very insightful especially as T can be quite a large expression (the breackets do make things better)}
 connecting $a$ to $b$
according to Definition \ref{definition:PushedPaths}.

%, concatenation of:
%\begin{itemize}
%\item the path $t \mapsto \Phi_K(t, a)$ for $t \in [0, T]$
%\item the path $t \mapsto \Phi_K(T, \Gamma(t) )$ for $t \in [0,   \length \left(\Image(\Gamma) \right)]$
%\item the path $t \mapsto \Phi_K(T-t, b)$ for $t \in [0, T]$
%\end{itemize}
Intuitively, the path $\PushedPath{\Gamma}{T}$ can be visualized as the path $\Gamma$ ``pushed'' during a time $T$ by the flow $\Phi_K$
while ``holding'' the end points $a$ and $b$ as shown in Figure \ref{figure:PathPushedGamma}.

Using \eqref{equation:GammaIncludedInTubularNeihghborAXLambdaAlpha},
if $\delta' = 2 \sqrt{ \alpha \tilde{\mu} C \epsilon^{\frac{1}{2}}}<   \lambda$ and 
\[ C \epsilon^{\frac{1}{2}}  < \min \left(2 \alpha,   \frac{\left( 2 \lambda -\delta'\right)  \delta'}{8 R_{\max}} , \frac{\lambda^2}{16 \alpha \tilde{\mu}} \right)\] 
we can apply Lemma \ref{lemma:EpsilonNeighborFlowsToLambda} with $\epsilon $ replaced by $C \epsilon^\frac{1}{2}$.

Note that the condition:
$\delta' = 2 \sqrt{ \alpha \tilde{\mu} C \epsilon^{\frac{1}{2}}}<   \lambda$ and 
 \newline $C \epsilon^{\frac{1}{2}}  < \min \left(2 \alpha,   \frac{\left( 2 \lambda -\delta'\right)  \delta'}{8 R_{\max}} , \frac{\lambda^2}{16 \alpha \tilde{\mu}} \right)$
 is implied by
 \begin{equation}\label{equation:complementaryConditionsOnEpsilon}
 \epsilon < \min \left(  \left( \frac{ 2 \alpha} {C }\right)^2, \left(\frac{\lambda^2 \alpha \tilde{\mu}}{16 R_{\max}^2 C}\right)^2, \left( \frac{\lambda^2}{16 \alpha \tilde{\mu} C}  \right)^2 \right),
 \end{equation}
so that Lemma  \ref{lemma:EpsilonNeighborFlowsToLambda} gives us:
\[
\Phi_K \left( C  \left(C \epsilon^\frac{1}{2} \right)^\frac{1}{2} , 
\Image(\Gamma) \right) \subset \ax_\lambda^\alpha (K),
\]
that is, together with  \eqref{equation:AXLambdaAlphaStableUnderFlow}, %Lemma \ref{lemma:FIsINotDecreasing},
\[
\Image \left( \PushedPath{\Gamma}{C^\frac{3}{2} \epsilon^\frac{1}{4}} \right)  \subset \ax_\lambda^\alpha (K).
\]
Since $\PushedPath{\Gamma}{C^\frac{3}{2} \epsilon^\frac{1}{4}} $ is a path from $a \in \ax_\lambda^\alpha (K)$
to $b \in \ax_\lambda^\alpha (K)$ inside 
$\ax_\lambda^\alpha (K)$, one has
\[
\length \left(  \PushedPath{\Gamma}{C^\frac{3}{2} \epsilon^\frac{1}{4}}  \right) \geq L = d_{\ax_\lambda^\alpha (K)} (a,b).
\]
Because by \eqref{equation:GammaIncludedInTubularNeihghborAXLambdaAlpha}
$\Gamma$ lies in $\left( K^{\oplus \alpha}\right)^c$, 
 Lemma \ref{lemma:PushedPathLengthExpansion} can be applied which gives, using \eqref{equation:LengthGAMMA}:
 
\begin{align*}
\length \left( \PushedPath{\Gamma}{C^\frac{3}{2} \epsilon^\frac{1}{4}} \right) &\leq  2  C^{\frac{3}{2}} \epsilon^\frac{1}{4}   + \left( L' + 2 C \epsilon^{\frac{1}{2}}  \right) e^\frac{ C^\frac{3}{2} \epsilon^\frac{1}{4} }{\alpha} \\
&=  L' + 2  C^{\frac{3}{2}} \epsilon^\frac{1}{4}   + 2 C \epsilon^{\frac{1}{2}} e^\frac{ C^\frac{3}{2} \epsilon^\frac{1}{4} }{\alpha} 
+ L'  \left( e^\frac{ C^\frac{3}{2} \epsilon^\frac{1}{4} }{\alpha}  - 1 \right).
\end{align*}
Since,  by Theorem \ref{theorem:AXLambdaAlphaIsGeodesicAndFiniteDiameter}, $L' \leq  \Gdiam( \ax_\lambda^\alpha(K')) < \infty$, one has
\begin{equation*}
L \leq  L' + 2  C^{\frac{3}{2}} \epsilon^\frac{1}{4}   +  2 C \epsilon^{\frac{1}{2}}    e^\frac{ C^\frac{3}{2} \epsilon^\frac{1}{4} }{\alpha} + \Gdiam( \ax_\lambda^\alpha(K')) \left( e^\frac{ C^\frac{3}{2} \epsilon^\frac{1}{4} }{\alpha}  - 1 \right).
\end{equation*}
Because we make, symmetrically, the same assumptions on $K$ and $K'$ one has,
\begin{align*}
(a,a'), (b,b') \in \mathcal{R} \Rightarrow&\left| d_{\ax_\lambda^\alpha(K)} (a,b) -  d_{\ax_\lambda^\alpha(K')} (a',b') \right| \\
& \leq 2  C^{\frac{3}{2}} \epsilon^\frac{1}{4}   +  2 C \epsilon^{\frac{1}{2}}    e^\frac{ C^\frac{3}{2} \epsilon^\frac{1}{4} }{\alpha} + D\left( e^\frac{ C^\frac{3}{2} \epsilon^\frac{1}{4} }{\alpha}  - 1 \right),
\end{align*}
where
\[
D = \max \left(  \Gdiam( \ax_\lambda^\alpha(K)),  \Gdiam( \ax_\lambda^\alpha(K')) \right).
\]
By Corollary \ref{corollary:AXLambdaAlphaIsGeodesicAndFiniteDiameter} (using the assumptions $r_\mu^{\alpha'} (K) > \alpha + \lambda $ and
$r_\mu^{\alpha'} (K') > \alpha + \lambda $), we know that $D< \infty$.
\end{proof}

\begin{acks}
We are greatly indebted to Erin Chambers for posing a number of questions that eventually led to this paper. We would also like to thank the other organizers of the workshop on `Algorithms for the medial axis'.  We are also indebted to Tatiana Ezubova for helping with the search for and translation of Russian literature.
The second author thanks all members of the Edelsbrunner and Datashape groups for the atmosphere in which the research was conducted. 

The research leading to these results has received funding from the European Research Council (ERC) under  the  European  Union's  Seventh  Framework  Programme (FP/2007-2013)  /  ERC  Grant Agreement No. 339025 GUDHI (Algorithmic Foundations of Geometry Understanding in Higher Dimensions).
Supported by the European Union's Horizon 2020 research and innovation programme under the Marie Sk{\l}odowska-Curie grant agreement No. 754411. The Austrian science fund (FWF) M-3073.
\end{acks}

%%%%%%%%%%%%%%%%%%%%%%%%%%%%%%%%%%%%%%%%%
%%%%%%%%%%%%%%%%%%%%%%%%%%%%%%%%%%%%%%%%%
%%%%%%%%%%%%%%%%%%%%%%%%%%%%%%%%%%%%%%%%%
%%%%%%%%%%%%%%%%%%%%%%%%%%%%%%%%%%%%%%%%%
%%%%%%%%%%%%%%%%%%%%%%%%%%%%%%%%%%%%%%%%%
%%%%%%%%%%%%%%%%%%%%%%%%%%%%%%%%%%%%%%%%%

\balance
\bibliographystyle{ACM-Reference-Format}
\bibliography{geomrefs}

\end{document}